\documentclass[onecolumn,12pt]{IEEEtran} 
\usepackage{longtable}
\usepackage{multirow}
\usepackage{tabularx}
\usepackage{lipsum}
\usepackage{multicol}
\usepackage{amssymb}
\usepackage{amsmath}
\usepackage{longtable}
\newtheorem{theorem}{Theorem}[section]
\newtheorem{lemma}[theorem]{Lemma}
\newtheorem{proposition}[theorem]{Proposition}
\newtheorem{corollary}[theorem]{Corollary}
\newtheorem{example}{Example}
\newtheorem{remark}{Remark}

\usepackage{xcolor}

\newenvironment{proof}{\begin{IEEEproof}}{\end{IEEEproof}}

\makeatletter

\newcommand{\cc}{2}
\newcommand{\GF}[2][\cc]{{\mathbb F}_{{#1}^{#2}}}

\newcommand{\GFm}{\@ifstar{\GF m^\star}{\GF m}}

\newcommand{\wt}{\mathrm{wt}}
\newcommand{\FF}{\mathbb{F}}

\newcommand{\tr}{\mathrm{Tr}}

\newcommand{\Tr}[2][1]{\tr_{#1}^{#2}}
\newcommand{\qm}{\cc^m}

\newcommand{\bF}{{\mathbb{F}}}
\newcommand{\cP}{{\mathcal{P}}}
\newcommand{\cB}{{\mathcal{B}}}

\newcommand{\C}{{\mathsf{C}}}

\newcommand{\bD}{{\mathbb{D}}}
\newcommand{\bc}{{\mathbf{c}}}
\newcommand{\support}{{\mathrm{suppt}}}

\newcommand{\Rmnum}[1]{\expandafter\@slowromancap\romannumeral #1@}

\newcommand{\s}[1]{\,#1\,}
\makeatother

\newtheorem{orp}{Open  problem}

\ifCLASSINFOpdf
\else
\fi
\begin{document}
%
\title{Cyclic bent functions and their applications in codes, codebooks, designs, MUBs and sequences}

\author{Cunsheng Ding, \and Sihem Mesnager, \and Chunming Tang, \and Maosheng Xiong\thanks{This work was supported by the National Natural Science Foundation of China
(Grant No. 11871058). C. Ding was supported by The Hong Kong Grants Council, Project No.
16300415. C. Tang also acknowledges support from
14E013, CXTD2014-4 and the Meritocracy Research Funds of China West Normal University.
M. Xiong was supported by The Hong Kong Grants Council, Project No. N\_HKUST619/17.
}

\thanks{C. Ding is with the Department of Computer Science and Engineering,
The Hong Kong University of Science and Technology, Clear Water Bay,
Kowloon, Hong Kong, China (email: cding@ust.hk).}

\thanks{S. Mesnager is with the Department of Mathematics, University of Paris
VIII, 93526 Saint-Denis, France, with LAGA UMR 7539,  CNRS, Sorbonne Paris
Cit\'e, University of Paris XIII, 93430 Paris,  France, and also with Telecom
ParisTech, 75013 Paris, France (e-mail: smesnager@univ-paris8.fr).}

\thanks{C. Tang is with the School of Mathematics  and Information,
China West Normal University, Nanchong 637002,  China, and also with the Department of Mathematics, The Hong Kong University
of Science and Technology, Clear Water Bay, Kowloon, Hong Kong (e-mail:
tangchunmingmath@163.com).
}
\thanks{M. Xiong is with the Department of Mathematics, The Hong Kong University
of Science and Technology, Clear Water Bay, Kowloon, Hong Kong (e-mail:
mamsxiong@ust.hk).
}
}

\maketitle

\begin{abstract}
Let  $m$ be an even positive  integer. A Boolean bent function $f$ on
$\GF{m-1}  \times \GF {}$ is called a \emph{cyclic bent function} if for any
$a\neq b\in \GF {m-1}$ and $\epsilon \in \GF{}$,  $f(ax_1,x_2)+f(bx_1,x_2+\epsilon)$
is always bent,
where $x_1\in \GF {m-1}, x_2 \in \GF {}$. Cyclic bent functions look extremely rare.
This paper focuses on cyclic bent functions on $\GF {m-1}  \times \GF {}$ and their
applications. The first objective of this paper is to construct a new class of cyclic bent functions, which includes all known constructions of  cyclic bent functions as special cases.
The second objective is to use cyclic bent functions  to
 construct good mutually unbiased bases (MUBs),  codebooks  and sequence families.
The third objective is to study cyclic semi-bent functions and their applications.
The fourth objective is to present a family of binary codes containing the Kerdock
code as a special case, and describe their support designs.
 The results of this paper show that cyclic bent functions and cyclic semi-bent functions have nice applications
 in several fields such as  symmetric cryptography, quantum physics, compressed  sensing and CDMA communication.
\end{abstract}


%

\IEEEpeerreviewmaketitle

\section{Introduction}

Let $m$ be a positive integer and let $f$ be a Boolean function from $V_m$ to the binary field $\mathbb F_2$, where $V_m$ is an $m$-dimensional vector space over $\mathbb F_2$.
The \emph{Walsh transformation} of the Boolean function $f$ at $a \in V_m$   is denoted by $\mathcal W_f(a)$ and defined by
\begin{align*}
\mathcal W_f(a)=\sum_{x\in V_m} (-1)^{f(x)+\langle a,x\rangle},
\end{align*}
where  $\langle \cdot, \cdot \rangle$ is an inner product over $V_m$. A Boolean function $f$ is called a \emph{bent} (resp., \emph{semi-bent}) function if for every $a \in V_m$, $\mathcal W_f(a)=\pm 2^{\frac{m}{2}}$
(resp., $\mathcal W_f(a)\in\{0,\pm 2^{\frac{m+1}{2}}\}$ for odd $m$). Note that the definitions of bent functions and semi-bent functions are independent of the choice of the inner product.
Bent  functions on  $V_m$ exist only for even $m$.
Bent functions were introduced by Rothaus \cite{Rot76} in 1976. They turned out to be rather complicated combinatorial objects. A  recent survey on bent functions can be found in \cite{CarletMesnagerDCC2016}. The book \cite{Mes15} is devoted to bent functions, their
generalisations and applications.

Let $m$ be an even positive integer throughout this paper unless otherwise stated.
A Boolean function $f(x_1, x_2)$ on $\GF {m-1} \times  \GF {}$ is called a \emph{cyclic bent
function} if all the functions $f(a x_1, x_2)+f(b x_1,x_2+\epsilon)$ are bent for all
$a\neq b\in \GF {m-1}$ and $\epsilon \in \GF{}$. So far, only the following two classes of cyclic bent functions are known in the literature:
\begin{itemize}
\item The function  from Kerdock codes  in \cite{Ker72}:
\begin{eqnarray}\label{eqn-Kerdockf}
K(x_1, x_2)=\sum_{i=1}^{\frac{m-2}{2}} \tr^{m-1}_1(x_1^{2^i+1})+x_2 \tr^{m-1}_1(x_1),
\end{eqnarray}
where and whereafter $\tr^{m-1}_1 (\cdot)$ denotes the absolute trace function on $\GF {m-1}$.
\item The function constructed by Zhou et al. in \cite{ZDL14}:
$$Z_{\gamma}(x_1,x_2)=\sum_{i=1}^{\frac{m-2}{2}} \tr^{m-1}_1(x_1^{2^i+1})+
\sum_{i=1}^{\frac{\ell-1}{2}} \tr^{m-1}_1((\gamma x_1)^{2^{ei}+1})+ x_2 \tr^{m-1}_1(x_1),$$ where
$e,\ell$ are positive integers satisfying $m-1=e \ell$ and $\gamma \in \mathbb F_{2^e}\setminus \{1\}$.
\end{itemize}

In this paper we will first present a general construction of cyclic bent functions.
With the framework of this construction, we will obtain many cyclic bent functions including
$K(x_1, x_2)$ and $Z_{\gamma}(x_1, x_2)$ as special cases.
We then explore applications of cyclic bent functions in constructing codebooks,
mutually unbiased bases (MUBs) and sequence families.

The first interesting application of cyclic bent functions is the construction of good codebooks. An $(N,K)$ \emph{codebook}
$\mathcal C$
is a set $\{\mathbf c_0, \ldots, \mathbf c_{N-1}\}$ of  $N$ unit norm $1\times K$ complex vectors $\mathbf c_i$.
The \emph{alphabet} of the codebook is the set of all different complex values that the coordinates of all $\mathbf c_i$ of $\mathcal C$ take.
The alphabet \emph{size} is the number of elements in the alphabet. As a performance measure of a codebook in practical applications,
the maximum crosscorrelation amplitude of $\mathcal C $ is defined as
\begin{align*}
I_{\mathrm{max}}(\mathcal C)=  \max_{0\le i <j \le N-1}   \mid
\mathbf c_i \mathbf c_j^{H} \mid,
\end{align*}
where $\mathbf c_j^{H}$ denotes  the conjugate  transpose of $\mathbf c_j$.
For any real-valued codebook $\mathcal C$ with $N> K(K+1)/2$,  $I_{\mathrm{max}}(\mathcal C)$ satisfies the following Levenshtein bound \cite{KL78,Lev83}:
\begin{align}\label{Lev-r}
I_{\mathrm{max}}(\mathcal C)\ge \sqrt{\frac{3N-K^2-2K}{(N-K)(K+2)}}.
\end{align}
For any complex-valued codebook $\mathcal C$ with $N> K^2$,  $I_{\mathrm{max}}(\mathcal C)$ satisfies the following Levenshtein bound \cite{KL78,Lev83}:
\begin{align}\label{Lev-c}
I_{\mathrm{max}}(\mathcal C)\ge \sqrt{\frac{2N-K^2-K}{(N-K)(K+1)}}.
\end{align}
In general, it is very hard to construct optimal real-valued (resp., complex-valued) codebooks meeting the Levenshtein bound in (\ref{Lev-r}) (resp., (\ref{Lev-c})).
Ding and Yin constructed optimal $(q^2+q, q)$ codebooks meeting the Levenshtein bound in (\ref{Lev-c}) for odd $q$ using planar functions.
Recently, Zhou, Ding and Li proposed a construction of optimal codebooks meeting the Levenshtein bound
  with a set of bent functions satisfying certain conditions  \cite{ZDL14}.
  Using the sets of bent functions generated by the cyclic bent functions $Z_{\gamma}(x_1,x_2)$, Zhou et  al. obtained  optimal
 real-valued $(2^{2m-1}+2^m,2^m)$ codebooks with alphabet size $4$, where $m$ is even.  Other optimal
 real-valued $(2^{2m-1}+2^m,2^m)$ codebooks with alphabet size $4$ were constructed by Kerdock codes in \cite{CCKS97,WF89}.
 Optimal  complex-valued $(2^{2m}+2^m,2^m)$  codebooks with alphabet size $6$ achieving the Levenshtein bound were derived from Kerdock codes over $\mathbb Z_4$   \cite{CCKS97} and
 the set of generalized quadratic bent functions on $\mathbb Z_4$ \cite{HY17}. The Levenshtein bounds are not tight in certain ranges, and thus cannot be met.
 In \cite{DingYin07,XiangDingMesnager2015}, for odd $m$, real-valued $(2^{2m}+2^m,2^m)$ codebooks almost meeting the Levenshtein bound of (\ref{Lev-r}) were constructed by employing
 some sets of semi-bent functions.
 For the construction of codebooks meeting or almost meeting the Levenshtein bound of  (\ref{Lev-r}),
 the following open problems have been suggested in \cite{XiangDingMesnager2015}.

\begin{orp}\label{orp:0}
Let $m$ be even.
Find a family $\{f_a,\,a\in\GF m^\star\}$ of $2^{m-1}-1$ bent functions of $m$ variables
such that the sum of any two distinct elements of the family is also bent.
\end{orp}

\begin{orp}\label{orp:1}
Let $n$ be odd.
 Find a family $\{f_a,\,a\in\GF n^\star\}$ of $2^n-1$ semi-bent functions of $n$ variables such that the sum
of any two distinct elements of the family is also semi-bent.
\end{orp}

A set of bent functions meeting the requirements in  Open problem \ref{orp:0} can be used to
construct optimal real-valued  $(2^{2m-1}+2^m, 2^m)$ codebook. A set of semi-bent functions meeting the requirements in Open problem \ref{orp:1} can been used to
construct  real-valued  $(2^{2m}+2^m, 2^m)$ codebooks almost achieving the Levenshtein bound of (\ref{Lev-r}).
We will  derive  a lot of sets of bent functions of $m$ variables  meeting the requirements in  Open problem \ref{orp:0}
 and sets of semi-bent functions with $m-1$ variables meeting the requirements in  Open problem \ref{orp:1} from any cyclic bent function on $\mathbb F_{2^{m-1}}\times \mathbb F_2$.
 Thus, optimal real-valued  $(2^{2m-1}+2^m, 2^m)$ codebooks   and  real-valued  $(2^{2(m-1)}+2^{m-1}, 2^{m-1})$ codebooks almost meeting the Levenshtein bound are obtained from  any
 cyclic bent function  on $\mathbb F_{2^{m-1}}\times \mathbb F_2$. Further, we also construct some classes of optimal  complex-valued  $(2^{2(m-1)}+2^{m-1}, 2^{m-1})$ codebooks
   meeting the Levenshtein bound of (\ref{Lev-c}) via any cyclic bent functions  on $\mathbb F_{2^{m-1}}\times \mathbb F_2$.

\par
The second interesting application of cyclic bent functions is to construct \emph{mutually unbiased bases (MUBs)}.
  The notion of MUBs  emerged in the literature of quantum mechanics in 1960 by  Schwinger \cite{Sch60}. Two orthonormal bases $\mathcal  B $ and $\mathcal  B'$ of the vector space $\mathbb C^{K}$  are called \emph{mutually unbiased} if
$| \langle\mathbf b, \mathbf b' \rangle|=|\sum_{i=1}^K   b_i \overline{b_i'}|=\frac{1}{\sqrt{K}} $ for all $\mathbf b=(b_1, \ldots, b_{K}) \in \mathcal B$ and all
$\mathbf b'=(b_1', \ldots, b_{K}') \in \mathcal B'$.

Any collection of pairwise MUBs of $\mathbb C^K$ has cardinality $K+1$ or less, see \cite{BBRV02}. Hence, we refer to a set of $K+1$ MUBs of $\mathbb C^K$ as a complete set.
It is  known that a complete set of MUBs exists if $K$ is a prime power.
Ivanovi$\acute{c}$ \cite{Iva81} gave an explicit construction of a complete set of MUBs of $\mathbb C^K$,
where $K$ is a prime.
His construction was later generalized to the case that $K$ is a prime power  in the influential paper by Wootters and Fields \cite{WF89}. Ding and Yin gave a further generalisation of the
construction of MUBs using general planar functions \cite[p. 939]{DingYin07}.
Constructing a complete set of MUBs
seems very hard in general. We will give a construction of  a complete set of MUBs of $\mathbb C^{2^m}$ using a
cyclic bent function, where $m$ is even.

\par
  The third interesting application of cyclic bent functions is to construct sequence families.
  A \emph{sequence} with period $K$ can be denoted by $\{s(t)\}_{t=0}^{K-1}$ or $\{s(t)\}_{t=0}^{\infty}$,
  and can also be considered as a vector $(s(0), \ldots, s(K-1))$ in the complex vector  space $\mathbb C^K$.
  A sequence $\{s(t)\}_{t=0}^{\infty}$  is called a \emph{binary} (resp., \emph{quaternary}) sequence if
  $s(t)=\pm 1$ (resp., $\pm 1, \pm \sqrt{-1}$).
  The \emph{correlation function} between two  sequences $\{s_0(t)\}_{t=0}^{\infty}$ and
  $\{s_1(t)\}_{t=0}^{\infty}$ with period $K$ at shift  $\tau$  is defined by
  \begin{align*}
  R_{s_0,s_1}(\tau)=\sum_{t=0}^{K-1} s_0(t+\tau) \overline{s_1(t)},
  \end{align*}
  where $0\le \tau <K$.
  If two sequences $\{s_0(t)\}_{t=0}^{\infty}$ and
  $\{s_1(t)\}_{t=0}^{\infty}$ are the same,   the  correlation function is also called the  autocorrelation   function.
  For a family $\mathcal U$ of sequences the important parameters are its size $N$,
   the period $K$ of the sequences in the family and the maximum magnitude of  correlation $R_{\mathrm{max}}(\mathcal U)$ among all pairs of sequences in the family $\mathcal U$,
   with the exclusion of the autocorrelation of a sequence and its $0$th shift (which always has the value $K$).
   A family of  sequences with size $N$ and period $K$ is called a $(N,K)$ family of sequences.
   Welch \cite{Wel74}  and Sidelnikov \cite{Sid71}  have established well-known lower bounds regarding the smallest possible $R_{\mathrm{max}}(\mathcal U)$  for $(N,K)$
   family of sequences. In general,
\begin{align*}
   R_{\mathrm{max}}(\mathcal U)=\begin{cases}
O(\sqrt{2K}), & \mathcal U \text{ is binary},\\
O(\sqrt{ K}), & \mathcal U \text{ is non-binary}.
\end{cases}
\end{align*}

In the binary case, the family of  Gold sequences with period $K=2^n-1$ and family size $N=2^n+1$, for $n$ an odd integer,
 has  been known to be asymptotically optimal \cite{SP80},
 with respect to the Welch and Sidelnikov bounds for binary sequences.
 In the non-binary case, the family $\mathcal A$ of quaternary sequences with period $K=2^n-1$ and family size $N=2^n+1$  was originally discovered by Sol\'e \cite{Sol89}, which is asymptotically optimal with respect to the Welch and Sidelnikov bounds.
The family $\mathcal A$ achieves the potential performance improvement for optimal non-binary versus binary designs.
The exact distribution of correlation values of the family $\mathcal A$ was given in \cite{BHK92}.
Further results on the family $\mathcal A$ can be found in Helleseth and Kumar \cite{HK98}, and Udaya and Siddiqi \cite{US98}.
By the Gray map of the  family $\mathcal A$,  Helleseth and Kumar \cite{HK98} defined the binary family $\mathcal Q(2)$ of  Kerdock sequences,
comprised of $N=2^{n}$  sequences of period $K=2(2^n-1)$, which is  optimal with respect to the well-known Welch bound.
Tang, Udaya, and Fan obtained the family $\mathcal Q(2)$ of  Kerdock sequences from a distinct technique \cite{TUF07}.
In \cite{THJ08} Tang, Helleseth and Johansen gave  the  correlation value distribution of the family $\mathcal Q(2)$.

Using any cyclic bent function $f$ on $\mathbb F_{2^{m-1}} \times \mathbb F_2$, we will in this paper
construct a family of
quaternary sequences of period $K=2^{m-1}-1$ and size $N=2^{m-1}-1$, and a family of binary sequences of period $K=2(2^{m-1}-1)$ and size $N=2^{m-1}$.
These families are   optimal with respect to the well-known Welch bound. We also completely determine
 the  correlation value distributions of these two families of sequences
using the Walsh transformation of Boolean functions. For odd  $n$, we will construct a family of binary sequences
with period $K=2^n-1$ and family size $N=2^n+1$ from a cyclic semi-bent function on $\mathbb F_{2^n}$. This family
is asymptotically optimal  with respect to the Welch and Sidelnikov bounds for binary sequences.
The  correlation distribution of this family of sequences will be determined.

The fourth interesting application of cyclic bent functions is in constructing optimal
nonlinear codes with the same parameters as the Kerdock code and $3$-designs. In this paper, we will construct a family of binary nonlinear codes with
cyclic bent functions, which contains the Kerdock code and holds $3$-designs.

The paper is organized as follows.
We start with presenting in Section \ref{sec:pre} some required background on Boolean functions and binary polynomials.
In Section \ref{sec:fun-codebook} we  construct optimal real-valued codebooks (Theorem \ref{thm:codebook-2}) with cyclic bent functions  and give a
general construction of cyclic bent functions (Theorem \ref{thm-F}).
In Section \ref{sec:MUBs} we construct a complete set of   MUBs (Theorem \ref{thm:MUBs}) using cyclic bent functions and obtain
optimal complex-valued codebooks (Theorem \ref{thm:codebook-4}).
In Section \ref{sec:family} we obtain a family of
quaternary sequences (Theorem \ref{thm:U-f}) and a family of binary sequences (Theorem \ref{thm:u-f-b}) with cyclic bent functions, and completely determine
the correlation value distributions of these two families of sequences.
In Section \ref{sec-Kerdockcodedesign} we present a family of nonlinear binary codes
from cyclic bent functions
and introduce their $3$-designs.
In Section \ref{sec:semi-bent} we introduce cyclic semi-bent functions, present a construction of cyclic semi-bent functions from cyclic bent functions (Theorem \ref{thm:semi-bent}), and
construct codebooks (Theorem \ref{thm:semi-codebook-2}) and families of sequences (Theorem \ref{thm:semi-sequence}) from cyclic semi-bent functions. We also construct a family of  binary codes with cyclic semi-bent functions. Furthermore, in Section \ref{subsec:char} we give a characterization of quadratic cyclic semi-bent functions.
Section \ref{sec:Conc} concludes this paper and makes some comments.

\section{Preliminaries}\label{sec:pre}

Recall that  for any positive integers $k$, and $r$
dividing $k$, the trace function from $\GF{k}$ to $\GF{r}$, denoted
by $\tr_{r}^k$, is the mapping defined for every $x\in\GF k$ as:
\begin{displaymath}
  \tr_{r}^k(x):=\sum_{i=0}^{\frac kr-1}
  x^{2^{ir}}=x+x^{2^r}+x^{2^{2r}}+\cdots+x^{2^{k-r}}.
\end{displaymath}

In particular, the {\em absolute trace} occurs for $r=1$.

Let $V_m$ be an $m$-dimensional vector space over $\mathbb F_2$ and $q$ be a function from $V_m$ to $\mathbb F_2$.
We  call $q$ a quadratic function
if $B_q(x,y)=q(x+y)+q(x)+q(y)+q(0)$
is bilinear, symmetric and symplectic (that is, null on the diagonal) over $ V_m \times V_m$.
The kernel of $B_q$,  denoted by $\ker B_q$, is the subspace $\{x\in\GF m\mid \forall y\in\GF m,\,B_q(x,y)=0\}$. Then we have  the following characterization of the bentness property of quadratic Boolean functions (see \cite{Car10}).

\begin{theorem}\label{thm:bent}
Let $q$ be a quadratic function over $\GF m$. Then, $q$ is bent if and only if $\dim_{\GF{}} (\ker B_q)=0$.
\end{theorem}

A linearized polynomial $L$ over $\GF m$ denotes any polynomial of the form $L(x)=\sum_{i=0}^{m-1}a_ix^{2^i}$
where the $a_i$'s are in $\GF m$. A linearized polynomial $L$ is a linear map of  the vector space $\GF m$ over $\GF {}$. Let $\ker L$ be the
subspace $\{x\in\GF m\mid L(x)=0\}$. The adjoint of a linearized polynomial $L$ over $\GF m$
is denoted by $L^\star$, which is the
unique linearized polynomial defined by
\begin{equation*}
\forall (x,y)\in\GF m\times\GF m,\quad \Tr m (xL(y)) = \Tr m (yL^\star(x)).
\end{equation*}

Let $\sigma $ be a permutation  on $V_m$ such that for any bent function $f$, $f \circ \sigma$ is also bent.  Then $\sigma (x)=x A+b$, where $A$ is an $m \times m$ non-singular binary matrix over $V_m$, $xA$ is the product of the row-vector $x$ and $A$, and $b\in V_m$. All these permutations form an automorphism of the set of bent functions. Two functions $f (x)$ and $g(x) = f \circ \sigma(x) $
are called linearly equivalent. If $f (x )$ is bent and $l (x )$ is an affine function,
then $f +l$ is also a bent function. Two functions $f $ and $f \circ \sigma + l$ are called extended-affine (for short, EA) equivalent.

\section{Cyclic bent functions and codebooks}\label{sec:fun-codebook}

In this section, we shall present a general construction of codebooks from cyclic bent functions. With this construction, optimal real-valued  codebooks meeting the Levenshtein bound can be obtained. We also construct a class of cyclic bent functions.

\subsection{A general construction of optimal codebooks with cyclic bent functions}

The following theorem is a solution to Open problem \ref{orp:0}, and a generalisation of the
construction in \cite{ZDL14}.

\begin{theorem}\label{thm:cyc-set-even}
Let $m$ be an even positive integer and $f(x_1, x_2)$ be a cyclic bent function over $\mathbb F_{2^{m-1}} \times \mathbb F_2$.
Let $\epsilon_a \in \mathbb F_2$,  where
$a\in \mathbb F_{2^{m-1}}$.
Define the following set
$$
\mathcal{F}=\{f(a x_1,  x_2+\epsilon_a) \mid a\in \mathbb{F}_{2^{m-1}}^\star\}.
$$
Then $\mathcal{F}$ is a family of
$2^{m-1}-1$ bent functions such that
the sum of any two distinct elements of this family is bent.
\end{theorem}

\begin{proof}
From the definition of cyclic bent functions, $f(a x_1,  x_2+\epsilon_a)+f(0 \cdot x_1, x_2)$ is bent for $a\in \mathbb{F}_{2^{m-1}}^\star$.
Note that $f(0 \cdot x_1, x_2)=u x_2+v$ for some $u, v\in \mathbb F_2$. Then, $f(a x_1,  x_2+\epsilon_a)$ and $f(a x_1,  x_2+\epsilon_a)+f(0 \cdot x_1, x_2)$
are EA-equivalent. Hence, $f(a x_1,  x_2+\epsilon_a)$ is bent.

Let $a, b \in  \mathbb{F}_{2^{m-1}}^\star$ with $a\neq b$. Then, $f(a x_1,  x_2+\epsilon_a)+f(b x_1,  x_2+\epsilon_b)$ must
be EA-equivalent to $f(a x_1,  x_2)+f(b x_1,  x_2+\epsilon_a+\epsilon_b)$, which is bent from the definition of cyclic bent functions.
Thus, $f(a x_1,  x_2+\epsilon_a)+f(b x_1,  x_2+\epsilon_b)$ is bent. This completes the proof.
\end{proof}

From Theorem \ref{thm:cyc-set-even} and Theorem 2 in \cite{ZDL14}, we have the following
theorem, which gives the construction of codebooks.

\begin{theorem}\label{thm:codebook-2}
Let $m$ be an even positive integer and $f(x_1, x_2)$ be a cyclic bent function over $\mathbb F_{2^{m-1}} \times \mathbb F_2$. Let $\epsilon_a \in \mathbb F_2$,  where
$a\in \mathbb F_{2^{m-1}}^{\star}$. Construct  a codebook as
\begin{align*}
\mathcal C_f= \bigcup_{a\in \mathbb F_{2^{m-1}}^\star} \mathcal B_a \bigcup \mathcal B_{0} \bigcup \mathcal B_{\infty},
\end{align*}
where $\mathcal B_{\infty}$ is the standard basis  of the $2^{m}$-dimensional Hilbert space $\mathbb     C^{2^{m}}$
in which each basis vector has a single nonzero entry with value $1$,
\begin{align*}
\mathcal B_{0} =\left \{\frac{1}{2^{\frac{m}{2}}} \left ( (-1)^{\tr^{m-1}_1(\lambda x_1)+\nu x_2} \right )_{(x_1,x_2) \in \mathbb F_{2^{m-1}}\times \mathbb F_2}
\mid (\lambda, \nu) \in \mathbb F_{2^{m-1}}\times \mathbb F_2  \right \},
\end{align*}
and for each $a\in \mathbb F_{2^{m-1}}^\star$,
\begin{align*}
\mathcal B_{a} =\left \{\frac{1}{2^{\frac{m}{2}}} \left ( (-1)^{f(ax_1, x_2+\epsilon_a)+\tr^{m-1}_1(\lambda x_1)+\nu x_2} \right )_{(x_1,x_2) \in \mathbb F_{2^{m-1}}\times \mathbb F_2}
\mid (\lambda, \nu) \in \mathbb F_{2^{m-1}}\times \mathbb F_2  \right \}.
\end{align*}
Then, $\mathcal C_f$ is an optimal real-valued $(2^{2m-1}+2^m, 2^m)$ codebook meeting the Levenshtein bound  of (\ref{Lev-r}) with alphabet size  $4$.
\end{theorem}

\begin{remark}
For each cyclic bent function $f$, we can select $2^{2^{m-1}-1}$  binary vectors $(\epsilon_a)_{a\in \mathbb F_{2^{m-1}}^{\star}}$ in $ \mathbb F_2^{2^{m-1}-1}$.
Each of the vectors $(\epsilon_a)_{a\in \mathbb F_{2^{m-1}}^{\star}}$ gives an optimal real-valued codebook with alphabet size $4$.
\end{remark}

The following proposition is very effective to test whether a Boolean function is cyclic bent.
\begin{proposition}
Let $(\lambda, \nu) \in \mathbb F_{2^{m-1}}\times \mathbb F_2$. Let $f$ be a Boolean function on  $\mathbb F_{2^{m-1}} \times \mathbb F_2$ such that
$f(x_1,x_2+1)+f(x_1,x_2)=\tr^{m-1}_1(\lambda  x_1)+\nu$. Then, the following three statements are equivalent.

(1) $f$ is a cyclic bent function;

(2) $f(x_1, x_2)+f(b x_1, x_2)$ is bent for any $b\in \mathbb F_{2^{m-1}} \setminus \{1\}$;

(3) $f$ is bent and $f(x_1, x_2)+f(b x_1, x_2)$ is bent for any $b\in \mathbb F_{2^{m-1}} \setminus \mathbb F_2$.
\end{proposition}
\begin{proof}
(1) $\Rightarrow$ (2) and (2) $\Rightarrow$ (3) are obvious.

We now prove (3) $\Rightarrow$ (1). For any $a, b \in \FF_{2^{m-1}}$ with $a \neq b$, we may assume that $a \ne 0$. For any $\epsilon \in \FF_2$, it is easy to see that function $f(ax_1, x_2)+f(bx_1, x_2+\epsilon)$ is EA-equivalent to $f(x_1, x_2)+f(\frac{b}{a} x_1, x_2)$, which is EA-equivalent to $f(x_1,x_2)$ if $b =0$. By (3), $f(ax_1, x_2)+f(bx_1, x_2+\epsilon)$ is bent. Thus $f$ is a cyclic bent function.
\end{proof}

\subsection{A new class of cyclic bent functions}

Let $m$  be an even positive integer and $l$ be a positive integer. Let
$e_0$, $e_1,\ldots, e_{l-1},  e_{l}$ be a sequence of positive integers, such that
$e_0=1$, $e_i|e_{i+1}$, $e_{l}=m-1$ and $e_i\neq  e_{i+1}$. Set
$f_i=\frac{m-1}{e_i}$. For
any $j\in \{0,1,\ldots,l-1\}$, define the quadratic function:
\begin{equation}\label{Qj(x)}
Q_j(x)=\tr_{e_{0}}^{e_{l}}
\left(\sum_{i=1}^{\frac{f_j-1}{2}}
x^{2^{ie_j}+1}\right).
\end{equation}
Let $\gamma_j\in \mathbb{F}_{2^{e_j}}~(
j=0,1,\ldots,l-1)$ such that
for any $j=0,1,\ldots,l-1$,
\begin{equation}\label{gamma-j}
\sum_{i=0}^{j}\gamma_i\neq 0.
\end{equation}
It is equivalent to $\gamma_0=1$, $\gamma_1\in \mathbb F_{2^{e_1}} \setminus \{\gamma_0\}$,
$\gamma_2\in \mathbb F_{2^{e_2}} \setminus \{\gamma_0+ \gamma_1 \}$,
$\ldots$ , $\gamma_{l-1}\in \mathbb F_{2^{e_{l-1}}} \setminus \{\gamma_0+ \gamma_1+\cdots +\gamma_{l-2} \}$.
Thus, there are $(2^{e_1}-1) (2^{e_{2}}-1) \cdots (2^{e_{l-1}}-1)$ ways to choose $\gamma_0, \gamma_1, \ldots, \gamma_{l-1}$.

Define the quadratic Boolean function $f(x_1,x_2)$ from  $\mathbb{F}_{2^{m-1}}
\times \mathbb{F}_2$ to
$\mathbb{F}_2$ by
\begin{equation}\label{eq:cyclic-bent}
f(x_1,x_2)=\sum_{j=0}^{l-1}Q_j(\gamma_j x_1)
+x_2\tr_{e_0}^{e_l}(x_1).
\end{equation}

For any $a\in \mathbb{F}_{2^{m-1}}^\star$, define the quadratic Boolean function:
\begin{equation}\label{fa12}
f_a(x_1,x_2):=f(ax_1,x_2)=\sum_{j=0}^{l-1}Q_j(\gamma_j ax_1)
+x_2\tr_{e_0}^{e_l}(ax_1),
\end{equation}
where $x_1\in \mathbb{F}_{2^{m-1}}$,
$x_2\in \mathbb{F}_2$.

The following theorem gives a new
class of cyclic bent functions.

\begin{theorem}\label{thm-F}
Let $m$ be an even positive integer, and $e_0$, $e_1,\ldots, e_{l-1},  e_{l}$ be a  sequence of positive integers with
$e_0=1$, $e_i|e_{i+1}$, $e_{l}=m-1$. Set
$f_i=\frac{m-1}{e_i}$. Let
$\gamma_j\in \mathbb{F}_{2^{e_j}}$ satisfy
(\ref{gamma-j}).
Let $f$ be the Boolean function  defined in
(\ref{eq:cyclic-bent}). Then, $f$ is a cyclic bent function.
\end{theorem}
The proof of this theorem will be given later.

\begin{remark}
When $l=1$,  one has $e_0=1$, $e_1=m-1$ and $\gamma_0=1$. Thus, the cyclic bent function  defined by (\ref{eq:cyclic-bent}) is
$\sum_{i=1}^{\frac{m-2}{2}} \tr^{m-1}_1(x_1^{2^i+1})+x_2 \tr^{m-1}_1(x_1)$, which is exactly the function from the Kerdock code reported in \cite{Ker72}.

When $l=2$,  one has $e_0=1$, $e_1=e$, $e_2=m-1$, $\gamma_0=1$ and $\gamma_1=\gamma \in \mathbb F_{2^e}\setminus \{1\}$. Thus,
the cyclic bent function  defined by (\ref{eq:cyclic-bent}) is
$\sum_{i=1}^{\frac{m-2}{2}} \tr^{m-1}_1(x_1^{2^i+1})+ \sum_{i=1}^{\frac{m-1-e}{2e}} \tr^{m-1}_1((\gamma x_1)^{2^{ei}+1})+ x_2 \tr^{m-1}_1(x_1)$, which is  exactly the function constructed
by Zhou et al.
in \cite{ZDL14}.
\end{remark}

As a consequence of Theorem \ref{thm-F}, we have the following corollary.

\begin{corollary}
Let $m$ be an even positive integer, and $e_0$, $e_1,\ldots, e_{l-1},  e_{l}$ be a  sequence of positive integers with
$e_0=1$, $e_i|e_{i+1}$, $e_{l}=m-1$. Set
$f_i=\frac{m-1}{e_i}$. Let $\gamma_j\in \mathbb{F}_{2^{e_j}}$ satisfy
(\ref{gamma-j}). Let $f$ be the function defined in  (\ref{eq:cyclic-bent})
and  $\epsilon_a\in \mathbb F_2$ for $a\in \mathbb{F}_{2^{m-1}}^\star$.
Define the following set
$$
\mathcal{F}=\{f(a x_1, x_2 + \epsilon_a) \mid a\in \mathbb{F}_{2^{m-1}}^\star\}.
$$
Then $\mathcal{F}$ is a family of
$2^{m-1}-1$ bent functions such that
the sum of any two distinct elements of this family is bent.
\end{corollary}

To prove Theorem \ref{thm-F}, we first present some auxiliary lemmas.

\begin{lemma}\label{Qjxz}
Let notation be defined in Theorem
\ref{thm-F}. Let $Q_j$ be defined in
(\ref{Qj(x)}). For any $x,z\in \mathbb{F}_{2^{m-1}}$,
$$
Q_j(\gamma_j a(x+z))+
Q_j(\gamma_j ax)+Q_j(\gamma_j az)
=\tr_{e_0}^{e_l}(a\gamma_j^2(ax+
\tr_{e_j}^{e_l}(ax))z).
$$
\end{lemma}

\begin{proof}
Note that
\begin{align*}
Q_j(\gamma_j a(x+z))
=\tr_{e_0}^{e_{l}}
\left(\sum_{i=1}^{\frac{f_j-1}{2}}(\gamma_j a)^{2^{ie_j}+1}
(x+z)^{2^{ie_j}+1}\right).
\end{align*}
It then follows from $(x+z)^{2^{ie_j}+1}=
x^{2^{ie_j}+1}+z^{2^{ie_j}+1}+
x^{2^{ie_j}}z+xz^{2^{ie_j}}$ that
$$
Q_j(\gamma_j a(x+z))=Q_j(\gamma_j ax)
+Q_j(\gamma_j az)
+\tr_{e_0}^{e_l}
\left(\sum_{i=1}^{\frac{f_j-1}{2}}(\gamma_j a)^{2^{ie_j}+1}
(x^{2^{ie_j}}z+xz^{2^{ie_j}})\right).
$$
Let $B=Q_j(\gamma_j a(x+z))+
Q_j(\gamma_j ax)+Q_j(\gamma_j az)$. Then
\begin{align*}
B=& \sum_{i=1}^{\frac{f_j-1}{2}}
\tr_{e_0}^{e_l}
\left((\gamma_j a)^{2^{ie_j}+1}
x^{2^{ie_j}}z\right)+\sum_{i=1}^{\frac{f_j-1}{2}}
\tr_{e_0}^{e_l}
\left((\gamma_j a)^{2^{ie_j}+1}
xz^{2^{ie_j}}\right).
\end{align*}
The second term on the right can be written as
\begin{eqnarray*}
\sum_{i=1}^{\frac{f_j-1}{2}}
\tr_{e_0}^{e_l}
\left(((\gamma_j a)^{2^{ie_j}+1}
xz^{2^{ie_j}})^{2^{e_j(f_j-i)}}\right) &=& \sum_{i=1}^{\frac{f_j-1}{2}}
\tr_{e_0}^{e_l}
\left((\gamma_j a)^{2^{e_j(f_j-i)}+1}
x^{2^{(f_j-i)e_j}}z\right) \\
&=& \sum_{i=
\frac{f_j+1}{2}}^{f_j-1}
\tr_{e_0}^{e_l}
\left((\gamma_j a)^{2^{ie_j}+1}
x^{2^{ie_j}}z\right).
\end{eqnarray*}
Adding this to the first term of $B$ we obtain
\begin{align*}
B=& \sum_{i=1}^{f_j-1}
\tr_{e_0}^{e_l}\left((\gamma_j a)
(\gamma_j ax)^{2^{ie_j}}z\right)=
\tr_{e_0}^{e_l}\left(\gamma_j a
\sum_{i=1}^{f_j-1}(\gamma_j ax)^{2^{ie_j}}z\right).
\end{align*}
Noting that $\tr_{e_j}^{e_l}(\gamma_j ax)
=\sum_{i=0}^{f_j-1}
(\gamma_j ax)^{2^{ie_j}}= \gamma_j \sum_{i=0}^{f_j-1}  (ax)^{2^{ie_j}}$, we have
$$
B=\tr_{e_0}^{e_l}\left(a\gamma_j^2(ax+
\tr_{e_j}^{e_l}(ax))z\right).
$$
The desired conclusion then follows.
\end{proof}

\begin{lemma}\label{Bfa}
Let the quadratic Boolean function
$f_a$ be defined in (\ref{fa12}).
Define
$$
B_{f_a}=f_a(x_1+z_1,x_2+z_2)
+f_a(x_1,x_2)+f_a(z_1,z_2),
$$
where $x_1,z_1\in \mathbb{F}_{2^{m-1}}$,
$x_2,z_2\in \mathbb{F}_2$. Then
$$
B_{f_a}=\tr_{e_0}^{e_l}\left ( a z_1 \left [\sum_{j=0}^{l-1}
 (a \gamma_j^2 x_1+\tr_{e_j}^{e_l}(
a \gamma_j^2 x_1) )
+x_2 \right ] \right )
+\tr_{e_0}^{e_l}(ax_1)z_2.
$$
\end{lemma}

\begin{proof}
From the definition of $B_{f_a}$, we have
\begin{align*}
B_{f_a}=&\sum_{j=0}^{l-1} \left (Q_j(\gamma_j a(x_1+z_1))+Q_j(\gamma_j ax_1)+Q_j(\gamma_j az_1) \right )\\
&+(x_2+z_2)\tr_{e_0}^{e_l}(a(x_1+z_1))
+x_2\tr_{e_0}^{e_l}(ax_1)+z_2 \tr_{e_0}^{e_l}(az_1).
\end{align*}
It then follows from Lemma \ref{Qjxz} that
\begin{align*}
B_{f_a}=& \sum_{j=0}^{l-1}
\tr_{e_0}^{e_l} \left (\gamma_j a(
\gamma_j ax_1+\tr_{e_j}^{e_l}(
\gamma_j ax_1))z_1 \right )
+x_2\tr_{e_0}^{e_l}(az_1)+z_2 \tr_{e_0}^{e_l}(ax_1)
\\=&
\tr_{e_0}^{e_l} \left ( \left [\sum_{j=0}^{l-1}
\gamma_j a \left (\gamma_j ax_1+\tr_{e_j}^{e_l}(
\gamma_j ax_1) \right )
+ax_2 \right ]z_1 \right )
+\tr_{e_0}^{e_l}(ax_1)z_2.
\end{align*}
This completes the proof.
\end{proof}

\begin{lemma}\label{lem:B-f-ab}
Let $f_a$ and $f_b$ be defined in
(\ref{fa12}), where
$a,b\in \mathbb{F}_{2^{m-1}}$ and
$a\neq b$. Let $f_{a,b}=f_a+f_b$ and
$B_{f_{a,b}}=
f_{a,b}(x_1+z_1, x_2+z_2)
+f_{a,b}(x_1,x_2)+f_{a,b}(z_1,z_2)$, where
$x_1,z_1\in \mathbb{F}_{2^{m-1}}$ and
$x_2,z_2\in \mathbb{F}_2$. Then
\begin{align*}
B_{f_{a,b}}=&  \tr_{e_{0}}^{e_l} \left(
 z_1\left[ \sum_{j=0}^{l-1} \left\{(a^2 +b^2)\gamma_j^2 x_1 +
 a\tr_{e_j}^{e_l}(\gamma_j^2 ax_1)
+b\tr_{e_j}^{e_l}(\gamma_j^2 bx_1) \right\}
 +(a+b)x_2 \right] \right) \\&
+\tr_{e_{0}}^{e_l}\left( (a+b)x_1z_2\right).
\end{align*}
\end{lemma}

\begin{proof}
The desired conclusion follows from Lemma \ref{Bfa}, by observing that $B_{f_{a,b}}=B_{f_a}+B_{f_b}$.
\end{proof}

\begin{lemma}\label{kerBfab}
Define  the following set
\begin{align*}
\ker(B_{f_{a,b}})=\{& (x_1,x_2)\in \mathbb{F}_{2^{m-1}}
\times \mathbb{F}_2 \mid
f_{a,b}(x_1+z_1,x_2+z_2)
+f_{a,b}(x_1,x_2)\\&\quad\quad\quad+
f_{a,b}(z_1+z_2)=0~
\text{for any}~ (z_1,z_2)
\in \mathbb{F}_{2^{m-1}}
\times \mathbb{F}_2
\}.
\end{align*}
Then $(x_1,x_2)\in
\ker(B_{f_{a,b}})$ if and only if
$(x_1,x_2)$ is a solution of the following
system of equations:
\begin{equation}\label{Bab-eqn}
\left\{
  \begin{array}{l}
    \sum_{j=0}^{l-1} \left((a^2 +b^2)\gamma_j^2 x_1 +
 a\tr_{e_j}^{e_l}(\gamma_j^2 ax_1)
+b\tr_{e_j}^{e_l}(\gamma_j^2 bx_1) \right)
 +(a+b)x_2=0  ,   \\
   \tr_{e_0}^{e_l}((a+b)x_1)=0,
  \end{array}
\right.
\end{equation}
\end{lemma}

\begin{proof}
The desired conclusion follows from Lemma \ref{lem:B-f-ab}.
\end{proof}

Now, we are ready to prove Theorem \ref{thm-F}.

\begin{proof}
For any $a,b\in \mathbb{F}_{2^{m-1}}$ with $a\neq b$, to prove that $f_{a,b}=f_a+b_b$ is bent, from Lemma \ref{kerBfab}, it suffices to prove that the system of equations (\ref{Bab-eqn}) has no none-zero solutions.

Suppose that the system of equations (\ref{Bab-eqn}) has nonzero solution
$(x_1,x_2)\in \mathbb{F}_{2^{m-1}}
\times \mathbb{F}_2$. Then
$x_1\neq 0$. We have
\begin{align} \label{l-1ab}
\sum_{j=0}^{l-1} \left((a^2 +b^2)\gamma_j^2 x_1^2 +
 ax_1\tr_{e_j}^{e_l}(\gamma_j^2 ax_1)
+bx_1\tr_{e_j}^{e_l}(\gamma_j^2 bx_1) \right)
 +(a+b)x_1x_2=0.
\end{align}
Since  $\tr_{e_{0}}^{e_l}((a+b)x_1)=0$ and $\tr_{e_{l}}^{e_l}((a+b)x_1)=(a+b)x_1 \neq 0$, there exists an integer $t$ where $0 \le t \le l-1$ such that
\begin{equation}\label{eq:t-a-b}
\left\{
  \begin{array}{l}
     \tr_{e_{t+1}}^{e_l}((a+b)x_1)\neq 0,\\
     \tr_{e_j}^{e_l}((a+b)x_1)=0,~
\text{for any}~0\leq j \leq t \,.
  \end{array}
\right.
\end{equation}
For any $j$ where $l\geq j\geq t+1$, since $\gamma_j \in \FF_{2^{e_j}}$, we have
\begin{align}\label{eq:Tr-(Tr+Tr)}
&\tr_{e_{t+1}}^{e_l} \left [ a^2\gamma_j^2x_1^2+ax_1
\tr_{e_j}^{e_l}(\gamma_j^2 ax_1)\right ]  \nonumber \\
=&
\tr_{e_{t+1}}^{e_j} \left (\tr_{e_{j}}^{e_l}\left [
a^2\gamma_j^2x_1^2+ax_1
\tr_{e_j}^{e_l}(\gamma_j^2 ax_1) \right ] \right )
\nonumber \\
=& \tr_{e_{t+1}}^{e_j} \left [ \left(\tr_{e_{j}}^{e_l}(\gamma_j ax_1)\right)^2+ \left( \tr_{e_{j}}^{e_l}(\gamma_j ax_1)\right)^2\right ] =0 \nonumber.
\end{align}
Let $E$ be the left hand side of the equation (\ref{l-1ab}). From the above argument we have
\begin{align*}
\tr_{e_{t+1}}^{e_l}(E)=&
\sum_{j=0}^{t} \tr_{e_{t+1}}^{e_l} \left[ (a^2 +b^2)\gamma_j^2 x_1^2 +
 ax_1\tr_{e_j}^{e_l}(\gamma_j^2 ax_1)
+bx_1\tr_{e_j}^{e_l}(\gamma_j^2 bx_1)  \right] \\
 & +\tr_{e_{t+1}}^{e_l} \left[ (a+b)x_1\right] x_2.
\end{align*}
From (\ref{eq:t-a-b}), for any $0 \le j \le t$, we may denote
$$
u_j=\tr_{e_{j}}^{e_l}(ax_1)=\tr_{e_{j}}^{e_l}(bx_1) \in \FF_{2^{e_j}}.
$$
Then we can obtain
\begin{align*}
\tr_{e_{t+1}}^{e_l}(E)=& \tr_{e_{t+1}}^{e_l}\left[ (a+b)x_1\right] \cdot
\left [
\left (\sum_{j=0}^{t}\gamma_j^2 \right ) \cdot
\tr_{e_{t+1}}^{e_l}\left[(a+b)x_1\right]+
\sum_{j=0}^t\gamma_j^2 u_j+x_2
\right ].
\end{align*}
Since $E=0$, $\tr_{e_{t+1}}^{e_l} \left[(a+b)x_1 \right]\neq 0$ and $\sum_{j=0}^t\gamma_j^2\neq 0$, we have
\begin{equation*}
\tr_{e_{t+1}}^{e_l} \left[(a+b)x_1 \right]=
\frac{1}{\sum_{j=0}^t\gamma_j^2}
\left (\sum_{j=0}^t\gamma_j^2 u_j+x_2 \right )
\in \mathbb{F}_{2^{e_t}}.
\end{equation*}
By using $e_l=m-1$ is odd, we find
\begin{align*}
\tr_{e_{t+1}}^{e_l} \left[(a+b)x_1 \right]
=& \tr_{e_{t}}^{e_{t+1}}\left[\tr_{e_{t+1}}^{e_l}
((a+b)x_1) \right]\\
=& \tr_{e_{t}}^{e_l} \left[(a+b)x_1 \right]=0,
\end{align*}
which contradicts the assumption that $\tr_{e_{t+1}}^{e_l} \left[(a+b)x_1 \right] \ne 0$. Hence,
the system of equations (\ref{Bab-eqn}) has no none-zero solution and
$f_{a,b}$ is bent. This completes the proof of Theorem \ref{thm-F}.
\end{proof}

\section{The construction of mutually unbiased bases using cyclic bent functions} \label{sec:MUBs}

In this section, we shall present a generic construction of a complete set of  MUBs from a cyclic bent function.
Let $\mathcal B_{\infty}$ denote the standard basis  of the $2^{m-1}$-dimensional Hilbert space $\mathbb     C^{2^{m-1}}$
in which each basis vector has a single nonzero entry with value $1$.
Let $f$ be  a Boolean function on $\GF{m-1} \times \GF{}$. For any $a\in \GF{m-1}$, define the following set:
\begin{align*}
\mathcal B_{a} :=\left \{ \left ( \left ( \varrho_0 (-1)^{f(ax,0)}  +   \varrho_1 (-1)^{f(ax,1)} \right )
\frac{(-1)^{\tr^{m-1}_1 (\lambda x)} }{\sqrt{2^{m-1}}}  \right )_{x\in \GF{m-1}} | \lambda \in \GF{m-1} \right \},
\end{align*}
where $\varrho_0 = \frac{1+\sqrt{-1}}{2}$ and $\varrho_1 = \frac{1-\sqrt{-1}}{2}$.

The following theorem gives a complete set of
 MUBs.
\begin{theorem}\label{thm:MUBs}
Let $m$ be an even positive integer and $f$ be a cyclic bent function on $\GF{m-1} \times \GF{}$. Then,
 the standard basis $\mathcal B_{\infty}$ and the sets $\mathcal B_{a}$, with $a \in \GF{m-1}$, form a complete set of $2^{m-1}+1$  MUBs of $\mathbb     C^{2^{m-1}}$.
\end{theorem}
The proof of this theorem will be given later. An auxiliary lemma is given below.
\begin{lemma}\label{lem:Z4}
Let $a_i, b_i, \in \GF{}$ and $A_i=\frac{1+\sqrt{-1}}{2} (-1)^{a_i} + \frac{1-\sqrt{-1}}{2} (-1)^{b_i}$ with $i=1,2$.

(1) $A_i \in \{ \pm 1, \pm \sqrt{-1}\}$.

(2) Let $\overline{A_2}$ be the the complex conjugate of the complex number $A_2$. Then,
\begin{align*}
 A_1 \overline{A_2}=\frac{(-1)^{a_1+a_2}+(-1)^{b_1+b_2}}{2}  + \frac{(-1)^{a_1+b_2}- (-1)^{b_1+a_2}}{2} \sqrt{-1}.
 \end{align*}

\end{lemma}

\begin{proof}
The proof is straightforward. We omit the details here.
%
%

\end{proof}

Proof of Theorem \ref{thm:MUBs}:

\begin{proof}
Let $\mathbf v_{a,\lambda}
=\left ( \left ( \varrho_0 (-1)^{f(ax,0)}  +   \varrho_1 (-1)^{f(ax,1)} \right )
\frac{(-1)^{\tr^{m-1}_1 (\lambda x)} }{\sqrt{2^{m-1}}}  \right )_{x\in \GF{m-1}}$,
where $\varrho_0 = \frac{1+\sqrt{-1}}{2}$,  $\varrho_1 = \frac{1-\sqrt{-1}}{2}$,
$a\in \mathbb{F}_{2^{m-1}}$, and
$\lambda\in \mathbb{F}_{2^{m-1}}$.
For any $a,a',\lambda,
\lambda'\in \mathbb{F}_{2^{m-1}}$,
let $A(a,x)=\varrho_0(-1)^{f(ax,0)}
+\varrho_1(-1)^{f(ax,1)}$.
We have
\begin{align}\label{Vaa'}
\langle \mathbf v_{a,\lambda},
\mathbf v_{a',\lambda'} \rangle=&
\sum_{x\in \mathbb{F}_{2^{m-1}}}
A(a,x)\overline{A(a',x)}
\frac{(-1)^{\tr_1^{m-1}((
\lambda+\lambda')x)}}{2^{m-1}}\nonumber\\
=&
\frac{1}{2^{m-1}} \sum_{x\in \mathbb{F}_{2^{m-1}}}
A(a,x)\overline{A(a',x)}(-1)^{\tr_1^{m-1}((
\lambda+\lambda')x)}.
\end{align}
If $a=a'$, then
$$
\langle \mathbf v_{a,\lambda},
\mathbf v_{a,\lambda'} \rangle
=\frac{1}{2^{m-1}} \sum_{x\in \mathbb{F}_{2^{m-1}}}
A(a,x)\overline{A(a,x)}(-1)^{\tr_1^{m-1}((
\lambda+\lambda')x)}.
$$
It follows from Part (1) of
Lemma \ref{lem:Z4} that $A(a,x)
\in \{\pm 1, \pm \sqrt{1}\}$. Consequently,
$$
\langle \mathbf v_{a,\lambda},
\mathbf v_{a,\lambda'} \rangle
=\frac{1}{2^{m-1}} \sum_{x\in \mathbb{F}_{2^{m-1}}}
(-1)^{\tr_1^{m-1}((
\lambda+\lambda')x)}
=\left\{
   \begin{array}{ll}
     1, & \lambda=\lambda', \\
     0, & \lambda\neq \lambda'.
   \end{array}
 \right.
$$
Thus, set $\mathcal{B}_a$ is an
orthonormal basis.

If $a\neq a'$, then from (\ref{Vaa'}) and Part (2) of
Lemma \ref{lem:Z4}, we have
\begin{align*}
&2\cdot 2^{m-1}
\langle \mathbf v_{a,\lambda},
\mathbf v_{a',\lambda'} \rangle\\
 = &
\sum_{x\in \mathbb{F}_{2^{m-1}}}
(-1)^{f(ax,0)+f(a'x,0)+\tr_1^{m-1}((\lambda
+\lambda')x)} \\&
+ \sum_{x\in \mathbb{F}_{2^{m-1}}}
(-1)^{f(ax,1)+f(a' x,1)+\tr_1^{m-1}((\lambda
+\lambda')x)} \\
&+ \sqrt{-1}\sum_{x\in \mathbb{F}_{2^{m-1}}}
(-1)^{f(ax,0)+f(a' x,1)+\tr_1^{m-1}((\lambda
+\lambda')x)}\\&
- \sqrt{-1}\sum_{x\in \mathbb{F}_{2^{m-1}}}
(-1)^{f(ax,1)+f(a' x,0)+\tr_1^{m-1}((\lambda
+\lambda')x)}\\
=&\sum_{x_1\in \mathbb{F}_{2^{m-1}}
, x_2\in \mathbb{F}_2}
(-1)^{f_{a,a',0}(x_1,x_2)+\tr_1^{m-1}((\lambda
+\lambda')x_1)}\\
&+ \sqrt{-1}
\sum_{x_1\in \mathbb{F}_{2^{m-1}}
, x_2\in \mathbb{F}_2}
(-1)^{f_{a,a',1}(x_1,x_2)+
\tr_1^{m-1}((\lambda
+\lambda')x_1)+x_2}\\
=& \mathcal{W}_{
f_{a,a',0}}(
\lambda+\lambda',0)
+\sqrt{-1}\mathcal{W}_{
f_{a,a',1}}(
\lambda+\lambda',1),
\end{align*}
where $f_{a,a',0}(x_1,x_2)= f(ax_1,x_2)+f(a' x_1,x_2)$ and $f_{a,a',1}(x_1,x_2)=f(ax_1,x_2)+f(a' x_1,x_2+1)$.
Since $f$ is a cyclic bent function, we have
\begin{align*}
|\langle \mathbf v_{a,\lambda},
\mathbf v_{a',\lambda'} \rangle|=
& \frac{1}{2^m}
\sqrt{(\mathcal{W}_{
f_{a,a',0}}(
\lambda+\lambda',0))^2+(\mathcal{W}_{
f_{a,a',1}}(
\lambda+\lambda',1))^2}\\
=& \frac{1}{2^m} \sqrt{2^{m}+ 2^{m}}\\
=& \frac{1}{\sqrt{2^{m-1}}}.
\end{align*}
The set $\mathcal{B}_a$
and the set $\mathcal{B}_{a'}$ are mutually unbiased.

By Part (1) of Lemma \ref{lem:Z4},
the entries of the vectors
$\mathbf v_{a,\lambda}$ have the absolute value
 $\frac{1}{\sqrt{2^{m-1}}}$. Hence, the
standard basis $\mathcal{B}_{\infty}$
and the set $\mathcal{B}_{a}$ are
mutually unbiased for all
$a\in\mathbb{F}_{2^{m-1}}$.
This completes the proof.
\end{proof}

Since $\mathcal B_{\infty}$ and the bases $\mathcal B_{a}$ ($a\in \mathbb F_{2^{m-1}}$)  form a complete set of MUBs, the square of the absolute value of the inner product
$\mid \langle\mathbf v, \mathbf v'\rangle \mid ^2$
is equal to $\frac{1}{2^{m-1}}$  for any two vectors $\mathbf v, \mathbf v'$  from distinct bases. Then the following result follows directly from the Levenshtein bound of (\ref{Lev-c}).

\begin{theorem}\label{thm:codebook-4}
Let $f$ be a cyclic bent function on $\mathbb F_{2^{m-1}} \times \mathbb F_2$, $\mathcal B_{\infty}$ and $\mathcal B_{a}$ ( $ a\in \mathbb F_{2^{m-1}}$) be
 defined as in Theorem \ref{thm:MUBs}, and    $\mathcal C= \mathcal B_{\infty} \bigcup _{a\in \mathbb F_{2^{m-1}}}  \mathcal B_{a}$. Then $\mathcal C$ is
 an optimal  $(2^{2(m-1)}+2^{m-1}, 2^{m-1})$ complex-valued codebook meeting the  Levenshtein bound
of (\ref{Lev-c})  with alphabet size $6$.
\end{theorem}

\section{Several sequence families from cyclic bent functions}\label{sec:family}

In this section, using cyclic bent functions, we shall present two generic constructions
of families of sequences almost meeting the Welch bound and completely determine their correlation value distributions.

\subsection{A quaternary sequence family from cyclic bent functions}

Let $\beta$ be a primitive element
in $\mathbb F_{2^{m-1}}$  and $f$ be a cyclic bent function on $\mathbb F_{2^{m-1}}\times \mathbb F_2$. For $\lambda \in \mathbb F_{2^{m-1}}$,
define the  sequence $\{s_{\lambda}(t)\}_{t=0}^{\infty}$ by
\begin{align*}
 s_{\lambda}(t)=  \left ( \varrho_0 (-1)^{f(\beta^t,0)}  +   \varrho_1 (-1)^{f(\beta^t,1)} \right )
(-1)^{\tr^{m-1}_1 (\lambda \beta^t)},
\end{align*}
where $\varrho_0 = \frac{1+\sqrt{-1}}{2}$ and $\varrho_1 = \frac{1-\sqrt{-1}}{2}$.
By Part (2) of Lemma \ref{lem:Z4},
$\{s_{\lambda}(t)\}_{t=0}^{\infty}$ is a quaternary sequence.

A family $\mathcal U_f$  of quaternary sequences is defined by
\begin{align*}
\mathcal U_f =\left  \{ \{s_{\lambda}(t)\}:  \lambda \in \mathbb F_{2^{m-1}} \right \} \cup \left \{ \{s_{\infty}(t)\} \right \},
\end{align*}
where $s_{\infty}(t) =  (-1)^{\tr^{m-1}_1 ( \beta^t) }$.

For $a,b\in \mathbb F_{2^{m-1}}$ and $\epsilon \in \mathbb F_2$, set $f_{a,b,\epsilon}:=f(ax_1,x_2)+ f(bx_1, x_2+\epsilon)$.
The correlation values of the sequence family $\mathcal U_f$ can be determined by the following lemma.

\begin{lemma}\label{Rss}
Let $f$ be a cyclic bent function with $f(0,0)=f(0,1)=0$, and $\tau$ be a non-negative integer with $0\le \tau <2^{m-1}-1$.

(1) If $\lambda, \lambda ' \in \mathbb F_{2^{m-1}}$, then
\begin{align*}
 R_{s_{\lambda}, s_{\lambda'}}(\tau)=\frac{1}{2}\left [ \mathcal W_{f_{1,\beta^{\tau},0}}(\lambda \beta^{\tau}+\lambda ',0)
 -\sqrt{-1} \mathcal W_{f_{1,\beta^{\tau},1}}(\lambda \beta^{\tau}+\lambda ',1) \right ] -1.
 \end{align*}

(2) If $\lambda \in \mathbb F_{2^{m-1}}$, then
\begin{align*}
R_{s_{\lambda}, s_{\infty}}(\tau)= \frac{1}{2} \left [ \mathcal W_{f_{\beta^{\tau}}}(\lambda \beta^{\tau}+1,0) +\sqrt{-1} \mathcal W_{f_{\beta^{\tau}}}(\lambda \beta^{\tau}+1,1) \right ]-1,
\end{align*}
and
\begin{align*}
R_{s_{\infty}, s_{\lambda}}(\tau)=\frac{1}{2} \left [ \mathcal W_{f}(\lambda + \beta^{\tau},0) - \sqrt{-1} \mathcal W_{f}(\lambda + \beta^{\tau},1) \right ]-1.
\end{align*}

(3) For the binary sequence  $\{s_{\infty}(t)\}$, its autocorrelation at shift $\tau$ is
\begin{align*}
 R_{s_{\infty}, s_{\infty}}(\tau)= \left\{
  \begin{array}{ll}
      -1,  & \text{if~~} \tau\neq 0, \\
    2^{m-1} -1,  & \text{if~~} \tau=0.
  \end{array}
\right.
 \end{align*}

\end{lemma}

\begin{proof}
(1) For $a, x_1\in \mathbb F_{2^{m-1}}$, set $A(a, x_1)= \varrho_0 (-1)^{f(ax_1,0)}  +   \varrho_1 (-1)^{f(ax_1,1)}$
with $\varrho_0=\frac{1+\sqrt{-1}}{2}$ and $\varrho_1=\frac{1-\sqrt{-1}}{2}$.
By the definition of the correlation of sequences, one has
\begin{align}\label{eq:R-A-A}
R_{s_{\lambda}, s_{\lambda'}}(\tau)=& \sum_{t=0}^{2^{m-1}-2} A(\beta^{\tau}, \beta^t) \overline{A(1, \beta^t)} (-1)^{\tr^{m-1}_1 ((\lambda \beta^{\tau}+\lambda ') \beta^t)}\nonumber\\
=& \sum_{x_1 \in \mathbb F_{2^{m-1}}} A(\beta^{\tau}, x_1) \overline{A(1, x_1)} (-1)^{\tr^{m-1}_1 ((\lambda \beta^{\tau}+\lambda ') x_1)}\nonumber\\
&-A(\beta^{\tau}, 0) \overline{A(1, 0)}\nonumber\\
= & \sum_{x_1 \in \mathbb F_{2^{m-1}}} A(\beta^{\tau}, x_1) \overline{A(1, x_1)} (-1)^{\tr^{m-1}_1 ((\lambda \beta^{\tau}+\lambda ') x_1)}-1.
\end{align}
From Part (2) of Lemma \ref{lem:Z4}, one gets
\begin{align*}
A(\beta^{\tau}, x_1) \overline{A(1, x_1)}=& \frac{1}{2}\left [ (-1)^{f_{1,\beta^{\tau},0}(x_1,0)}+(-1)^{f_{1,\beta^{\tau},0}(x_1,1)} \right ]\\
&- \left [ (-1)^{f_{1,\beta^{\tau},1}(x_1,0)}-(-1)^{f_{1,\beta^{\tau},1}(x_1,1)} \right ] \cdot \frac{\sqrt{-1}}{2}\\
=&\frac{1}{2} \sum_{x_2 \in \mathbb F_2} (-1)^{f_{1,\beta^{\tau},0}(x_1,x_2)} -  \frac{\sqrt{-1}}{2}\cdot \sum_{x_2 \in \mathbb F_2}  (-1)^{f_{1,\beta^{\tau},1}(x_1,x_2) + x_2}.
\end{align*}
From (\ref{eq:R-A-A}), one obtains
\begin{align*}
R_{s_{\lambda}, s_{\lambda'}}(\tau)=&-1+\frac{1}{2}\cdot \sum_{(x_1, x_2) \in \mathbb F_{2^{m-1}} \times \mathbb F_2 } (-1)^{f_{1,\beta^{\tau},0}(x_1,x_2)+ \tr^{m-1}_1 ((\lambda \beta^{\tau}+\lambda ') x_1)}\\
&-\frac{\sqrt{-1}}{2}\cdot \sum_{(x_1, x_2) \in \mathbb F_{2^{m-1}} \times \mathbb F_2 }  (-1)^{f_{1,\beta^{\tau},1}(x_1,x_2) + \tr^{m-1}_1 ((\lambda \beta^{\tau}+\lambda ') x_1)+x_2}\\
=&\frac{1}{2}\left [\mathcal W_{f_{1,\beta^{\tau},0}}(\lambda \beta^{\tau}+\lambda ',0)
 -\sqrt{-1} \mathcal W_{f_{1,\beta^{\tau},1}}(\lambda \beta^{\tau}+\lambda ',1) \right ]-1.
\end{align*}

(2) Using the definition of $\{s_{\lambda}(t)\}$ and $\{s_{\infty}(t)\}$, one has
\begin{align}
\label{R-lam-inf}
R_{s_{\lambda}, s_{\infty}}(\tau)=& \sum_{t=0}^{2^{m-1}-2} A(\beta^{\tau}, \beta^t)  (-1)^{\tr^{m-1}_1((\lambda \beta^{\tau}+1)\beta^t)}\nonumber\\
=& \sum_{x_1 \in \mathbb F_{2^{m-1}}} A(\beta^{\tau}, x_1 )  (-1)^{\tr^{m-1}_1((\lambda \beta^{\tau}+1)x_1)}-A(\beta^{\tau}, 0 ) \nonumber\\
=&\sum_{x_1 \in \mathbb F_{2^{m-1}}} A(\beta^{\tau}, x_1 )  (-1)^{\tr^{m-1}_1((\lambda \beta^{\tau}+1)x_1)}-1.
\end{align}
Note that
\begin{align*}
A(\beta^{\tau}, x_1 ) =&  \frac{1+\sqrt{-1}}{2} (-1)^{f(\beta^{\tau} x_1,0)}  +   \frac{1-\sqrt{-1}}{2} (-1)^{f(\beta^{\tau} x_1,1)}\\
=&  \frac{1}{2}  \sum_{x_2 \in \mathbb F_2}  (-1)^{f(\beta^{\tau} x_1,x_2)} +      \frac{\sqrt{-1}}{2}  \sum_{x_2 \in \mathbb F_2}  (-1)^{f(\beta^{\tau} x_1,x_2) + x_2}.
\end{align*}
By (\ref{R-lam-inf}), one obtains
\begin{align*}
R_{s_{\lambda}, s_{\infty}}(\tau)=&-1+  \frac{1}{2} \sum_{(x_1, x_2) \in \mathbb F_{2^{m-1}} \times \mathbb F_2 }
 (-1)^{f(\beta^{\tau} x_1,x_2)+\tr^{m-1}_1((\lambda \beta^{\tau}+1)x_1)}   \\
&    \frac{\sqrt{-1}}{2} \sum_{(x_1, x_2) \in \mathbb F_{2^{m-1}} \times \mathbb F_2 }
 (-1)^{f(\beta^{\tau} x_1,x_2)+\tr^{m-1}_1((\lambda \beta^{\tau}+1)x_1)+x_2}   \\
 =& \frac{1}{2} \left [ \mathcal W_{f_{\beta^{\tau}}}(\lambda \beta^{\tau}+1,0) +\sqrt{-1} \mathcal W_{f_{\beta^{\tau}}}(\lambda \beta^{\tau}+1,1) \right ]-1.
\end{align*}
For $R_{s_{\infty}, s_{\lambda}}(\tau)$, one has
\begin{align}
\label{eq:R-inf-lam}
R_{s_{\infty}, s_{\lambda}}(\tau)=& \sum_{t=0}^{2^{m-1}-2} \overline {A(1, \beta^t)}  (-1)^{\tr^{m-1}_1((\lambda + \beta^{\tau})\beta^t)}\nonumber\\
=& \sum_{x_1 \in \mathbb F_{2^{m-1}} }  \overline {A(1, x_1)}  (-1)^{\tr^{m-1}_1((\lambda + \beta^{\tau})x_1)} -\overline{A(1, 0 )} \nonumber\\
=&\sum_{x_1 \in \mathbb F_{2^{m-1}} }  \overline {A(1, x_1)}  (-1)^{\tr^{m-1}_1((\lambda + \beta^{\tau})x_1)} -1.
\end{align}
Note that
\begin{align*}
 \overline {A(1, x_1)}   =&  \frac{1-\sqrt{-1}}{2} (-1)^{f( x_1,0)}  +   \frac{1+\sqrt{-1}}{2} (-1)^{f( x_1,1)}\\
=&  \frac{1}{2}  \sum_{x_2 \in \mathbb F_2}  (-1)^{f( x_1,x_2)} -      \frac{\sqrt{-1}}{2}  \sum_{x_2 \in \mathbb F_2}  (-1)^{f( x_1,x_2) + x_2}.
\end{align*}
It follows from (\ref{eq:R-inf-lam}) that
\begin{align*}
R_{s_{\infty}, s_{\lambda}}(\tau)=&-1+  \frac{1}{2} \sum_{(x_1, x_2) \in \mathbb F_{2^{m-1}} \times \mathbb F_2 }
 (-1)^{f( x_1,x_2)+\tr^{m-1}_1((\lambda + \beta^{\tau})x_1)}   \\
&   - \frac{\sqrt{-1}}{2} \sum_{(x_1, x_2) \in \mathbb F_{2^{m-1}} \times \mathbb F_2 }
 (-1)^{f( x_1,x_2)+\tr^{m-1}_1((\lambda + \beta^{\tau})x_1)+x_2}   \\
 =& \frac{1}{2} \left [ \mathcal W_{f}(\lambda + \beta^{\tau},0) - \sqrt{-1} \mathcal W_{f}(\lambda + \beta^{\tau},1) \right ]-1.
\end{align*}

(3) For the sequences $\{s_{\infty}(t)\}$, one has
\begin{align*}
 R_{s_{\infty}, s_{\infty}}(\tau)=&\sum_{t=0}^{2^{m-1}-2} (-1)^{\tr^{m-1}_1((\beta^{\tau}+1)\beta^t)}\\
 =& \sum_{x_1\in \mathbb F_{2^{m-1}}} (-1)^{\tr^{m-1}_1((\beta^{\tau}+1)x_1)}-1\\
 = &\left\{
  \begin{array}{ll}
      -1,  & \text{if~~} \tau\neq 0, \\
    2^{m-1} -1,  & \text{if~~} \tau=0.
  \end{array}
\right.
 \end{align*}
This completes the proof.
\end{proof}

\begin{remark}
The condition $f(0,0)=f(0,1)=0$ is not a problem. For any cyclic bent function $f$, we can replace $f$ with $f'(x_1,x_2)=f(x_1, x_2)+f(0,x_2)$.
Then, $f'$ is still a cyclic bent function, and satisfies the condition $f'(0,0)=f'(0,1)=0$.
\end{remark}

\begin{proposition}
Let $f$ be a cyclic bent function on $\mathbb F_{2^{m-1}}\times \mathbb F_2$ with $f(0,0)=f(0,1)=0$.
Then, for any $\lambda, \lambda' \in \mathbb F_{2^{m-1}} \cup \{\infty\}$ and $0 \le \tau < 2^{m-1}-1$,
the correlation $R_{s_{\lambda}, s_{\lambda}}(\tau)$ of the sequences $\{s_{\lambda}(t)\}$ and $\{s_{\lambda'}(t)\}$ at shift $\tau$  belongs to the set
$\left \{-1, -1+2^{m-1}, -1+\left (\pm 1 \pm \sqrt{-1} \right )2^{\frac{m-2}{2}} \right \}$.
\end{proposition}

\begin{proof}
By Lemma \ref{Rss},
we just need to prove that $R_{s_{\lambda}, s_{\lambda'}}(0)
\in \left \{-1, -1+2^{m-1}, \left (\pm 1 \pm \sqrt{-1} \right )2^{\frac{m-2}{2}} \right \}$
when $\lambda,\lambda'\in
\mathbb{F}_{2^{m-1}}$ and
$\tau=0$.
Note that
$$
R_{s_{\lambda}, s_{\lambda'}}(0)= \sum_{t=0}^{2^{m-1}-2} A(1, \beta^t) \overline{A(1, \beta^t)} (-1)^{\tr^{m-1}_1 ((\lambda +\lambda ') \beta^t)},
$$
where $A(1, \beta^t)=
\frac{1+\sqrt{-1}}{2}(-1)^{f(\beta^t,0)}
+ \frac{1-\sqrt{-1}}{2}(-1)^{f(\beta^t,1)}$.
By Part (1) of Lemma \ref{lem:Z4},
$A(1, \beta^t) \overline{A(1, \beta^t)}=1$, we have
$$
R_{s_{\lambda}, s_{\lambda'}}(0)= \sum_{t=0}^{2^{m-1}-2}   (-1)^{\tr^{m-1}_1 ((\lambda +\lambda ') \beta^t)}
=
\left\{
  \begin{array}{ll}
    -1, & \hbox{if $\lambda\neq \lambda'$,} \\
    2^{m-1}-1, & \hbox{if
$\lambda=\lambda'$.}
  \end{array}
\right.
$$
This completes the proof.
\end{proof}

\begin{lemma}\label{w1b}
Let $b\in \mathbb F_{2^{m-1}}$ with $b\neq 1$. Then
\begin{align*}
\sum_{\mu \in \mathbb F_{2^{m-1}}}\mathcal  W_{f_{1,b,0}}(\mu,0) =2^m,~~~~\sum_{\mu \in \mathbb F_{2^{m-1}}}\mathcal  W_{f_{1,b,1}}(\mu,1) =0,
\end{align*}
and
\begin{align*}
\sum_{\mu \in \mathbb F_{2^{m-1}}}\mathcal  W_{f_{1,b,0}}(\mu,0)  \mathcal  W_{f_{1,b,1}}(\mu,1)=0.
\end{align*}

\end{lemma}

\begin{proof}
First of all, we have
\begin{align*}
\sum_{\mu \in \mathbb F_{2^{m-1}}}\mathcal  W_{f_{1,b,0}}(\mu,0) = & \sum_{\mu \in \mathbb F_{2^{m-1}}}  \sum_{(x_1,x_2) \in \mathbb F_{2^{m-1}}\times \mathbb F_2 }   (-1)^{f_{1,b,0}(x_1,x_2)+ \tr^{m-1}_1(\mu x_1)}\\
=& \sum_{(x_1,x_2) \in \mathbb F_{2^{m-1}} \times \mathbb F_2}    (-1)^{f_{1,b,0}(x_1,x_2)}  \sum_{\mu \in \mathbb F_{2^{m-1}}}   (-1)^{\tr^{m-1}_1(\mu x_1)}\\
=&2^{m-1} \sum_{x_2\in \mathbb F_2}   (-1)^{f_{1,b,0}(0,x_2)} \\
=& 2^{m}.
\end{align*}

Then we have
\begin{align*}
\sum_{\mu \in \mathbb F_{2^{m-1}}}\mathcal  W_{f_{1,b,1}}(\mu,1) = & \sum_{\mu \in \mathbb F_{2^{m-1}}}
\sum_{(x_1,x_2) \in \mathbb F_{2^{m-1}}\times \mathbb F_2 }   (-1)^{f_{1,b,1}(x_1,x_2)+ \tr^{m-1}_1(\mu x_1)+x_2}\\
=& \sum_{(x_1,x_2) \in \mathbb F_{2^{m-1}} \times \mathbb F_2}    (-1)^{f_{1,b,1}(x_1,x_2)+x_2}  \sum_{\mu \in \mathbb F_{2^{m-1}}}   (-1)^{\tr^{m-1}_1(\mu x_1)}\\
=&2^{m-1} \sum_{x_2\in \mathbb F_2}   (-1)^{f_{1,b,1}(0,x_2)+x_2} \\
=& 2^{m-1} \sum_{x_2\in \mathbb F_2}   (-1)^{f(0,0)+f(0,1)+x_2}\\
=& 0.
\end{align*}

Finally, let $A= \sum_{\mu \in \mathbb F_{2^{m-1}}}\mathcal  W_{f_{1,b,0}}(\mu,0)  \mathcal  W_{f_{1,b,1}}(\mu,1)$. Then,
\begin{align}\label{eq:WW-x2-y2}
A=& \sum_{\mu \in \mathbb F_{2^{m-1}}} \left (\sum_{(x_1,x_2) \in \mathbb F_{2^{m-1}} \times \mathbb F_{2}} (-1)^{f_{1,b,0}(x_1,x_2)+ \tr^{m-1}(\mu x_1)} \right ) \nonumber \\
~~~~~~~&~~~~~~~~~\cdot \left (\sum_{(y_1,y_2) \in \mathbb F_{2^{m-1}} \times \mathbb  F_{2}} (-1)^{f_{1,b,1}(y_1,y_2)+ \tr^{m-1}(\mu y_1) +y_2} \right )\nonumber \\
=& \sum_{\substack{(x_1,x_2) \in  \mathbb F_{2^{m-1}} \times \mathbb  F_{2} \\(y_1, y_2) \in  \mathbb F_{2^{m-1}} \times \mathbb  F_{2} }}
(-1)^{f_{1,b,0}(x_1,x_2)+ f_{1,b,1}(y_1,y_2)+ y_2}  \cdot  \sum_{\mu \in \mathbb F_{2^{m-1}}}  (-1)^{\tr^{m-1}(\mu (x_1+y_1))}\nonumber\\
=& 2^{m-1}  \sum_{x_1 \in \mathbb F_{2^{m-1}} }  \sum_{(x_2,y_2)\in \mathbb F_2\times \mathbb F_2}        (-1)^{f_{1,b,0}(x_1,x_2)+ f_{1,b,1}(x_1,y_2)+ y_2}.
\end{align}
Note that $f_{1,b,0}(x_1,x_2)+ f_{1,b,1}(x_1,y_2)=f(x_1,x_2)+f(x_1,y_2)+
f(bx_1,x_2)+f(bx_1,y_2+1)$. One has  $$f_{1,b,0}(x_1,x_2)+ f_{1,b,1}(x_1,y_2)=\begin{cases}
f(bx_1,x_2)+f(bx_1,y_2+1), &\text{if }x_2=y_2, \\
f(x_1,x_2)+f(x_1,y_2), &\text{if } x_2=y_2+1.
\end{cases}
$$
That is,
$$f_{1,b,0}(x_1,x_2)+ f_{1,b,1}(x_1,y_2)=\begin{cases}
f(bx_1,0)+f(bx_1,1), &\text{if }x_2=y_2, \\
f(x_1,0)+f(x_1,1), &\text{if } x_2=y_2+1.
\end{cases}$$
Consequently, for any $x_1\in \mathbb F_{2^{m-1}}$,
\begin{align*}
&\sum_{(x_2,y_2)\in \mathbb F_2\times \mathbb F_2}        (-1)^{f_{1,b,0}(x_1,x_2)+ f_{1,b,1}(x_1,y_2)+ y_2}\\
= & \sum_{(x_2,y_2)\in \{(0,0), (1,1)\}}  (-1)^{f(bx_1,0)+f(bx_1,1)+y_2}\\
& +\sum_{(x_2,y_2)\in \{(0,1), (1,0)\}}  (-1)^{f(x_1,0)+f(x_1,1)+y_2}\\
= & (-1)^{f(bx_1,0)+f(bx_1,1)} \sum_{y_2\in \mathbb F_2} (-1)^{y_2}\\
&+ (-1)^{f(x_1,0)+f(x_1,1)} \sum_{y_2\in \mathbb F_2} (-1)^{y_2}\\
=& 0.
\end{align*}
It then follows from (\ref{eq:WW-x2-y2}) that $A=0$, which completes the proof.
\end{proof}

\begin{lemma}\label{wfb}
Let $f$ be a cyclic bent function with $f(0,0)=f(0,1)=0$ and $b\in \mathbb F_{2^{m-1}}^{\star}$. Then
\begin{align*}
\sum_{\mu \in \mathbb F_{2^{m-1}}}\mathcal  W_{f_{b}}(\mu,0) =2^m,~~~~\sum_{\mu \in \mathbb F_{2^{m-1}}}\mathcal  W_{f_{b}}(\mu,1) =0,
\end{align*}
and
\begin{align*}
\sum_{\mu \in \mathbb F_{2^{m-1}}}\mathcal  W_{f_{b}}(\mu,0)  \mathcal  W_{f_{b}}(\mu,1)=0.
\end{align*}

\end{lemma}

\begin{proof}
First of all, we have
\begin{align*}
\sum_{\mu \in \mathbb F_{2^{m-1}}}\mathcal  W_{f_{b}}(\mu,0) = & \sum_{\mu \in \mathbb F_{2^{m-1}}}  \sum_{(x_1,x_2)
\in \mathbb F_{2^{m-1}}\times \mathbb F_2 }   (-1)^{f_{b}(x_1,x_2)+ \tr^{m-1}_1(\mu x_1)}\\
=& \sum_{(x_1,x_2) \in \mathbb F_{2^{m-1}} \times \mathbb F_2}    (-1)^{f_{b}(x_1,x_2)}  \sum_{\mu \in \mathbb F_{2^{m-1}}}   (-1)^{\tr^{m-1}_1(\mu x_1)}\\
=&2^{m-1} \sum_{x_2\in \mathbb F_2}   (-1)^{f_{b}(0,x_2)} \\
=& 2^{m-1}\left ((-1)^{f(0,0)}+(-1)^{f(0,1)} \right )\\
=& 2^{m}.
\end{align*}

Then we have
\begin{align*}
\sum_{\mu \in \mathbb F_{2^{m-1}}}\mathcal  W_{f_{b}}(\mu,1) = & \sum_{\mu \in \mathbb F_{2^{m-1}}}
\sum_{(x_1,x_2) \in \mathbb F_{2^{m-1}}\times \mathbb F_2 }   (-1)^{f_{b}(x_1,x_2)+ \tr^{m-1}_1(\mu x_1)+x_2}\\
=& \sum_{(x_1,x_2) \in \mathbb F_{2^{m-1}} \times \mathbb F_2}    (-1)^{f_{b}(x_1,x_2)+x_2}  \sum_{\mu \in \mathbb F_{2^{m-1}}}   (-1)^{\tr^{m-1}_1(\mu x_1)}\\
=&2^{m-1} \sum_{x_2\in \mathbb F_2}   (-1)^{f_{b}(0,x_2)+x_2} \\
=& 2^{m-1} \left (  (-1)^{f(0,0)}-(-1)^{f(0,1)} \right )\\
=& 0.
\end{align*}

Finally, let $A= \sum_{\mu \in \mathbb F_{2^{m-1}}}\mathcal  W_{f_{b}}(\mu,0)  \mathcal  W_{f_{b}}(\mu,1)$. Then,
\begin{align}\label{eq:WW-x2-y2-fb}
A=& \sum_{\mu \in \mathbb F_{2^{m-1}}} \left (\sum_{(x_1,x_2) \in \mathbb F_{2^{m-1}} \times \mathbb F_{2}} (-1)^{f_{b}(x_1,x_2)+ \tr^{m-1}(\mu x_1)} \right ) \nonumber \\
~~~~~~~&~~~~~~~~~\cdot \left (\sum_{(y_1,y_2) \in \mathbb F_{2^{m-1}} \times \mathbb  F_{2}} (-1)^{f_{b}(y_1,y_2)+ \tr^{m-1}(\mu y_1) +y_2} \right )\nonumber \\
=& \sum_{\substack{(x_1,x_2) \in  \mathbb F_{2^{m-1}} \times \mathbb  F_{2} \\(y_1, y_2) \in  \mathbb F_{2^{m-1}} \times \mathbb  F_{2} }}
(-1)^{f_{b}(x_1,x_2)+ f_{b}(y_1,y_2)+ y_2}  \cdot  \sum_{\mu \in \mathbb F_{2^{m-1}}}  (-1)^{\tr^{m-1}(\mu (x_1+y_1))}\nonumber\\
=& 2^{m-1}  \sum_{x_1 \in \mathbb F_{2^{m-1}} }  \sum_{(x_2,y_2)\in \mathbb F_2\times \mathbb F_2}        (-1)^{f_{b}(x_1,x_2)+ f_{b}(x_1,y_2)+ y_2}.
\end{align}
Note that, for any $x_1\in \mathbb F_{2^{m-1}}$,
\begin{align*}
&\sum_{(x_2,y_2)\in \mathbb F_2\times \mathbb F_2}        (-1)^{f_{b}(x_1,x_2)+ f_{b}(x_1,y_2)+ y_2}\\
= & \sum_{(x_2,y_2)\in \{(0,0), (1,1)\}}  (-1)^{f(bx_1,x_2)+f(bx_1,x_2)+y_2}\\
& +\sum_{(x_2,y_2)\in \{(0,1), (1,0)\}}  (-1)^{f(bx_1,x_2)+f(bx_1,x_2+1)+y_2}\\
= &  \sum_{y_2\in \mathbb F_2} (-1)^{y_2}+ (-1)^{f(bx_1,0)+f(bx_1,1)} \sum_{y_2\in \mathbb F_2} (-1)^{y_2}\\
=& 0.
\end{align*}
By (\ref{eq:WW-x2-y2-fb}), we have $A=0$, which completes the proof.
\end{proof}

For any bent functions $g$ and $h$ on $\mathbb F_{2^{m-1}} \times \mathbb F_2$, define
\begin{align}\label{eq:J-W-W-dis}
J_{g,h}(\epsilon_1, \epsilon_2)=\{\mu \in \mathbb F_{2^{m-1}} | \mathcal  W_{g}(\mu,0)=(-1)^{\epsilon_1} 2^{\frac{m}{2}} , \mathcal  W_{h}(\mu,1)=(-1)^{\epsilon_2} 2^{\frac{m}{2}}\}.
\end{align}

\begin{lemma}\label{wgh}
Let $g$ and $h$ be two bent functions on $\mathbb F_{2^{m-1}} \times \mathbb F_2$  such that \\
(1) $\sum_{\mu \in \mathbb F_{2^{m-1}}}\mathcal  W_{g}(\mu,0) =2^m$;\\
(2) $\sum_{\mu \in \mathbb F_{2^{m-1}}}\mathcal  W_{h}(\mu,1) =0$; and \\
(3) $\sum_{\mu \in \mathbb F_{2^{m-1}}}   W_{g}(\mu,0) \mathcal  W_{h}(\mu,1) =0$.\\
Then, $\# J_{g,h}(\epsilon_1, \epsilon_2)= 2^{m-3} + (-1)^{\epsilon_1} 2^{\frac{m-4}{2}} $, where $\epsilon_1, \epsilon_2 \in \mathbb F_2$.
\end{lemma}

\begin{proof}
From the definition of  $J_{g,h}(\epsilon_1, \epsilon_2)$ in   (\ref{eq:J-W-W-dis}), one gets
\begin{align*}
\# J_{g,h}(\epsilon_1, \epsilon_2)=& \sum_{\mu \in \mathbb F_{2^{m-1}}}  \frac{(-1)^{\epsilon_1 + \epsilon_2}}{2^{m+2}} \left ( \mathcal  W_{g}(\mu,0)+ (-1)^{\epsilon_1} 2^{\frac{m}{2}} \right )
\cdot \left ( \mathcal  W_{h}(\mu,1)+ (-1)^{\epsilon_2} 2^{\frac{m}{2}} \right )\\
=&  \frac{(-1)^{\epsilon_1 + \epsilon_2}}{2^{m+2}}  \sum_{\mu \in \mathbb F_{2^{m-1}}} (-1)^{\epsilon_1 + \epsilon_2} 2^m +
 \frac{(-1)^{\epsilon_1 + \epsilon_2}}{2^{m+2}}  \sum_{\mu \in \mathbb F_{2^{m-1}}} \mathcal W_{g}(\mu,0) \mathcal  W_{h}(\mu,1)\\
 &+ \frac{(-1)^{\epsilon_2 }}{2^{m+2}} 2^{\frac{m}{2}}  \sum_{\mu \in \mathbb F_{2^{m-1}}} \mathcal  W_{h}(\mu,1)+
 \frac{(-1)^{\epsilon_1}}{2^{m+2}} 2^{\frac{m}{2}}  \sum_{\mu \in \mathbb F_{2^{m-1}}} \mathcal  W_{g}(\mu,0)\\
 =& 2^{m-3} + (-1)^{\epsilon_1} 2^{\frac{m-4}{2}} .
\end{align*}
This completes the proof.
\end{proof}

\begin{theorem}\label{thm:U-f}
Let $f$ be a cyclic bent function on $\mathbb F_{2^{m-1}} \times \mathbb F_2$ with $f(0,0)=f(0,1)=0$. Then, the correlation value distribution of
the family $\mathcal U_f$ is given in
Table \ref{table-Uf}.
\begin{table}[htbp]
\centering
\caption{The correlation distribution of the family $\mathcal U_f$}
\label{table-Uf}
\begin{tabular}{|c|c|}
  \hline
  Value & Frequency \\
  \hline
  $-1+2^{m-1}$& $2^{m-1}+1$\\
\hline
$-1$& $2^{2m-2}-2$\\
\hline
$-1+(1+\sqrt{-1})2^\frac{m-2}{2}$&
 $(2^{2m-2}-2)(2^{m-3}+
2^{\frac{m-4}{2}})$\\
\hline
$-1+(1-\sqrt{-1})2^\frac{m-2}{2}$&
 $(2^{2m-2}-2)(2^{m-3}+
2^{\frac{m-4}{2}})$\\
\hline
$-1+(-1+\sqrt{-1})2^\frac{m-2}{2}$&
 $(2^{2m-2}-2)(2^{m-3}-
2^{\frac{m-4}{2}})$\\
\hline
$-1+(-1-\sqrt{-1})2^\frac{m-2}{2}$&
 $(2^{2m-2}-2)(2^{m-3}-
2^{\frac{m-4}{2}})$\\
\hline
\end{tabular}
\end{table}
\end{theorem}

\begin{proof}
For any $\lambda,\lambda'\in
\mathbb{F}_{2^{m-1}}$, we have
\begin{equation}\label{rss0}
R_{s_{\lambda}, s_{\lambda'}}(0)=
\left\{
  \begin{array}{ll}
    -1, & \hbox{if $\lambda\neq \lambda'$,} \\
    2^{m-1}-1, & \hbox{if
$\lambda=\lambda'$.}
  \end{array}
\right.
\end{equation}
We will discuss the correlation value distribution
of $R_{s_{\lambda}, s_{\lambda'}}(\tau)$, where
$\lambda,\lambda\in \mathbb{F}_{2^{m-1}}
\cup \{\infty\}$ and
$0\leq \tau   < 2^{m-1}-1$, by distinguishing among the following cases.

Case 0.  $R_{s_{\lambda}, s_{\lambda'}}(\tau)
=2^{m-1}-1$. Let $D_0$ be the set
of $(\lambda,\lambda',\tau)$ satisfying
$R_{s_{\lambda}, s_{\lambda'}}(\tau)
=2^{m-1}-1$, where
$\lambda,\lambda\in \mathbb{F}_{2^{m-1}}
\cup \{\infty\}$ and
$0\leq \tau   < 2^{m-1}-1$.
From Lemma \ref{Rss} and (\ref{rss0}), we have
$$
D_0=\{(\lambda,\lambda,0)|
\lambda \in \mathbb{F}_{2^{m-1}}
\cup \{\infty\}\}.
$$
Hence, $\#D_0=2^{m-1}+1$.

Case 1. $R_{s_{\lambda}, s_{\lambda'}}(\tau)
=-1$. Let $D_1$ be the set
of $(\lambda,\lambda',\tau)$ satisfying
$R_{s_{\lambda}, s_{\lambda'}}(\tau)
=-1$. From Lemma \ref{Rss} and (\ref{rss0}), we have
$$
D_1=\{(\infty,\infty,\tau)|
1\leq \tau < 2^{m-1}-1
 \}\cup
\{(\lambda,\lambda',0) |
\lambda\neq \lambda'\in
\mathbb{F}_{2^{m-1}}\}.
$$
Hence, $\#D_1=2^{m-1}-2+
2^{m-1}(2^{m-1}-1)
=2^{2m-2}-2$.

Case 2. $R_{s_{\lambda}, s_{\lambda'}}(\tau)
=-1+(1+\sqrt{-1})2^{
\frac{m-2}{2}}$. Let $D_2$ be the set
of $(\lambda,\lambda',\tau)$ satisfying
$R_{s_{\lambda}, s_{\lambda'}}(\tau)
=-1+(1+\sqrt{-1})2^{
\frac{m-2}{2}}$. From Lemma \ref{Rss} and (\ref{rss0}),
$(\lambda,\lambda',\tau)
\in D_2$ if and only if $(\lambda,
\lambda',\tau)$ satisfies one of the following
three conditions:
\begin{equation}\label{case21}
\left\{
  \begin{array}{l}
  \mathcal{W}_{f_1,\beta^\tau,0}(
\lambda \beta^\tau+\lambda',0)
=(-1)^02^{\frac{m}{2}},  \\
    \mathcal{W}_{f_1,\beta^\tau,1}(
\lambda \beta^\tau+\lambda',1)
=(-1)^1 2^{\frac{m}{2}},  \\
    (\lambda, \lambda')\in
\mathbb{F}_{2^{m-1}}\times \mathbb{F}_{2^{m-1}}, 1\le \tau <2^{m-1}-1,
  \end{array}
\right.
\end{equation}
\begin{equation}\label{case22}
\left\{
  \begin{array}{l}
  \mathcal{W}_{f_{\beta^\tau}}(
\lambda \beta^\tau+1,0)
=(-1)^0 2^{\frac{m}{2}},  \\
    \mathcal{W}_{f_{\beta^\tau}}(
\lambda \beta^\tau+1,1)
=(-1)^0 2^{\frac{m}{2}},  \\
    \lambda\in
\mathbb{F}_{2^{m-1}},\lambda'
=\infty, 0\leq\tau < 2^{m-1}-1,
  \end{array}
\right.
\end{equation}
\begin{equation}\label{case23}
\left\{
  \begin{array}{l}
  \mathcal{W}_{f}(
\lambda'+ \beta^\tau,0)
=(-1)^0 2^{\frac{m}{2}},  \\
    \mathcal{W}_{f}(
\lambda'+ \beta^\tau,1)
=(-1)^1 2^{\frac{m}{2}},  \\
    \lambda=\infty,\lambda'\in
\mathbb{F}_{2^{m-1}},0\leq
\tau<2^{m-1}-1.
  \end{array}
\right.
\end{equation}
Let $D_{2,1}, D_{2,2}, D_{2,3}$ be
sets of $(\lambda,\lambda',
\tau)$ satisfying Conditions
(\ref{case21}), (\ref{case22}), and
(\ref{case23}) respectively.
For the set $D_{2,1}$, we have
$(\lambda,\lambda',\tau)
\in D_{2,1}$ if and only if
$$
\left\{
  \begin{array}{l}
   \tau\neq 0, \lambda,\lambda'\in
\mathbb{F}_{2^{m-1}}  ,  \\
    \lambda\beta^\tau+\lambda'
\in J_{g,h}(0, 1),
  \end{array}
\right.
$$
where $g=f_{1,\beta^{\tau},0}$ and
$h=f_{1,\beta^\tau,1}$. This is equivalent to
$$
\left\{
  \begin{array}{l}
   \tau\neq 0, \lambda,\lambda'\in
\mathbb{F}_{2^{m-1}}  ,  \\
    \lambda'\in -\lambda\beta^\tau+
  J_{g,h}(0, 1),
  \end{array}
\right.
$$
Then
\begin{align*}
\#D_{2,1}=&  \sum_{\tau=1}^{
2^{m-1}-2} \sum_{\lambda\in \mathbb{F}_{2^{m-1}}} \#(-\lambda
\beta^\tau+J_{g,h}(0,1))\\
=&\sum_{\tau=1}^{
2^{m-1}-2} \sum_{\lambda\in \mathbb{F}_{2^{m-1}}} \#J_{g,h}(0,1)\\
=& 2^{m-1} \sum_{\tau=1}^{
2^{m-1}-2}  \#J_{g,h}(0,1).
\end{align*}
From Lemmas \ref{w1b} and
\ref{wgh}, we have
$$
\#D_{2,1}=2^{m-1}(2^{m-1}-2)
(2^{m-3}+2^{\frac{m-4}{2}}).
$$
For the set $D_{2,2}$, we have
$(\lambda,\lambda',\tau)
\in D_{2,2}$ if and only if
$$
\left\{
  \begin{array}{l}
   0\leq \tau <2^{m-1}-1, \lambda'
=\infty,\lambda\in
\mathbb{F}_{2^{m-1}}  ,  \\
    \lambda\beta^\tau+1
\in J_{f_{\beta^\tau},f_{\beta^\tau}}(0, 0).
  \end{array}
\right.
$$
This is equivalent to
$$
\left\{
  \begin{array}{l}
   0\leq \tau <2^{m-1}-1, \lambda'
=\infty,\lambda\in
\mathbb{F}_{2^{m-1}},  \\
    \lambda\in \beta^{-\tau}(
  J_{f_{\beta^\tau},f_{\beta^\tau}}(0, 0)-1).
  \end{array}
\right.
$$
Then
\begin{align*}
\#D_{2,2}=&\sum_{\tau=0}^{
2^{m-1}-2} \#\left( \beta^{-\tau}
\left( J_{f_{\beta^\tau},f_{\beta^\tau}}(0,0)
-1\right)\right)\\
=& \sum_{\tau=0}^{
2^{m-1}-2} \#J_{f_{\beta^\tau},f_{\beta^\tau}}(0,0).
\end{align*}
By Lemmas \ref{wfb} and \ref{wgh}, we have
$$
\#D_{2,2}=(2^{m-1}-1)
(2^{m-3}+2^{\frac{m-4}{2}}).
$$
For the set $D_{2,3}$, we have
$(\lambda,\lambda',\tau)
\in D_{2,3}$ if and only if
$$
\left\{
  \begin{array}{l}
   0\leq \tau <2^{m-1}-1, \lambda
=\infty,\lambda'\in
\mathbb{F}_{2^{m-1}},  \\
    \lambda'+\beta^\tau
\in J_{f,f}(0, 1).
  \end{array}
\right.
$$
This is equivalent to
$$
\left\{
  \begin{array}{l}
   0\leq \tau <2^{m-1}-1, \lambda
=\infty,\lambda'\in
\mathbb{F}_{2^{m-1}},  \\
    \lambda'\in
  J_{f,f}(0, 0)-\beta^{\tau}.
  \end{array}
\right.
$$
Then
\begin{align*}
\#D_{2,3}=&\sum_{\tau=0}^{
2^{m-1}-2} (\# J_{f,f}(0,1)
-\beta^\tau ) \\
=& \sum_{\tau=0}^{
2^{m-1}-2} \#J_{f,f}(0,1)\\
=& (2^{m-1}-1)\#J_{f,f}(0,1).
\end{align*}
It follows from Lemmas \ref{wfb} and
\ref{wgh} that
$$
\#D_{2,3}=(2^{m-1}-1)
(2^{m-3}+2^{\frac{m-4}{2}}).
$$
As a result, we have
\begin{align*}
\#D_{2}=&\#D_{2,1}+\#D_{2,2}+\#D_{2,3}\\
=& (2^{2m-2}-2)(2^{m-3}+2^{\frac{m-4}{2}}).
\end{align*}

Case 3.
$R_{s_{\lambda}, s_{\lambda'}}(\tau)
=-1+(-1+\sqrt{-1})2^{
\frac{m-2}{2}}$. Let $D_3$ be the set
of $(\lambda,\lambda',\tau)$ satisfying
$R_{s_{\lambda}, s_{\lambda'}}(\tau)
=-1+(-1+\sqrt{-1})2^{
\frac{m-2}{2}}$. By Lemma \ref{Rss} and (\ref{rss0}),
$(\lambda,\lambda',\tau)
\in D_3$ if and only if $(\lambda,
\lambda',\tau)$ satisfies one of the following
three conditions:
\begin{equation}\label{case31}
\left\{
  \begin{array}{l}
  \mathcal{W}_{f_1,\beta^\tau,0}(
\lambda \beta^\tau+\lambda',0)
=(-1)^1 2^{\frac{m}{2}},  \\
    \mathcal{W}_{f_1,\beta^\tau,1}(
\lambda \beta^\tau+\lambda',1)
=(-1)^1 2^{\frac{m}{2}},  \\
    \lambda,\lambda'\in
\mathbb{F}_{2^{m-1}},\tau\neq 0,
  \end{array}
\right.
\end{equation}
\begin{equation}\label{case32}
\left\{
  \begin{array}{l}
  \mathcal{W}_{f_{\beta^\tau}}(
\lambda \beta^\tau+1,0)
=(-1)^1 2^{\frac{m}{2}},  \\
    \mathcal{W}_{f_{\beta^\tau}}(
\lambda \beta^\tau+1,1)
=(-1)^0 2^{\frac{m}{2}},  \\
    \lambda\in
\mathbb{F}_{2^{m-1}},\lambda'
=\infty, 0\leq\tau < 2^{m-1}-1,
  \end{array}
\right.
\end{equation}
\begin{equation}\label{case33}
\left\{
  \begin{array}{l}
  \mathcal{W}_{f}(
\lambda'+ \beta^\tau,0)
=(-1)^1 2^{\frac{m}{2}},  \\
    \mathcal{W}_{f}(
\lambda'+ \beta^\tau,1)
=(-1)^1 2^{\frac{m}{2}},  \\
    \lambda=\infty,\lambda'\in
\mathbb{F}_{2^{m-1}},0\leq
\tau<2^{m-1}-1.
  \end{array}
\right.
\end{equation}
Let $D_{3,1}, D_{3,2}, D_{3,3}$ be
sets of $(\lambda,\lambda',
\tau)$ satisfying Conditions
(\ref{case31}), (\ref{case32}), and
(\ref{case33}) respectively.
With similar discussion of Case 2, we have
\begin{align*}
&\#D_{3,1}=(2^{m-1}-2)2^{m-1}
(2^{m-3}-2^{\frac{m-4}{2}}),\\
&\#D_{3,2}=(2^{m-1}-1)(2^{m-3}
-2^{\frac{m-4}{2}}),\\
&\#D_{3,3}=(2^{m-1}-1)(2^{m-3}
-2^{\frac{m-4}{2}}).
\end{align*}
Then
\begin{align*}
\#D_{3}=& \#D_{3,1}+\#D_{3,2}
+\#D_{3,3}
\\=&(2^{2m-2}-2)(2^{m-3}
-2^{\frac{m-4}{2}}).
\end{align*}

Case 4.
$R_{s_{\lambda}, s_{\lambda'}}(\tau)
=-1+(1-\sqrt{-1})2^{
\frac{m-2}{2}}$. Let $D_4$ be the set
of $(\lambda,\lambda',\tau)$ satisfying
$R_{s_{\lambda}, s_{\lambda'}}(\tau)
=-1+(1-\sqrt{-1})2^{
\frac{m-2}{2}}$. By Lemma \ref{Rss} and (\ref{rss0}),
$(\lambda,\lambda',\tau)
\in D_4$ if and only if $(\lambda,
\lambda',\tau)$ satisfies one of the following
three conditions:
\begin{equation*}
\left\{
  \begin{array}{l}
  \mathcal{W}_{f_1,\beta^\tau,0}(
\lambda \beta^\tau+\lambda',0)
=(-1)^0 2^{\frac{m}{2}},  \\
    \mathcal{W}_{f_1,\beta^\tau,1}(
\lambda \beta^\tau+\lambda',1)
=(-1)^0 2^{\frac{m}{2}},  \\
    \lambda,\lambda'\in
\mathbb{F}_{2^{m-1}},\tau\neq 0,
  \end{array}
\right.
\end{equation*}
\begin{equation*}
\left\{
  \begin{array}{l}
  \mathcal{W}_{f_{\beta^\tau}}(
\lambda \beta^\tau+1,0)
=(-1)^0 2^{\frac{m}{2}},  \\
    \mathcal{W}_{f_{\beta^\tau}}(
\lambda \beta^\tau+1,1)
=(-1)^1 2^{\frac{m}{2}},  \\
    \lambda\in
\mathbb{F}_{2^{m-1}},\lambda'
=\infty, 0\leq\tau < 2^{m-1}-1,
  \end{array}
\right.
\end{equation*}
\begin{equation*}
\left\{
  \begin{array}{l}
  \mathcal{W}_{f}(
\lambda'+ \beta^\tau,0)
=(-1)^0 2^{\frac{m}{2}},  \\
    \mathcal{W}_{f}(
\lambda'+ \beta^\tau,1)
=(-1)^0 2^{\frac{m}{2}},  \\
    \lambda=\infty,\lambda'\in
\mathbb{F}_{2^{m-1}},0\leq
\tau<2^{m-1}-1.
  \end{array}
\right.
\end{equation*}
With similar discussion of Case 2, we have
\begin{align*}
\#D_{4}=& \#D_{2}=(2^{2m-2}-2)(2^{m-3}
+2^{\frac{m-4}{2}}).
\end{align*}

Case 5.
$R_{s_{\lambda}, s_{\lambda'}}(\tau)
=-1+(-1-\sqrt{-1})2^{
\frac{m-2}{2}}$. Let $D_5$ be the set
of $(\lambda,\lambda',\tau)$ satisfying
$R_{s_{\lambda}, s_{\lambda'}}(\tau)
=-1+(-1-\sqrt{-1})2^{
\frac{m-2}{2}}$. By Lemma \ref{Rss} and (\ref{rss0}),
$(\lambda,\lambda',\tau)
\in D_5$ if and only if $(\lambda,
\lambda',\tau)$ satisfies one of the following
three conditions:
\begin{equation*}
\left\{
  \begin{array}{l}
  \mathcal{W}_{f_1,\beta^\tau,0}(
\lambda \beta^\tau+\lambda',0)
=(-1)^1 2^{\frac{m}{2}},  \\
    \mathcal{W}_{f_1,\beta^\tau,1}(
\lambda \beta^\tau+\lambda',1)
=(-1)^0 2^{\frac{m}{2}},  \\
    \lambda,\lambda'\in
\mathbb{F}_{2^{m-1}},\tau\neq 0,
  \end{array}
\right.
\end{equation*}
\begin{equation*}
\left\{
  \begin{array}{l}
  \mathcal{W}_{f_{\beta^\tau}}(
\lambda \beta^\tau+1,0)
=(-1)^1 2^{\frac{m}{2}},  \\
    \mathcal{W}_{f_{\beta^\tau}}(
\lambda \beta^\tau+1,1)
=(-1)^1 2^{\frac{m}{2}},  \\
    \lambda\in
\mathbb{F}_{2^{m-1}},\lambda'
=\infty, 0\leq\tau < 2^{m-1}-1,
  \end{array}
\right.
\end{equation*}
\begin{equation*}
\left\{
  \begin{array}{l}
  \mathcal{W}_{f}(
\lambda'+ \beta^\tau,0)
=(-1)^1 2^{\frac{m}{2}},  \\
    \mathcal{W}_{f}(
\lambda'+ \beta^\tau,1)
=(-1)^0 2^{\frac{m}{2}},  \\
    \lambda=\infty,\lambda'\in
\mathbb{F}_{2^{m-1}},0\leq
\tau<2^{m-1}-1.
  \end{array}
\right.
\end{equation*}
With similar discussion of Case 3, we have
\begin{align*}
\#D_{5}=& \#D_{3}=(2^{2m-2}-2)(2^{m-3}
-2^{\frac{m-4}{2}}).
\end{align*}
This completes the proof.
\end{proof}

\begin{remark}
For any cyclic bent function $f$ on $\mathbb F_{2^{m-1}}\times \mathbb F_2$, according to
Theorem \ref{thm:U-f} and Reference \cite{BHK92}, the quaternary family $\mathcal U_f$
has the same period $K=2^{m-1}-1$, family size $N=2^{m-1}+1$,  $R_{\mathrm{max}}\le 1+\sqrt{2^{m-1}}$ and
correlation value distribution as the well-known quaternary  family $\mathcal A$ of
  sequences defined in \cite{BHK92}, which offers a lower value of $R_{\mathrm{max}}$ for the same period and family size in comparison with the binary family of Gold sequences
and is asymptotically optimal with respect to the Welch lower bound for complex-valued sequences.
\end{remark}

\subsection{A binary sequence family from cyclic bent functions}

Let $f$ be a cyclic bent function on $\mathbb F_{2^{m-1}} \times \mathbb F_2$. For $\nu \in \mathbb F_2,\lambda \in \mathbb F_{2^{m-1}}$,
the binary sequence $\{s_{\lambda, \nu}\}_{t=0}^{2(2^{m-1}-1)} $ of length $2(2^{m-1}-1)$ is defined by
\begin{align}\label{eq:binary-seq}
s_{\lambda,\nu}(t)=\begin{cases}
(-1)^{f(\beta^{t_0},0)+\tr^{m-1}_1(\lambda \beta^{t_0})},  &\text{if } t=2t_0, \\
(-1)^{f(\beta^{2^{m-2}}\beta^{t_0},1)+\tr^{m-1}_1(\beta^{2^{m-2}} \lambda \beta^{t_0})+\nu},  &\text{if } t=2t_0+1,
\end{cases}
\end{align}
where $0 \le t <2(2^{m-1}-1)$ and $\beta$ is a generator of $\mathbb F_{2^{m-1}}^\star$.

\begin{lemma}\label{lem:cor-o-e}
Let $f$ be a cyclic bent function on $\mathbb F_{2^{m-1}} \times \mathbb F_2$ with $f(0,0)=f(0,1)=0$, and,
$f_{1,b, \epsilon}(x_1, x_2)= f(x_1,x_2)+f(bx_1, x_2+\epsilon)$ for $b\in \mathbb F_{2^{m-1}}$ and $\epsilon \in \mathbb F_2$.
Let $\lambda, \lambda' \in \mathbb F_{2^{m-1}}$, $\nu, \nu' \in \mathbb F_2$ and $0\le \tau <2(2^{m-1}-1)$. Let $\{s_{\lambda,\nu}(t)\}$ and $ \{s_{\lambda',\nu'}(t)\}$
 be binary sequences defined in (\ref{eq:binary-seq}). Then

(1) for $\tau =2\tau_0+1$,
\begin{align*}
R_{s_{\lambda,\nu}, s_{\lambda',\nu'}}(\tau)=(-1)^{\nu}\mathcal W_{f_{1,\beta^{\tau_0+2^{m-2}},1}} (\lambda \beta^{\tau_0+2^{m-2}}+\lambda',\nu+\nu')-(-1)^{\nu}-(-1)^{\nu'};
\end{align*}

(2) for $\tau=2\tau_0$,
\begin{align*}
R_{s_{\lambda,\nu}, s_{\lambda',\nu'}}(\tau)= \mathcal W_{f_{1,\beta^{\tau_0},0}} (\lambda \beta^{\tau_0}+\lambda',\nu+\nu')-1-(-1)^{\nu+\nu'}.
\end{align*}
\end{lemma}

\begin{proof}
Denote $(-1)^{f(\beta^{t_0},0)+\tr^{m-1}_1(\lambda \beta^{t_0})}$
(resp., $(-1)^{f(\beta^{t_0},1)+\tr^{m-1}_1(\lambda \beta^{t_0})+\nu}$)
 by $a_{\lambda}(t_0)$ (resp., $b_{\lambda,\nu}(t_0)$).
By the definition of the correlation of sequences, one has
\begin{align*}
R_{s_{\lambda,\nu}, s_{\lambda',\nu'}}(2 \tau_0+1)=& \sum_{t=0}^{2(2^{m-1}-1)-1}  s_{\lambda,\nu}(t+2\tau_0+1) s_{\lambda',\nu'}(t)\\
=& \sum_{t=0}^{2^{m-1}-2}  s_{\lambda,\nu}(2t_0+2\tau_0+1) s_{\lambda',\nu'}(2t_0)\\
&+ \sum_{t=0}^{2^{m-1}-2}  s_{\lambda,\nu}(2t_0+2\tau_0+2) s_{\lambda',\nu'}(2t_0+1).
\end{align*}
By the definition of the sequence $\left \{s_{\lambda,\nu}(t) \right \}$,
\begin{align*}
R_{s_{\lambda,\nu}, s_{\lambda',\nu'}}(2 \tau_0+1)
=& \sum_{t=0}^{2^{m-1}-2}  b_{\lambda,\nu}(t_0+\tau_0+2^{m-2}) a_{\lambda',\nu'}(t_0)\\
&+ \sum_{t=0}^{2^{m-1}-2}  a_{\lambda,\nu}(t_0+\tau_0+1) b_{\lambda',\nu'}(t_0+2^{m-2})\\
=& \sum_{t=0}^{2^{m-1}-2}  b_{\lambda,\nu}(t_0+\tau_0+2^{m-2}) a_{\lambda',\nu'}(t_0)\\
&+ \sum_{t=0}^{2^{m-1}-2}  a_{\lambda,\nu}(t_0+\tau_0-2^{m-2}+1) b_{\lambda',\nu'}(t_0).
\end{align*}

Since $\left (t_0+\tau_0+2^{m-2}\right ) -\left ( t_0+\tau_0-2^{m-2}+1  \right )=2^{m-1}-1$, one gets
\begin{align*}
R_{s_{\lambda,\nu}, s_{\lambda',\nu'}}(2 \tau_0+1)=& \sum_{t=0}^{2^{m-1}-2}  b_{\lambda,\nu}(t_0+\tau_0+2^{m-2}) a_{\lambda',\nu'}(t_0)\\
&+ \sum_{t=0}^{2^{m-1}-2}  a_{\lambda,\nu}(t_0+\tau_0+2^{m-2}) b_{\lambda',\nu'}(t_0).
\end{align*}
Consequently,
\begin{align*}
R_{s_{\lambda,\nu}, s_{\lambda',\nu'}}(2 \tau_0+1)=& \sum_{t=0}^{2^{m-1}-2}     (-1)^{f(\beta^{t_0},0)+f(\beta^{\tau_0+2^{m-2}+t_0},1)  +\tr^{m-1}_1( (\lambda \beta^{\tau_0+2^{m-2}}+\lambda')\beta^{t_0})+\nu}\\
&+ \sum_{t=0}^{2^{m-1}-2}     (-1)^{f(\beta^{t_0},1)+f(\beta^{\tau_0+2^{m-2}+t_0},0)  +\tr^{m-1}_1( (\lambda \beta^{\tau_0+2^{m-2}}+\lambda')\beta^{t_0})+\nu'}\\
=& \sum_{x_1 \in \mathbb F_{2^{m-1}}}    (-1)^{f_{1,\beta^{\tau_0+2^{m-2}},1}(x_1,0)  +\tr^{m-1}_1( (\lambda \beta^{\tau_0+2^{m-2}}+\lambda')x_1)+\nu}\\
&+ \sum_{x_1 \in \mathbb F_{2^{m-1}}}    (-1)^{f_{1,\beta^{\tau_0+2^{m-2}},1}(x_1,1) +\tr^{m-1}_1( (\lambda \beta^{\tau_0+2^{m-2}}+\lambda')x_1)+\nu'}\\
&-(-1)^{\nu}-(-1)^{\nu'}.
\end{align*}
Then,
\begin{align*}
R_{s_{\lambda,\nu}, s_{\lambda',\nu'}}(2 \tau_0+1)=& (-1)^{\nu}\sum_{\substack{x_1 \in \mathbb F_{2^{m-1}}\\ x_2 \in \mathbb F_2}}    (-1)^{f_{1,\beta^{\tau_0+2^{m-2}},1}(x_1,x_2) +\tr^{m-1}_1( (\lambda \beta^{\tau_0+2^{m-2}}+\lambda')x_1) +(\nu+\nu')x_2}\\
&-(-1)^{\nu}-(-1)^{\nu'}\\
=& (-1)^{\nu} \mathcal W_{f_{1,\beta^{\tau_0+2^{m-2}},1}} (\lambda \beta^{\tau_0+2^{m-2}}+\lambda',\nu+\nu')-(-1)^{\nu}-(-1)^{\nu'}.
\end{align*}
One has also
\begin{align*}
R_{s_{\lambda,\nu}, s_{\lambda',\nu'}}(2 \tau_0)=& \sum_{t=0}^{2(2^{m-1}-1)-1}  s_{\lambda,\nu}(t+2\tau_0) s_{\lambda',\nu'}(t)\\
=& \sum_{t=0}^{2^{m-1}-2}  s_{\lambda,\nu}(2t_0+2\tau_0) s_{\lambda',\nu'}(2t_0)\\
&+ \sum_{t=0}^{2^{m-1}-2}  s_{\lambda,\nu}(2t_0+2\tau_0+1) s_{\lambda',\nu'}(2t_0+1).
\end{align*}
By the definition $\left \{ s_{\lambda,\nu}(t) \right \}$,
\begin{align*}
R_{s_{\lambda,\nu}, s_{\lambda',\nu'}}(2 \tau_0)
=& \sum_{t=0}^{2^{m-1}-2}  a_{\lambda,\nu}(t_0+\tau_0) a_{\lambda',\nu'}(t_0)\\
&+ \sum_{t=0}^{2^{m-1}-2}  b_{\lambda,\nu}(t_0+\tau_0+2^{m-2}) b_{\lambda',\nu'}(t_0+2^{m-2})\\
=& \sum_{t=0}^{2^{m-1}-2}  a_{\lambda,\nu}(t_0+\tau_0) a_{\lambda',\nu'}(t_0)+ \sum_{t=0}^{2^{m-1}-2}  b_{\lambda,\nu}(t_0+\tau_0) b_{\lambda',\nu'}(t_0).
\end{align*}
Thus,
\begin{align*}
R_{s_{\lambda,\nu}, s_{\lambda',\nu'}}(2 \tau_0)=& \sum_{t=0}^{2^{m-1}-2}
(-1)^{f(\beta^{t_0},0)+f(\beta^{\tau_0+t_0},0)  +\tr^{m-1}_1( (\lambda \beta^{\tau_0}+\lambda')\beta^{t_0})}\\
&+ \sum_{t=0}^{2^{m-1}-2}     (-1)^{f(\beta^{t_0},1)+f(\beta^{\tau_0+t_0},1)  +\tr^{m-1}_1( (\lambda \beta^{\tau_0}+\lambda')\beta^{t_0})+\nu+\nu'}\\
=& \sum_{x_1 \in \mathbb F_{2^{m-1}}}    (-1)^{f_{1,\beta^{\tau_0},0}(x_1,0)  +\tr^{m-1}_1( (\lambda \beta^{\tau_0}+\lambda')x_1)}\\
&+ \sum_{x_1 \in \mathbb F_{2^{m-1}}}    (-1)^{f_{1,\beta^{\tau_0},0}(x_1,1) +\tr^{m-1}_1( (\lambda \beta^{\tau_0}+\lambda')x_1)+\nu+\nu'}-1-(-1)^{\nu+\nu'}.
\end{align*}
Finally,
\begin{align*}
R_{s_{\lambda,\nu}, s_{\lambda',\nu'}}(2 \tau_0)=&  \sum_{ \substack{x_1 \in \mathbb F_{2^{m-1}}\\ x_2\in \mathbb F_2}}    (-1)^{f_{1,\beta^{\tau_0},0}(x_1,x_2) +\tr^{m-1}_1( (\lambda \beta^{\tau_0}+\lambda')x_1)+(\nu+\nu')x_2}-1-(-1)^{\nu+\nu'}\\
=& \mathcal W_{f_{1,\beta^{\tau_0},0}} (\lambda \beta^{\tau_0}+\lambda',\nu+\nu')-1-(-1)^{\nu+\nu'}.
\end{align*}
This completes the proof.
\end{proof}

\begin{lemma}\label{lem:cor-10}
Let $f$ be a cyclic bent function on $\mathbb F_{2^{m-1}} \times \mathbb F_2$ with $f(0,0)=f(0,1)=0$, and
$f_{1,b, \epsilon}(x_1, x_2)= f(x_1,x_2)+f(bx_1, x_2+\epsilon)$ for $b\in \mathbb F_{2^{m-1}}$ and $\epsilon \in \mathbb F_2$.
Let $\lambda, \lambda' \in \mathbb F_{2^{m-1}}$, $\nu, \nu' \in \mathbb F_2$ and $0\le \tau_0 <2^{m-1}-1$. Let $\{s_{\lambda,\nu}(t)\}$ and $ \{s_{\lambda',\nu'}(t)\}$
 be binary sequences defined in (\ref{eq:binary-seq}).

(1) If $\tr^{m-1}_1(\lambda)=0, \tr^{m-1}_1(\lambda')=0$ and $f(x_1,0)+f(x_1,1)=\tr^{m-1}_1(x_1)$, then
\begin{align*}
R_{s_{\lambda,\nu}, s_{\lambda',\nu'}}(2\tau_0+1)=
\begin{cases}
-(-1)^{\nu}-(-1)^{\nu'}, &\text{ when } \tau_0=2^{m-2}-1, \\
\pm 2^{\frac{m}{2}} \pm 1 \pm 1, &\text{ when } \tau_0 \neq 2^{m-2}-1.
\end{cases}
\end{align*}

(2) For $\tau=2\tau_0$,
\begin{align*}
R_{s_{\lambda,\nu}, s_{\lambda',\nu'}}(\tau)=
\begin{cases}
2(2^{m-1}-1),  &\text{ if } \tau_0=0, (\lambda,\nu)=  (\lambda',\nu'),     \\
-1 - (-1)^{\nu+\nu'},  &\text{ if } \tau_0=0, (\lambda,\nu)\neq  (\lambda',\nu'),     \\
\pm 2^{\frac{m}{2}}-1\pm 1, &\text{ if } \tau_0 \neq 0.
\end{cases}
\end{align*}

\end{lemma}
\begin{proof}
(1) If $\tau_0 \neq 2^{m-2}-1$, then $\beta^{\tau_0+2^{m-2}}\neq 1$ and
$f_{1,\beta^{\tau_0+2^{m-2}},1}$ is a bent function. By Lemma \ref{lem:cor-o-e},
$R_{s_{\lambda,\nu}, s_{\lambda',\nu'}}(2\tau_0+1)= \pm 2^{\frac{m}{2}} \pm 1 \pm 1$.

If $\tau_0 = 2^{m-2}-1$, then  $\beta^{\tau_0+2^{m-2}}= 1$ and
$f_{1,\beta^{\tau_0+2^{m-2}},1}
(x_1,x_2)=f(x_1,x_2)+f(x_1,x_2+1)=\tr^{m-1}_1(x_1)$.
Since $\tr^{m-1}_1(\lambda+\lambda'+1)=\tr^{m-1}_1(\lambda)+ \tr^{m-1}_1(\lambda')+\tr^{m-1}_1(1)=1$,  we have $\lambda+\lambda'+1\neq 0$.
It follows from Lemma \ref{lem:cor-o-e} that
\begin{align*}
R_{s_{\lambda,\nu}, s_{\lambda',\nu'}}(2\tau_0+1)= &
(-1)^{\nu}\sum_{\substack{x_1 \in \mathbb F_{2^{m-1}}\\ x_2 \in \mathbb F_2}}  (-1)^{\tr^{m-1}_1 ((\lambda+\lambda' +1)x_1) +(\nu+\nu')x_2}-(-1)^{\nu}-(-1)^{\nu'}\\
=& -(-1)^{\nu}-(-1)^{\nu'},
\end{align*}
which is due to the fact $\lambda+\lambda'+1\neq 0$.

 (2) If $\tau_0\neq 0$, then $f_{1,\beta^{\tau_0},0}$ is a bent function. By Lemma \ref{lem:cor-o-e},
 $R_{s_{\lambda,\nu}, s_{\lambda',\nu'}}(2\tau_0)= \pm 2^{\frac{m}{2}}-1\pm 1$.

 If $\tau_0= 0$, then  $f_{1,\beta^{\tau_0},0}=0$.
 By Lemma \ref{lem:cor-o-e},
 \begin{align*}
 R_{s_{\lambda,\nu}, s_{\lambda',\nu'}}(2\cdot 0)=&\sum_{ \substack{x_1 \in \mathbb F_{2^{m-1}}\\ x_2\in \mathbb F_2}}
 (-1)^{\tr^{m-1}_1( (\lambda +\lambda')x_1)+(\nu+\nu')x_2}-1-(-1)^{\nu+\nu'}\\
 =& \begin{cases}
2(2^{m-1}-1),  &\text{ if }  (\lambda,\nu)=  (\lambda',\nu'),     \\
-1 - (-1)^{\nu+\nu'},  &\text{ if } (\lambda,\nu)\neq  (\lambda',\nu').
\end{cases}
 \end{align*}
This completes the proof.
\end{proof}

\begin{lemma}
Let $g$ be a bent function on $\mathbb F_{2^{m-1}}\times \mathbb F_2$ with $g(0,0)=g(0,1)=0$.
Let $N_{g,\epsilon}= \# \left \{ \mu \in \mathbb F_{2^{m-1}} |\mathcal W_g(\mu,0) = (-1)^{\epsilon} 2^{\frac{m}{2}}  \right \}$
and $N_{g,\epsilon}'= \# \left \{ \mu \in \mathbb F_{2^{m-1}} |\mathcal W_g(\mu,1) = (-1)^{\epsilon} 2^{\frac{m}{2}}  \right \}$. Then,
\begin{align*}
N_{g,\epsilon}'=N_{g,\epsilon}= 2^{m-2}+(-1)^{\epsilon} 2^{\frac{m-2}{2}}.
\end{align*}

\end{lemma}
\begin{proof}
Since $g$ is bent, one has
\begin{align*}
N_{g,\epsilon}=& \sum_{\mu \in \mathbb F_{2^{m-1}}}   \frac{(-1)^{\epsilon}}{2\cdot 2^{\frac{m}{2}}}  \left  (\mathcal W_g(\mu,0) + (-1)^{\epsilon} 2^{\frac{m}{2}} \right )\\
=& \frac{1}{2}\sum_{\mu \in \mathbb F_{2^{m-1}}}  1   +   \frac{(-1)^{\epsilon}}{2\cdot 2^{\frac{m}{2}}}  \sum_{\mu \in \mathbb F_{2^{m-1}}}   \mathcal W_g(\mu,0) \\
=& 2^{m-2} + \frac{(-1)^{\epsilon}}{2\cdot 2^{\frac{m}{2}}}   \sum_{\mu \in \mathbb F_{2^{m-1}}}   \mathcal W_g(\mu,0) .
\end{align*}
By the definition of $\mathcal W_g(\mu,0)$,
\begin{align*}
\sum_{\mu \in \mathbb F_{2^{m-1}}}   \mathcal W_g(\mu,0)= & \sum_{\mu \in \mathbb F_{2^{m-1}}}  \sum_{(x_1,x_2) \in \mathbb F_{2^{m-1}}\times \mathbb F_2}
 (-1)^{g(x_1,x_2)+\tr^{m-1}_1(\mu x_1)}\\
 =&\sum_{(x_1,x_2) \in \mathbb F_{2^{m-1}}\times \mathbb F_2}   (-1)^{g(x_1,x_2)}  \sum_{\mu \in \mathbb F_{2^{m-1}}}   (-1)^{\tr^{m-1}_1(\mu x_1)}\\
 =& 2^{m-1} \sum_{x_2 \in \mathbb F_2} (-1)^{g(0,x_2)}\\
 =& 2^{m-1} \left ((-1)^{g(0,0)}+(-1)^{g(0,1)} \right )\\
 =& 2^{m}.
\end{align*}
Thus,
\begin{align*}
N_{g,\epsilon}=& 2^{m-2} + \frac{(-1)^{\epsilon}}{2\cdot 2^{\frac{m}{2}}}   \sum_{\mu \in \mathbb F_{2^{m-1}}}   \mathcal W_g(\mu,0)\\
=& 2^{m-2} + \frac{(-1)^{\epsilon}}{2\cdot 2^{\frac{m}{2}}} 2^{m}\\
=& 2^{m-2} + (-1)^{\epsilon} 2^{\frac{m-2}{2}}.
\end{align*}
Thus, the desired conclusion is true for $N_{g,\epsilon}$.

Similarly,
\begin{align*}
\mathcal W_g(\mu,1)=& \sum_{(x_1,x_2) \in \mathbb F_{2^{m-1}}\times \mathbb F_2}  (-1)^{g(x_1,x_2)+\tr^{m-1}_1(\mu x_1)+x_2}\\
=& \sum_{(x_1,x_2) \in \mathbb F_{2^{m-1}}\times \mathbb F_2} (-1)^{g'(x_1,x_2)+\tr^{m-1}_1(\mu x_1)}\\
=& \mathcal W_{g'}(\mu,0),
\end{align*}
where $g'(x_1,x_2)=g(x_1,x_2)+x_2$.
Thus, $N_{g,\epsilon}'=N_{g',\epsilon}=2^{m-2} + (-1)^{\epsilon} 2^{\frac{m-2}{2}}$.
This completes the proof.
\end{proof}

\begin{lemma}\label{eq:sum-w-TR-0-0-0}
Let $g$ be a bent function on $\mathbb F_{2^{m-1}} \times \mathbb F_2$ with $g(0,0)=g(0,1)=0$.
Let $u, \mu_1, \mu_2  \in \mathbb F_2$ and $b\in \mathbb F_{2^{m-1}}\setminus \mathbb F_2$.
Let $W= \sum_{(\lambda_1, \lambda_2) \in \mathbb F_{2^{m-1}}^2}  \mathcal W_{g}(b \lambda_1+ \lambda_2,u) (-1)^{\tr^{m-1}_1(\mu_1 \lambda_1+\mu_2 \lambda_2 )}$.
Then,
\begin{align*}
W=\begin{cases}
2^{2m-1}, & \mbox{ if } \mu_1=\mu_2=u=0,\\
0, & \mbox{ otherwise}.
\end{cases}
\end{align*}
\end{lemma}

\begin{proof}
By definition, one has
\begin{align*}
W=& \sum_{(\lambda_1, \lambda_2) \in \mathbb F_{2^{m-1}}^2}  (-1)^{\tr^{m-1}_1(\mu_1 \lambda_1+\mu_2 \lambda_2 )}\nonumber \\
& \cdot \sum_{(x_1,x_2) \in \mathbb F_{2^{m-1}}\times \mathbb F_2}  (-1)^{g(x_1,x_2)+ \tr^{m-1}_1((b \lambda_1+ \lambda_2)x_1) + u x_2} \nonumber  \\
=&  \sum_{(x_1,x_2) \in \mathbb F_{2^{m-1}}\times \mathbb F_2} (-1)^{g(x_1,x_2)+u x_2} \nonumber \\
& \cdot \sum_{(\lambda_1, \lambda_2) \in \mathbb F_{2^{m-1}}^2}  (-1)^{\tr^{m-1}_1\left ( (b x_1+\mu_1) \lambda_1+ (x_1+\mu_2)\lambda_2 \right )}.
\end{align*}

If $(\mu_1,\mu_2)=(0,0)$, then
\begin{align*}
W= & 2^{2(m-1)} \sum_{x_2 \in \mathbb F_2}  (-1)^{g(0,x_2)+u x_2} \\
=& 2^{2(m-1)} \sum_{x_2 \in \mathbb F_2}  (-1)^{u x_2}\\
=&\begin{cases}
2^{2m-1}, & \text{ if } u=0, \\
0, & \text{ if } u=1.
\end{cases}
\end{align*}

If $(\mu_1,\mu_2)\neq (0,0)$, then  $(b x_1+\mu_1, x_1+\mu_2)\neq (0,0)$  and
\begin{align*}
\sum_{(\lambda_1, \lambda_2) \in \mathbb F_{2^{m-1}}^2}  (-1)^{\tr^{m-1}_1\left ( (b x_1+\mu_1) \lambda_1+ (x_1+\mu_2)\lambda_2 \right )}=0,
\end{align*}
 for any
$x_1\in \mathbb F_{2^{m-1}}$. Thus, $W=0$.
This completes the proof.
\end{proof}

\begin{lemma}\label{lem:W-lam-lam}
Let $g$ be a bent function on $\mathbb F_{2^{m-1}} \times \mathbb F_2$ with $g(0,0)=g(0,1)=0$.
For $u, \epsilon \in \mathbb F_2$ and $b\in \mathbb F_{2^{m-1}}\setminus \mathbb F_2$, define the following set
 \begin{align*}
 T_{u,\epsilon}=\left \{ (\lambda_1, \lambda_2) \in H \times H | \mathcal W_{g}(b \lambda_1+ \lambda_2,u)=(-1)^{\epsilon} 2^{\frac{m}{2}}\right \},
 \end{align*}
 where $H:=\left \{\lambda \in \mathbb F_{2^{m-1}}|  \tr^{m-1}_1(\lambda)=0 \right \}$. Then,
 \begin{align*}
\# T_{u,\epsilon}=\begin{cases}
 2^{m-2}\left (2^{m-3}+(-1)^{\epsilon} 2^{\frac{m-4}{2}} \right ), & \mbox{if }  u=0, \\
 2^{m-2} \cdot  2^{m-3} , & \mbox{if } u=1.
\end{cases}
\end{align*}
\end{lemma}

\begin{proof}
One has
\begin{align}\label{eq:H-H}
\# T_{u,\epsilon}= & \frac{(-1)^{\epsilon}}{2^3\cdot 2^{\frac{m}{2}}} \sum_{(\lambda_1,\lambda_2)\in \mathbb F_{2^{m-1}}^2}
\left  (\mathcal W_{g}(b \lambda_1+ \lambda_2,u)+(-1)^{\epsilon} 2^{\frac{m}{2}}\right )\nonumber \\
&~~~~~~~~~~~~~~~~~~~~~~~\cdot \left (1+(-1)^{\tr^{m-1}_1(\lambda_1)} \right ) \left (1+(-1)^{\tr^{m-1}_1(\lambda_2)} \right )\nonumber \\
=& \frac{(-1)^{\epsilon}}{2^3\cdot 2^{\frac{m}{2}}} A+ \frac{1}{2^3}B,
\end{align}
where $A=\sum_{(\lambda_1,\lambda_2)\in \mathbb F_{2^{m-1}}^2} \mathcal W_{g}\left (b \lambda_1+ \lambda_2,u \right )
\left (1+(-1)^{\tr^{m-1}_1(\lambda_1)} \right ) \left (1+(-1)^{\tr^{m-1}_1(\lambda_2)} \right )$   and\\
$B=\sum_{(\lambda_1,\lambda_2)\in \mathbb F_{2^{m-1}}^2} \left (1+(-1)^{\tr^{m-1}_1(\lambda_1)} \right ) \left (1+(-1)^{\tr^{m-1}_1(\lambda_2)} \right )$.

For $B$, one deduces that
\begin{align*}
B=& \sum_{(\lambda_1,\lambda_2)\in \mathbb F_{2^{m-1}}^2} 1+ \sum_{(\lambda_1,\lambda_2)\in \mathbb F_{2^{m-1}}^2} (-1)^{\tr^{m-1}_1(\lambda_1+\lambda_2)}\\
& \sum_{(\lambda_1,\lambda_2)\in \mathbb F_{2^{m-1}}^2}  (-1)^{\tr^{m-1}_1(\lambda_1)}    +\sum_{(\lambda_1,\lambda_2)\in \mathbb F_{2^{m-1}}^2} (-1)^{\tr^{m-1}_1(\lambda_2)} \\
=& 2^{2(m-1)}.
\end{align*}

For $A$, using Lemma \ref{eq:sum-w-TR-0-0-0}, one obtains
\begin{align*}
A=&\sum_{(\lambda_1,\lambda_2)\in \mathbb F_{2^{m-1}}^2}  \mathcal W_{g}(b \lambda_1+ \lambda_2,u)\\
& +
\sum_{(\lambda_1,\lambda_2)\in \mathbb F_{2^{m-1}}^2}  \mathcal W_{g}(b \lambda_1+ \lambda_2,u)  (-1)^{\tr^{m-1}_1(\lambda_1+\lambda_2)}\\
& + \sum_{(\lambda_1,\lambda_2)\in \mathbb F_{2^{m-1}}^2}  \mathcal W_{g}(b \lambda_1+ \lambda_2,u) (-1)^{\tr^{m-1}_1(\lambda_1)} \\
& +
\sum_{(\lambda_1,\lambda_2)\in \mathbb F_{2^{m-1}}^2}  \mathcal W_{g}(b \lambda_1+ \lambda_2,u) (-1)^{\tr^{m-1}_1(\lambda_2)}\\
=&\begin{cases}
2^{2m-1}, & \mbox{if }  u=0, \\
0, & \mbox{if }  u=1.
\end{cases}
\end{align*}

By  (\ref{eq:H-H}),
\begin{align*}
\# T_{u,\epsilon}=& \frac{(-1)^{\epsilon}}{2^3\cdot 2^{\frac{m}{2}}} A+ \frac{1}{2^3}B\\
=& \frac{(-1)^{\epsilon}}{2^3\cdot 2^{\frac{m}{2}}} A+ \frac{1}{2^3} 2^{2(m-1)}\\
=&\begin{cases}
\frac{(-1)^{\epsilon}}{2^3\cdot 2^{\frac{m}{2}}}\cdot 2^{2m-1}+ \frac{1}{2^3} 2^{2(m-1)}, & \mbox{if }  u=0, \\
\frac{1}{2^3} 2^{2(m-1)}, & \mbox{if }  u=1.
\end{cases}
\end{align*}
Consequently,
\begin{align*}
\# T_{u,\epsilon}=\begin{cases}
 2^{m-2}\left (2^{m-3}+(-1)^{\epsilon} 2^{\frac{m-4}{2}} \right ), & \mbox{if }  u=0, \\
 2^{m-2} \cdot  2^{m-3}, & \mbox{if } u=1.
\end{cases}
\end{align*}
This completes the proof.
\end{proof}

We now have the following main result on the correlation value distribution of the binary  sequence family from cyclic bent functions.

\begin{theorem}\label{thm:u-f-b}
Let $f$ be a cyclic bent function on $\mathbb F_{2^{m-1}}\times \mathbb F_2$ such that $f(0,0)=f(0,1)=0$ and $f(x_1,0)+f(x_1,1)=\tr^{m-1}_1(x_1)$.
Let the family $\mathcal U_{f}^{b}$ of
binary sequences  of length $2(2^{m-1}-1)$ be defined by
\begin{align*}
\mathcal U_{f}^{b}=\left \{ \{s_{\lambda, \nu}\}_{t=0}^{2(2^{m-1}-1)} |
\nu \in \mathbb F_2,\lambda \in \mathbb F_{2^{m-1}} \text{ with } \tr^{m-1}_1(\lambda)=0 \right \},
\end{align*}
where $\{s_{\lambda, \nu}\}_{t=0}^{2(2^{m-1}-1)}$ is the sequence defined in
(\ref{eq:binary-seq}). Then, the correlation value distribution of the family
$\mathcal U_{f}^{b}$ is given in Table \ref{table-Uf-b}.

\begin{table}[htbp]
\centering
\caption{The correlation distribution of the family $\mathcal U_f^b$}
\label{table-Uf-b}
\begin{tabular}{|c|c|}
  \hline
  Value & Frequency \\
  \hline
  $2(2^{m-1}-1)$ & $2^{m-1}$\\
\hline
$-2$ &  $2^{m-1}(3\cdot 2^{m-3}-1)$\\
\hline
$0$& $2^{2m-2}$\\
\hline
$2$& $2^{2m-4}$\\
\hline
$2^{\frac{m}{2}}-2$ &
 $3 \cdot  2^{m-2} \left (2^{m-1}-2 \right ) \left (2^{m-3}+ 2^{\frac{m-4}{2}} \right )$\\
\hline
$2^{\frac{m}{2}}$&
 $2^{2m-3}\left (2^{m-1}-2 \right )$\\
\hline
$2^{\frac{m}{2}}+2$&
 $ 2^{m-2} \left (2^{m-1}-2 \right ) \left (2^{m-3}- 2^{\frac{m-4}{2}} \right )$\\
\hline
$-2^{\frac{m}{2}}-2$ &
 $3 \cdot  2^{m-2} \left (2^{m-1}-2 \right ) \left (2^{m-3}- 2^{\frac{m-4}{2}} \right )$\\
\hline
$-2^{\frac{m}{2}}$&
 $2^{2m-3}\left (2^{m-1}-2 \right )$\\
\hline
$-2^{\frac{m}{2}}+2$&
 $  2^{m-2} \left (2^{m-1}-2 \right ) \left (2^{m-3}+ 2^{\frac{m-4}{2}} \right )$\\
\hline
\end{tabular}
\end{table}
\end{theorem}

\begin{proof}
Let $H:=\left \{\lambda \in \mathbb F_{2^{m-1}}|  \tr^{m-1}_1(\lambda)=0 \right \}$. We investigate the following 10 cases for  the
correlation:

Case 0: Let $C_0$ be the set of all $ (\tau,\lambda,\nu,\lambda',\nu')$  with $\lambda, \lambda' \in H$, $\nu,\nu'\in \mathbb F_2$, $0\le \tau <2(2^{m-1}-1)$ and
 $R_{s_{\lambda,\nu}, s_{\lambda',\nu'}}(\tau)=2\left (2^{m-1}-1 \right )$. By Lemma \ref{lem:cor-10}, $(\tau,\lambda,\nu,\lambda',\nu')\in C_0$ if and only if
 $\tau_0=0, (\lambda,\nu)=  (\lambda',\nu') $. Thus, $\# C_0= 2^{m-1}$.

 Case 1: Let $C_1$ be the set of all $ (\tau,\lambda,\nu,\lambda',\nu')$  with $\lambda, \lambda' \in H$, $\nu,\nu'\in \mathbb F_2$, $0\le \tau <2(2^{m-1}-1)$ and
 $R_{s_{\lambda,\nu}, s_{\lambda',\nu'}}(\tau)=-2$. By Lemma \ref{lem:cor-10}, $ (\tau,\lambda,\nu,\lambda',\nu')\in C_1$  if and only if
 $\tau=2\left (2^{m-2}-1 \right )+1$, $\nu=\nu'=0$, or, $\tau=0$, $ \nu=\nu', \lambda \neq  \lambda'$.
 Thus, $\# C_1= 2^{2(m-2)}+2\cdot 2^{m-2} \cdot \left (2^{m-2}-1 \right )=2^{m-1}(3\cdot 2^{m-3}-1)$.

Case 2: Let $C_2$ be the set of all $ (\tau,\lambda,\nu,\lambda',\nu')$  with $\lambda, \lambda' \in H$, $\nu,\nu'\in \mathbb F_2$, $0\le \tau <2(2^{m-1}-1)$ and
 $R_{s_{\lambda,\nu}, s_{\lambda',\nu'}}(\tau)=0$. By Lemma \ref{lem:cor-10}, $ (\tau,\lambda,\nu,\lambda',\nu')\in C_2$  if and only if
 $\tau=2\left (2^{m-2}-1 \right )+1$, $\nu + \nu'=1$, or, $\tau=0$, $ \nu +\nu'=1$.
 Thus, $\# C_2= 2\cdot 2^{2(m-2)}+2\cdot 2^{2(m-2)}=2^{2(m-1)}$.

Case 3: Let $C_3$ be the set of all $ (\tau,\lambda,\nu,\lambda',\nu')$  with $\lambda, \lambda' \in H$, $\nu,\nu'\in \mathbb F_2$, $0\le \tau <2(2^{m-1}-1)$ and
 $R_{s_{\lambda,\nu}, s_{\lambda',\nu'}}(\tau)=2$. By Lemma \ref{lem:cor-10}, $ (\tau,\lambda,\nu,\lambda',\nu')\in C_3$  if and only if
 $\tau=2\left (2^{m-2}-1 \right )+1$, $\nu=\nu'=1$.
 Thus, $\# C_3= 2^{2(m-2)}$.

 Case 4: Let $C_4$ be the set of all $ (\tau,\lambda,\nu,\lambda',\nu')$  with $\lambda, \lambda' \in H$, $\nu,\nu'\in \mathbb F_2$, $0\le \tau <2(2^{m-1}-1)$ and
 $R_{s_{\lambda,\nu}, s_{\lambda',\nu'}}(\tau)=-2+2^{\frac{m}{2}}$. By Lemmas \ref{lem:cor-o-e} and \ref{lem:cor-10}, $ (\tau,\lambda,\nu,\lambda',\nu')\in C_4$  if and only if
 $\tau= 2\tau_0+1\neq 2\left (2^{m-2}-1 \right )+1$, $\nu=\nu'=0$,
 $\mathcal W_{f_{1,\beta^{\tau_0+2^{m-2}},1}} (\lambda \beta^{\tau_0+2^{m-2}}+\lambda',\nu+\nu')=(-1)^{\nu} 2^{\frac{m}{2}}$, or,
  $\tau=2\tau_0 \neq 0$, $ \nu=\nu'$, $\mathcal W_{f_{1,\beta^{\tau_0},0}} (\lambda \beta^{\tau_0}+\lambda',\nu+\nu')=2^{\frac{m}{2}}$.
By Lemma \ref{lem:W-lam-lam},
 \begin{align*}
\# C_4=& \left (2^{m-1}-2 \right )\cdot  2^{m-2}\left (2^{m-3}+ 2^{\frac{m-4}{2}} \right )+
\left (2^{m-1}-2 \right )\cdot 2 \cdot   2^{m-2}\left (2^{m-3}+ 2^{\frac{m-4}{2}} \right )  \\
=& 3 \cdot  2^{m-2} \left (2^{m-1}-2 \right ) \left (2^{m-3}+ 2^{\frac{m-4}{2}} \right ).
\end{align*}

Case 5: Let $C_5$ be the set of all $ (\tau,\lambda,\nu,\lambda',\nu')$  with $\lambda, \lambda' \in H$, $\nu,\nu'\in \mathbb F_2$, $0\le \tau <2(2^{m-1}-1)$ and
 $R_{s_{\lambda,\nu}, s_{\lambda',\nu'}}(\tau)=2^{\frac{m}{2}}$. By Lemmas \ref{lem:cor-o-e} and \ref{lem:cor-10}, $ (\tau,\lambda,\nu,\lambda',\nu')\in C_5$  if and only if
 $\tau= 2\tau_0+1\neq 2\left (2^{m-2}-1 \right )+1$, $\nu+ \nu'=1$,
 $\mathcal W_{f_{1,\beta^{\tau_0+2^{m-2}},1}} (\lambda \beta^{\tau_0+2^{m-2}}+\lambda',\nu+\nu')=(-1)^{\nu} 2^{\frac{m}{2}}$, or,
  $\tau=2\tau_0 \neq 0$, $ \nu+\nu'=1$, $\mathcal W_{f_{1,\beta^{\tau_0},0}} (\lambda \beta^{\tau_0}+\lambda',\nu+\nu')=2^{\frac{m}{2}}$.
By Lemma \ref{lem:W-lam-lam},
 \begin{align*}
\# C_5=&2\cdot  \left (2^{m-1}-2 \right )\cdot  2 \cdot 2^{2m-5}\\
=& 2^{2m-3}\left (2^{m-1}-2 \right ).
\end{align*}

 Case 6: Let $C_6$ be the set of all $ (\tau,\lambda,\nu,\lambda',\nu')$  with $\lambda, \lambda' \in H$, $\nu,\nu'\in \mathbb F_2$, $0\le \tau <2(2^{m-1}-1)$ and
 $R_{s_{\lambda,\nu}, s_{\lambda',\nu'}}(\tau)=2+2^{\frac{m}{2}}$. By Lemmas \ref{lem:cor-o-e} and \ref{lem:cor-10}, $ (\tau,\lambda,\nu,\lambda',\nu')\in C_6$ if and only if
 $\tau= 2\tau_0+1\neq 2\left (2^{m-2}-1 \right )+1$, $\nu=\nu'=1$,
 $\mathcal W_{f_{1,\beta^{\tau_0+2^{m-2}},1}} (\lambda \beta^{\tau_0+2^{m-2}}+\lambda',\nu+\nu')=(-1)^{\nu} 2^{\frac{m}{2}}$.
By Lemma \ref{lem:W-lam-lam},
 \begin{align*}
\# C_6= \left (2^{m-1}-2 \right )\cdot  2^{m-2}\left (2^{m-3}- 2^{\frac{m-4}{2}} \right ).
\end{align*}

Case 7: Let $C_7$ be the set of all $ (\tau,\lambda,\nu,\lambda',\nu')$  with $\lambda, \lambda' \in H$, $\nu,\nu'\in \mathbb F_2$, $0\le \tau <2(2^{m-1}-1)$ and
 $R_{s_{\lambda,\nu}, s_{\lambda',\nu'}}(\tau)=-2-2^{\frac{m}{2}}$. By Lemmas \ref{lem:cor-o-e} and \ref{lem:cor-10}, $ (\tau,\lambda,\nu,\lambda',\nu')\in C_7$  if and only if
 $\tau= 2\tau_0+1\neq 2\left (2^{m-2}-1 \right )+1$, $\nu=\nu'=0$,
 $\mathcal W_{f_{1,\beta^{\tau_0+2^{m-2}},1}} (\lambda \beta^{\tau_0+2^{m-2}}+\lambda',\nu+\nu')=(-1)^{\nu+1} 2^{\frac{m}{2}}$, or,
  $\tau=2\tau_0 \neq 0$, $ \nu=\nu'$, $\mathcal W_{f_{1,\beta^{\tau_0},0}} (\lambda \beta^{\tau_0}+\lambda',\nu+\nu')=-2^{\frac{m}{2}}$.
By Lemma \ref{lem:W-lam-lam},
 \begin{align*}
\# C_7=& \left (2^{m-1}-2 \right )\cdot  2^{m-2}\left (2^{m-3}- 2^{\frac{m-4}{2}} \right )+
\left (2^{m-1}-2 \right )\cdot 2 \cdot   2^{m-2}\left (2^{m-3}- 2^{\frac{m-4}{2}} \right )  \\
=& 3 \cdot  2^{m-2} \left (2^{m-1}-2 \right ) \left (2^{m-3}- 2^{\frac{m-4}{2}} \right ).
\end{align*}

Case 8: Let $C_8$ be the set of all $ (\tau,\lambda,\nu,\lambda',\nu')$  with $\lambda, \lambda' \in H$, $\nu,\nu'\in \mathbb F_2$, $0\le \tau <2(2^{m-1}-1)$ and
 $R_{s_{\lambda,\nu}, s_{\lambda',\nu'}}(\tau)=-2^{\frac{m}{2}}$. By Lemmas \ref{lem:cor-o-e} and \ref{lem:cor-10}, $ (\tau,\lambda,\nu,\lambda',\nu')\in C_8$  if and only if
 $\tau= 2\tau_0+1\neq 2\left (2^{m-2}-1 \right )+1$, $\nu+ \nu'=1$,
 $\mathcal W_{f_{1,\beta^{\tau_0+2^{m-2}},1}} (\lambda \beta^{\tau_0+2^{m-2}}+\lambda',\nu+\nu')=(-1)^{\nu+1} 2^{\frac{m}{2}}$, or,
  $\tau=2\tau_0 \neq 0$, $ \nu+\nu'=1$, $\mathcal W_{f_{1,\beta^{\tau_0},0}} (\lambda \beta^{\tau_0}+\lambda',\nu+\nu')=-2^{\frac{m}{2}}$.
By Lemma \ref{lem:W-lam-lam},
 \begin{align*}
\# C_8=&2\cdot  \left (2^{m-1}-2 \right )\cdot  2 \cdot 2^{2m-5}\\
=& 2^{2m-3}\left (2^{m-1}-2 \right ).
\end{align*}

 Case 9: Let $C_9$ be the set of all $ (\tau,\lambda,\nu,\lambda',\nu')$  with $\lambda, \lambda' \in H$, $\nu,\nu'\in \mathbb F_2$, $0\le \tau <2(2^{m-1}-1)$ and
 $R_{s_{\lambda,\nu}, s_{\lambda',\nu'}}(\tau)=2-2^{\frac{m}{2}}$. By Lemmas \ref{lem:cor-o-e} and \ref{lem:cor-10}, $ (\tau,\lambda,\nu,\lambda',\nu')\in C_9$  if and only if
 $\tau= 2\tau_0+1\neq 2\left (2^{m-2}-1 \right )+1$, $\nu=\nu'=1$,
 $\mathcal W_{f_{1,\beta^{\tau_0+2^{m-2}},1}} (\lambda \beta^{\tau_0+2^{m-2}}+\lambda',\nu+\nu')=(-1)^{\nu+1} 2^{\frac{m}{2}}$.
By Lemma \ref{lem:W-lam-lam},
 \begin{align*}
\# C_9= \left (2^{m-1}-2 \right )\cdot  2^{m-2}\left (2^{m-3}+ 2^{\frac{m-4}{2}} \right ).
\end{align*}
This completes the proof.
\end{proof}

\begin{remark}
For any cyclic bent function $f$ on $\mathbb F_{2^{m-1}}\times \mathbb F_2$,
from Theorem \ref{thm:u-f-b} and Reference \cite{THJ08}, the binary family
$\mathcal U_f^b$ has the same period $K=2(2^{m-1}-1)$, family size $N=2^{m-1}$,
maximum non-trivial correlation  $R_{\mathrm{max}}=2^{\frac{m}{2}}+2$ and
correlation value distribution as the family $\mathcal Q(2)$ of binary Kerdock
sequences \cite{HK98,TUF07}, which is an optimal family of binary sequences
with respect to the well-known Welch bound.
\end{remark}

\section{Nonlinear codes and $3$-designs from cyclic bent functions}\label{sec-Kerdockcodedesign}

Let $\cP$ be a set of $v \ge 1$ elements, and let $\cB$ be a set of $k$-subsets of $\cP$, where $k$ is
a positive integer with $1 \leq k \leq v$. Let $t$ be a positive integer with $t \leq k$. The pair
$\bD = (\cP, \cB)$ is called a $t$-$(v, k, \lambda)$ {\em design\index{design}}, or simply {\em $t$-design\index{$t$-design}}, if every $t$-subset of $\cP$ is contained in exactly $\lambda$ elements of
$\cB$. The elements of $\cP$ are called points, and those of $\cB$ are referred to as blocks.
We usually use $b$ to denote the number of blocks in $\cB$.  A $t$-design is called {\em simple\index{simple}} if $\cB$ does not contain repeated blocks. In this paper, we consider only simple
$t$-designs.  A $t$-design is called {\em symmetric\index{symmetric design}} if $v = b$. It is clear that $t$-designs with $k = t$ or $k = v$ always exist. Such $t$-designs are {\em trivial}. In this paper, we consider only $t$-designs with $v > k > t$.
A $t$-$(v,k,\lambda)$ design is referred to as a {\em Steiner system\index{Steiner system}} if $t \geq 2$ and $\lambda=1$, and is denoted by $S(t,k, v)$.

A $(v, M, d)$ code $\C$ over $\bF_2$ is a subset of $\bF_{2}^v$ with cardinality $M$
such that the minimum Hamming distance between all pairs of distinct codewords in $\C$ is $d$.
If $\C$ is a linear subspace of $\bF_{2}^v$ over $\bF_2$, we say that $\C$ is a binary
linear code.
Let $A_i$ denote the total number of codewords of Hamming weight $i$ in $\C$. The sequence
$(A_0, A_1, \ldots, A_n)$ is called the \emph{weight distribution} of $\C$.
The \emph{distance distribution} of $\C$ is the sequence $(B_0, B_1, \ldots, B_n)$,
where
$$
B_i = \frac{1}{M}|\{(\bc, \bc'): \bc, \bc' \in \C \mbox{ and } d(\bc, \bc')=i\}|
$$
and $d(\bc, \bc')$ denotes the Hamming distance between $\bc$ and $\bc'$. If
$\C$ is linear, then its weight distribution and distance distribution coincide.

The support of a codeword $\bc=(c_1, c_2, \ldots, c_v)$ is
defined by
$$
\support(\bc)=\{1 \leq i \leq v: c_i \neq 0\}.
$$
Let $\cP=\{1,2, \ldots, v\}$ and for each $1 \leq k \leq v$ define
$$
\cB(k)=\{\support(\bc): \wt(\bc)=k, \ \bc \in \C\}.
$$
If $\cB(k)$ is nonempty, then the pair $\bD = (\cP, \cB(k))$ may be a $t$-$(v, k, \lambda)$ design
for some $t$ and $\lambda$. Such a design is called a \emph{support design} of the code $\C$.
The Assmus-Mattson theorem gives sufficient conditions for a linear code to have support
designs \cite{AM74}. A generalization of the Assmus-Mattson theorem was developed for
both linear and nonlinear binary codes (see \cite[Chapter 6]{MS77} or \cite{Tonchev}).

The objective of this section is to present a family of binary codes containing the Kerdock
code and describe their support designs.

\begin{theorem}
Let $f(x_1, x_2)$ be a cyclic bent function on $\bF_{2^{m-1}} \times \bF_2$ such that
$f(0,0)=f(0,1)=0$. Define a binary code by
\begin{eqnarray}\label{eqn-gKerdockcode}
\C(f):=\{\bc_{(a, \lambda, u, v)}: a, \, \lambda \in \bF_{2^{m-1}}, \ u, \, v \in \bF_2 \},
\end{eqnarray}
where
\begin{eqnarray}
\bc_{(a, \lambda, u, v)}=
(f(ax_1, x_2) + \tr_1^{m-1}(\lambda x_1)+ ux_2+v)_{(x_1,x_2) \in \bF_{2^{m-1}} \times \bF_2}.
\end{eqnarray}
Then $\C(f)$ is a binary code with parameters $(2^m, 2^{2m}, 2^{m-1}-2^{(m-2)/2})$ and
weight distribution
\begin{eqnarray*}
&& A_0=A_{2^m}=1, \\
&& A_{2^{m-1}}=2^{m+1}-2, \\
&& A_{2^{m-1} \pm 2^{(m-2)/2}}=2^m(2^{m-1}-1).
\end{eqnarray*}
The distance distribution coincides with the weight distribution of $\C(f)$.
Further, the incidence structure $(\{1, 2, \ldots, 2^m\}, \cB(k))$ is a $3$-design, where
$$
k \in \{2^{m-1}, 2^{m-1} \pm 2^{(m-2)/2}\}.
$$
\end{theorem}

\begin{proof}
By assumption, $f(0x_1, x_2)$ is the zero constant function. Then we clearly have
\begin{eqnarray*}
\wt(\bc_{(0, \lambda, u, v)})
=\left\{
\begin{array}{ll}
0  & \mbox{ if } (\lambda, u, v)=(0,0,0), \\
2^m & \mbox{ if } (\lambda, u, v)=(0,0,1), \\
2^{m-1} & \mbox{otherwise.}
\end{array}
\right.
\end{eqnarray*}

If $a \neq 0$, by definition
$f(ax_1, x_2) + \tr_1^{m-1}(\lambda x_1)+ ux_2+v$ is always bent, as $f(ax_1, x_2)$ is bent
by definition. Hence, the Hamming weight
$\wt(\bc_{(a, \lambda, u, v)})$ is equal to $2^{m-1} \pm 2^{(m-2)/2}$ whenever $a \neq 0$.

Since $f(0x_1, x_2)$ is the zero constant function and $f(ax_1, x_2)+f(bx_1, x_2)$ is either
the zero function or a bent function, $\bc_{(a, \lambda, u, v)}$ and $\bc_{(a', \lambda', u', v')}$ are equal if and only if $(a, \lambda, u, v)=(a', \lambda', u', v')$. In addition,
the code $\C(f)$ is self-complementary. The desired conclusions on the parameters and
weight distribution then follow. It is easily verified that the weight distribution and
the distance distribution coincide.

Note that the code $\C(f)$ has the same parameters, weight distribution and distance
distribution as the Kerdock code.
In addition, the $3$-design property proof of the support designs of the Kerdock code
with the generalised Assmus-Mattson theorem depends only the parameters, weight distribution
and distance distribution  of the Kerdock code \cite[Theorem 6.1]{Tonchev}. The proof of the $3$-design property for the Kerdock code also works for the code $\C(f)$.
\end{proof}

\begin{remark}
When $f$ is the cyclic bent function $K(x_1, x_2)$ defined in (\ref{eqn-Kerdockf}), the code
$\C(f)$ becomes the Kerdock code. When $f$ is one of the other cyclic bent functions defined
in (\ref{eq:cyclic-bent}), $\C(f)$ is a new code with the same parameters, weight distribution and
distance distribution as the Kerdock code.
Hence, there are many binary codes with the same parameters as the Kerdock codes. The equivalence of these codes and their designs is open, and is a hard problem in general.
\end{remark}

\section{Cyclic semi-bent  functions}\label{sec:semi-bent}

In this section, we shall consider cyclic semi-bent functions and their applications.
Let $n$ be odd. A Boolean function $g(x)$ on $\GF {n}$ is called a \emph{cyclic semi-bent
function} if all the functions $g(a x)+g(b x)$ are semi-bent for any
$a\neq b\in \GF {n}$.

The following theorem solves
Open  problem \ref{orp:1}. It follows from the definition of cyclic semi-bent functions.
\begin{theorem}\label{thm:semi-set}
Let $n$ be an odd  positive integer and $g(x)$ be a cyclic semi-bent function on $\mathbb F_{2^{n}}$.
Define the following set
$$
\mathcal{F}=\{g(a x) \mid a\in \mathbb{F}_{2^{n}}^\star\}.
$$
Then $\mathcal{F}$ is a family of
$2^{n}-1$ semi-bent functions such that
the sum of any two distinct elements of this family is semi-bent.
\end{theorem}

\subsection{Constructing cyclic semi-bent  functions from cyclic bent functions}

The following lemma will be needed to prove Corollary \ref{cor:semi-0-1} .

\begin{lemma}\label{lem:bent-semi}
Let $m$ be an even positive integer and  $g$ be  a bent  function on $\GF{m-1}\times \GF{}$.
Then, for any $\lambda \in \GF{m-1}$,
 $\{| \mathcal W_{g(\cdot,0)} (\lambda ) |, | \mathcal W_{g(\cdot,1)} (\lambda ) | \} = \{ 0, 2^{\frac{m}{2}} \}$,
 where $\mathcal W_{g(\cdot,\epsilon)} (\lambda )  =\sum_{x \in \GF{m-1}} (-1)^{g(x,\epsilon)+\tr^{m-1}_1(\lambda x)}$ with $\epsilon \in \GF{}$.
\end{lemma}

\begin{proof}
By the definition of Walsh
transform, we have
\begin{align*}
\mathcal{W}_g(\lambda, v)
=& \sum_{x_1\in \mathbb{F}_{2^{m-1}},
x_2\in \mathbb{F}_2}(-1)^{g(x_1,x_2)+
\tr_1^{m-1}(\lambda x_1)+vx_2}\\
=& \sum_{x_1\in \mathbb{F}_{2^{m-1}}}
((-1)^{g(x_1,0)+\tr_1^{m-1}(\lambda x_1)}
+(-1)^v(-1)^{g(x_1,1)+\tr_1^{m-1}(\lambda x_1)})\\
=& \mathcal{W}_{g(\cdot ,0)}(\lambda)+
(-1)^v\mathcal{W}_{g(\cdot ,1)}(\lambda).
\end{align*}
Taking $v=0$ and $v=1$ separately, we have
\begin{equation}\label{W01}
\left\{
  \begin{array}{l}
   \mathcal{W}_{g(\cdot ,0)}(\lambda)
+\mathcal{W}_{g( \cdot ,1)}(\lambda)=
\mathcal{W}_g(\lambda,0),     \\
     \mathcal{W}_{g(\cdot ,0)}(\lambda)-
\mathcal{W}_{g(\cdot,1)}(\lambda)=
\mathcal{W}_g(\lambda,1).
  \end{array}
\right.
\end{equation}
Since $g$ is bent, we have
$$
0=(\mathcal{W}_g(\lambda,0))^2-
(\mathcal{W}_g(\lambda,1))^2
=4\mathcal{W}_{g(\cdot, 0)}(\lambda)
\mathcal{W}_{g(\cdot, 1)}(\lambda).
$$
Hence, for any $\lambda\in
\mathbb{F}_{2^{m-1}}$, we have
$\mathcal{W}_{g(\cdot, 0)}(\lambda)=0$
or $\mathcal{W}_{g(\cdot, 1)}(\lambda)=0$.
From (\ref{W01}), the desired conclusion follows.
\end{proof}

One can derive semi-bent functions from a bent function as follows.

\begin{corollary}\label{cor:semi-0-1}
Let $f$ be a bent function on $\mathbb F_{2^{m-1}}\times \mathbb F_2$. Then,
both $f(x_1,0)$ and $f(x_1,1)$ are semi-bent functions on $\mathbb F_{2^{m-1}}$.
\end{corollary}
\begin{proof}
The desired conclusion follows from Lemma \ref{lem:bent-semi}.
\end{proof}

The following theorem shows how one can derive cyclic semi-bent functions from a cyclic bent function.

\begin{theorem}\label{thm:semi-bent}
Let $f$ be a cyclic bent function on $\mathbb F_{2^{m-1}}\times \mathbb F_2$. Then,
both $f(x_1,0)$ and $f(x_1,1)$ are cyclic semi-bent functions on $\mathbb F_{2^{m-1}}$.
\end{theorem}
\begin{proof}
Let $\epsilon \in \mathbb F_2$, $a\neq  b \in \mathbb F_{2^{m-1}}$. Since $f(x_1,x_2)$ is a cyclic bent function,
$f(a x_1,x_2)+f(b x_1,x_2)$ is bent. By Corollary \ref{cor:semi-0-1}, $f(a x_1,\epsilon )+f(b x_1,\epsilon )$ is semi-bent.
Thus, $f(x_1,\epsilon)$ is a cyclic semi-bent function.
\end{proof}

Combining Theorems \ref{thm:cyc-set-even} and Corollary \ref{cor:semi-0-1}, we have the following theorem, which solves
Open  problem \ref{orp:1}.
\begin{theorem}\label{thm:cyc-semi-many}
Let $m$ be an even positive integer and $f(x_1, x_2)$ be a cyclic bent function on $\mathbb F_{2^{m-1}} \times \mathbb F_2$. Let $\epsilon_a \in \mathbb F_2$ where
$a\in \mathbb F_{2^{m-1}}$.
Define the following set
$$
\mathcal{F}=\{f(a x_1, \epsilon_a) \mid a\in \mathbb{F}_{2^{m-1}}^\star\}.
$$
Then $\mathcal{F}$ is a family of
$2^{m-1}-1$ semi-bent functions such that
the sum of any two distinct elements of this family is semi-bent.
\end{theorem}

Combining Theorems \ref{thm-F} and  \ref{thm:semi-bent}, we have the following corollary.
\begin{corollary}\label{cor:Z4-semi}
Let $m$ be an even positive integer, and $e_0$, $e_1,\ldots, e_{l-1},  e_{l}$ be a  sequence of positive integers with
$e_0=1$, $e_i|e_{i+1}$, $e_{l}=m-1$. Set
$f_i=\frac{m-1}{e_i}$. Let $\gamma_j\in \mathbb{F}_{2^{e_j}}$ satisfy
(\ref{gamma-j}). Define the quadratic Boolean function $f(x)$ from  $\mathbb{F}_{2^{m-1}}$ to
$\mathbb{F}_2$ defined by 
\begin{equation}\label{eq:cyclic-semi-bent}
f(x)=\sum_{j=0}^{l-1}Q_j(\gamma_j x),
\end{equation}
where $
Q_j(x)=\tr_{e_{0}}^{e_{l}}
(\sum_{i=1}^{\frac{f_j-1}{2}}
x^{2^{ie_j}+1}). $
Then, $f$ is a cyclic semi-bent function.
\end{corollary}
\begin{remark}
When $l=2$,  one has $e_0=1$, $e_1=e$, $e_2=m-1$, $\gamma_0=1$ and $\gamma_1=\gamma \in \mathbb F_{2^e}\setminus \{1\}$. Thus,
the cyclic semi-bent function  defined by (\ref{eq:cyclic-semi-bent}) is
$\sum_{i=1}^{\frac{m-2}{2}} \tr^{m-1}_1(x^{2^i+1})+ \sum_{i=1}^{\frac{m-1-e}{2e}} \tr^{m-1}_1((\gamma x)^{2^{ei}+1})$ and is exactly the function constructed
by Xiang et al. in \cite{XiangDingMesnager2015}.
\end{remark}

The following corollary follows from Theorem \ref{thm:semi-set} and Corollary \ref{cor:Z4-semi}.

\begin{corollary}
Let $m$ be an even positive integer, and $e_0$, $e_1,\ldots, e_{l-1},  e_{l}$ be a  sequence of positive integers defined by 
$e_0=1$, $e_i|e_{i+1}$, $e_{l}=m-1$. Set
$f_i=\frac{m-1}{e_i}$. Let $\gamma_j\in \mathbb{F}_{2^{e_j}}$ satisfy
(\ref{gamma-j}). Let $f(x)$ be the Boolean funcition from  $\mathbb{F}_{2^{m-1}}$ to
$\mathbb{F}_2$
\begin{align*}
f(x)=\sum_{j=0}^{l-1}Q_j(\gamma_j x),
\end{align*}
where $
Q_j(x)=\tr_{e_{0}}^{e_{l}}
(\sum_{i=1}^{\frac{f_j-1}{2}}
x^{2^{ie_j}+1}). $
Define the following set
$$
\mathcal{F}=\{f(a x) \mid a\in \mathbb{F}_{2^{m-1}}^\star\}.
$$
Then $\mathcal{F}$ is a family of
$2^{m-1}-1$ semi-bent functions such that
the sum of any two distinct elements of this family is semi-bent.
\end{corollary}

Combining Theorem \ref{thm:semi-set} and Theorem 4 in  \cite{XiangDingMesnager2015}, we have the following theorem.
\begin{theorem}\label{thm:semi-codebook-2}
Let $n$ be an odd positive integer and $g(x)$ be a cyclic semi-bent function on $\mathbb F_{2^{n}}$. Construct  a codebook as
\begin{align*}
\mathcal C_g= \bigcup_{a\in \mathbb F_{2^{n}}^\star} \mathcal B_a \bigcup \mathcal B_{0} \bigcup \mathcal B_{\infty},
\end{align*}
where $\mathcal B_{\infty}$ is the standard basis  of the $2^{n}$-dimensional Hilbert space $\mathbb     C^{2^{n}}$
in which each basis vector has a single nonzero entry with value $1$,
\begin{align*}
\mathcal B_{0} =\left \{\frac{1}{2^{\frac{n}{2}}} \left ( (-1)^{\tr^{n}_1(\lambda x)} \right )_{x \in \mathbb F_{2^{n}}}
\mid \lambda \in \mathbb F_{2^{n}}  \right \},
\end{align*}
and for each $a\in \mathbb F_{2^{n}}^\star$,
\begin{align*}
\mathcal B_{a} =\left \{\frac{1}{2^{\frac{n}{2}}} \left ( (-1)^{g(ax)+\tr^{n}_1(\lambda x)} \right )_{x \in \mathbb F_{2^{n}}}
\mid \lambda \in \mathbb F_{2^{n}}\right \}.
\end{align*}
Then, $\mathcal C_g$ is a real-valued $(2^{2n}+2^n, 2^n)$ codebook almost meeting the Levenshtein bound  of (\ref{Lev-r}) with the maximum crosscorrelation magnitude $\frac{1}{\sqrt{2^{n}}}$ and  alphabet size  $4$.
\end{theorem}

Using Theorem \ref{thm:cyc-semi-many} and Theorem 4 in  \cite{XiangDingMesnager2015}, we obtain the following.

\begin{corollary}\label{cor:codebook-semi-2}
Let $m$ be an even positive integer and $f(x_1, x_2)$ be a cyclic bent function on $\mathbb F_{2^{m-1}} \times \mathbb F_2$. Let $\epsilon_a \in \mathbb F_2$ where
$a\in \mathbb F_{2^{m-1}}^{\star}$. Construct  a codebook as
\begin{align*}
\mathcal C_f'= \bigcup_{a\in \mathbb F_{2^{m-1}}^\star} \mathcal B_a \bigcup \mathcal B_{0} \bigcup \mathcal B_{\infty},
\end{align*}
where $\mathcal B_{\infty}$ is the standard basis  of the $2^{m-1}$-dimensional Hilbert space $\mathbb     C^{2^{m-1}}$
in which each basis vector has a single nonzero entry with value $1$,
\begin{align*}
\mathcal B_{0} =\left \{\frac{1}{2^{\frac{m-1}{2}}} \left ( (-1)^{\tr^{m-1}_1(\lambda x_1)} \right )_{x_1 \in \mathbb F_{2^{m-1}}}
\mid  \lambda  \in \mathbb F_{2^{m-1}}  \right \},
\end{align*}
and for each $a\in \mathbb F_{2^{m-1}}^\star$,
\begin{align*}
\mathcal B_{a} =\left \{\frac{1}{2^{\frac{m-1}{2}}} \left ( (-1)^{f(ax_1, \epsilon_a)+\tr^{m-1}_1(\lambda x_1)} \right )_{x_1  \in \mathbb F_{2^{m-1}}}
\mid \lambda \in \mathbb F_{2^{m-1}} \right \}.
\end{align*}
Then, $\mathcal C_f'$ is an optimal real-valued $(2^{2(m-1)}+2^{m-1}, 2^{m-1})$ codebook almost meeting the Levenshtein bound  of (\ref{Lev-r})
with the maximum crosscorrelation magnitude $\frac{1}{\sqrt{2^{m-1}}}$ and  alphabet size  $4$.
\end{corollary}

\begin{remark}
For each cyclic bent function $f$, we can select $2^{2^{m-1}-1}$  binary vectors $(\epsilon_a)_{a\in \mathbb F_{2^{m-1}}^{\star}}$ in $ \mathbb F_2^{2^{m-1}-1}$.
Each of the vectors $(\epsilon_a)_{a\in \mathbb F_{2^{m-1}}^{\star}}$ gives an almost optimal real-valued codebook with alphabet size $4$.
\end{remark}

\subsection{A binary sequence family from cyclic semi-bent functions}

Let $n$ be an odd positive integer, $\beta$ be a primitive element
in $\mathbb F_{2^{n}}$  and $g$ be a cyclic semi-bent function on $\mathbb F_{2^{n}}$.
For $\lambda \in \mathbb F_{2^{n}}$,
define the  sequence $\{s_{\lambda}(t)\}_{t=0}^{\infty}$ by
\begin{align*}
 s_{\lambda}(t)=  (-1)^{g(\beta^t)+\tr^{n}_1 (\lambda \beta^t)}.
\end{align*}
Then, $\{s_{\lambda}(t)\}_{t=0}^{\infty}$ is a binary sequence.

A family $\mathcal U'_g$  of binary sequences is defined by
\begin{align}\label{eq:family-semi}
\mathcal U'_g =\left  \{ \{s_{\lambda}(t)\}:  \lambda \in \mathbb F_{2^{n}} \right \} \cup \left \{ \{s_{\infty}(t)\} \right \},
\end{align}
where $s_{\infty}(t) =  (-1)^{\tr^{n}_1 ( \beta^t) }$.

The correlation of the sequence family $\mathcal U'_g$ can be calculated by the following lemma.
\begin{lemma}\label{Rss-semi}
Let $g$ be a cyclic semi-bent function on $\mathbb F_{2^n}$ with $g(0)=0$, and $\tau$ be a non-negative integer with $0\le \tau <2^{n}-1$.

(1) If $\lambda, \lambda ' \in \mathbb F_{2^{n}}$, then
\begin{align*}
 R_{s_{\lambda}, s_{\lambda'}}(\tau)=\mathcal W_{g_{1,\beta^{\tau}}}(\lambda \beta^{\tau}+\lambda')-1,
 \end{align*}
where $g_{1,\beta^{\tau}}(x)=g(x)+g(\beta^{\tau} x)$.

(2) If $\lambda \in \mathbb F_{2^{n}}$, then
\begin{align*}
R_{s_{\lambda}, s_{\infty}}(\tau)= \mathcal W_{g}(\lambda +\beta^{-\tau})-1,
\end{align*}
and
\begin{align*}
R_{s_{\infty}, s_{\lambda}}(\tau)=\mathcal W_{g}(\lambda +\beta^{\tau})-1.
\end{align*}

(3) For the binary sequence $\{s_{\infty}(t)\}$, its autocorrelation at shift $\tau$ is
\begin{align*}
 R_{s_{\infty}, s_{\infty}}(\tau)= \left\{
  \begin{array}{ll}
      -1,  & \text{if~~} \tau\neq 0, \\
    2^{n} -1,  & \text{if~~} \tau=0.
  \end{array}
\right.
 \end{align*}

\end{lemma}

\begin{proof}
(1) By definition,
\begin{align*}
 R_{s_{\lambda}, s_{\lambda'}}(\tau)=& \sum_{t=0}^{2^n-2} s_{\lambda}(t+\tau) s_{\lambda'}(t)\\
 =& \sum_{t=0}^{2^n-2} (-1)^{g(\beta^{t+\tau})+\tr^{n}_1 (\lambda \beta^{t+\tau})+ g(\beta^{t})+\tr^{n}_1 (\lambda' \beta^{t})}\\
 =& \sum_{t=0}^{2^n-2}  (-1)^{g(\beta^{t})+g(\beta^{\tau}\beta^{t})+\tr^{n}_1 ((\lambda \beta^{\tau}+\lambda' ) \beta^{t})}.
\end{align*}
Thus,
\begin{align*}
 R_{s_{\lambda}, s_{\lambda'}}(\tau)=& \sum_{x\in \mathbb F_{2^n}}  (-1)^{g(x)+g(\beta^{\tau}x)+\tr^{n}_1 ((\lambda \beta^{\tau}+\lambda' ) x}-1\\
 =& \mathcal W_{g_{1,\beta^{\tau}}}(\lambda \beta^{\tau}+\lambda')-1.
\end{align*}

(2) For $R_{s_{\lambda}, s_{\infty}}(\tau)$, one has
\begin{align*}
R_{s_{\lambda}, s_{\infty}}(\tau)= & \sum_{t=0}^{2^n-2} s_{\lambda}(t+\tau) s_{\infty}(t)\\
=& \sum_{t=0}^{2^n-2} (-1)^{g(\beta^{t+\tau})+\tr^{n}_1 (\lambda \beta^{t+\tau})+\tr^{n}_1 ( \beta^{t})}\\
=& \sum_{t=0}^{2^n-2} (-1)^{g(\beta^{\tau}\beta^{t})+\tr^{n}_1 ((\lambda \beta^{\tau}+1)\beta^t)}\\
=& \sum_{x\in \mathbb F_{2^n}} (-1)^{g(\beta^{\tau} x)+\tr^{n}_1 ((\lambda \beta^{\tau}+1) x)}-1.
\end{align*}
Since $x \mapsto \beta^{-\tau} x$ is a permutation on $\mathbb F_{2^n}$, one gets
\begin{align*}
R_{s_{\lambda}, s_{\infty}}(\tau)=& \sum_{x\in \mathbb F_{2^n}} (-1)^{g( x)+\tr^{n}_1 ((\lambda +\beta^{-\tau}) x)}-1\\
=& \mathcal W_{g}(\lambda +\beta^{-\tau})-1.
\end{align*}

For $R_{s_{\infty}, s_{\lambda}}(\tau)$, one has
\begin{align*}
R_{s_{\infty}, s_{\lambda}}(\tau)=& \sum_{t=0}^{2^n-2} (-1)^{\tr^{n}_1 ( \beta^{t+\tau}) +g(\beta^{t})+\tr^{n}_1 (\lambda \beta^{t})}\\
=& \sum_{x\in \mathbb F_{2^n}} (-1)^{g( x)+\tr^{n}_1 ((\lambda +\beta^{\tau}) x)}-1\\
=& \mathcal W_{g}(\lambda +\beta^{\tau})-1.
\end{align*}

(3) For the autocorrelation of  $\{s_{\infty}(t)\}$, one has
\begin{align*}
 R_{s_{\infty}, s_{\infty}}(\tau)=&\sum_{t=0}^{2^n-2} (-1)^{\tr^{n}_1 (( \beta^{\tau}+1 ) \beta^{t})}\\
 =&\sum_{x\in \mathbb F_{2^n}} (-1)^{\tr^{n}_1 (( \beta^{\tau}+1 )x)}-1\\
 =&
  \left\{
  \begin{array}{ll}
      -1,  & \text{if~} \tau\neq 0, \\
    2^{n} -1,  & \text{if~} \tau=0.
  \end{array}
\right.
 \end{align*}
\end{proof}

The proof of the following lemma can be found in \cite{Mes15}.

\begin{lemma}\label{lem:Walsh-semi}
Let $g$ be a  semi-bent
function on $\mathbb F_{2^n}$ with $g(0)=0$. Then the
Walsh distribution of $g$ is listed in the
following Table \ref{tab:semi-Walsh}:
\begin{table}[htbp]
  \centering
\caption{The Walsh distribution of semi-bent function $g$}
\label{tab:semi-Walsh}
\begin{tabular}{|c|c|}
\hline
$\mathcal W_g(\cdot)$& frequncy\\
\hline
0& $2^{n-1}$\\
\hline
$2^{\frac{n+1}{2}}$ & $2^{n-2}+ 2^{\frac{n-3}{2}}$\\
\hline
$-2^{\frac{n+1}{2}}$ & $2^{n-2}-  2^{\frac{n-3}{2}}$\\
\hline
\end{tabular}
\end{table}
\end{lemma}

We have the following main result on the correlation distribution of the binary  sequence  family from cyclic semi-bent functions.
\begin{theorem}\label{thm:semi-sequence}
Let $g$ be a cyclic semi-bent function on $\mathbb F_{2^{n}}$ such that $g(0)=0$.
Then,
the correlation value distribution of family $\mathcal U'_{g}$ is given in Table \ref{table-Uf-b-semi}.

\begin{table}[htbp]
\centering
\caption{The correlation distribution of the family $\mathcal U'_g$}
\label{table-Uf-b-semi}
\begin{tabular}{|c|c|}
  \hline
  Value & Frequency \\
  \hline
  $2^{n}-1$ & $2^n+1$\\
\hline
$-1$ &  $2^{n+1}\left (2^n-1 \right )+ \left (2^n-2 \right ) \left (2^{2n-1}+1 \right )$\\
\hline
$2^{\frac{n+1}{2}}-1$& $\left (2^{2n}-2 \right ) \left ( 2^{n-2} +    2^{\frac{n-3}{2}} \right) $\\
\hline
$-2^{\frac{n+1}{2}}-1$& $\left (2^{2n}-2 \right ) \left ( 2^{n-2} -    2^{\frac{n-3}{2}} \right) $\\
\hline
\end{tabular}
\end{table}

\end{theorem}
\begin{proof}
We distinguish among the following four cases.

 Case 0: Let $C_0$ be the set of all $ (\tau,\lambda,\lambda')$  with $\lambda, \lambda' \in \mathbb F_{2^n} \cup \{\infty\}$,
 $0\le \tau <2^{n}-1$ and
 $R_{s_{\lambda}, s_{\lambda'}}(\tau)=2^n-1$. By Lemma \ref{Rss-semi}, $ (\tau,\lambda,\lambda')\in C_0$  if and only if
 $\tau=0$,  $\lambda =\lambda'$.
 Thus, $\# C_0= 2^n+1$.

Case 1: Let $C_1$ be the set of all $ (\tau,\lambda,\lambda')$  with $\lambda, \lambda' \in \mathbb F_{2^n} \cup \{\infty\}$,
 $0\le \tau <2^{n}-1$ and
 $R_{s_{\lambda}, s_{\lambda'}}(\tau)=-1$. By Lemma \ref{Rss-semi}, $ (\tau,\lambda,\lambda')\in C_0$  if and only if
 $\tau=0$,  $\lambda \neq \lambda' \in \mathbb F_{2^n}$, or,
 $\tau \neq 0$,  $\mathcal W_{g_{1,\beta^{\tau}}}(\lambda \beta^{\tau}+\lambda')=0$, or, $\mathcal W_{g}(\lambda +\beta^{-\tau})=0, \lambda'=\infty$, or, $\mathcal W_{g}(\lambda' +\beta^{\tau})=0, \lambda=\infty$, or, $\tau \neq 0$, $\lambda=\lambda'=\infty$.
By Lemma \ref{lem:Walsh-semi}, $\# C_1= 2^n\left (2^n-1 \right )+ (2^n-2)\cdot 2^n \cdot 2^{n-1}+2\cdot (2^n-1)\cdot 2^{n-1}+(2^n-2)$,
 that is, $\# C_1 =2^{n+1}\left (2^n-1 \right )+ \left (2^n-2 \right ) \left (2^{2n-1}+1 \right )$.

 Case 2: Let $C_2$ be the set of all $ (\tau,\lambda,\lambda')$  with $\lambda, \lambda' \in \mathbb F_{2^n} \cup \{\infty\}$,
 $0\le \tau <2^{n}-1$ and
 $R_{s_{\lambda}, s_{\lambda'}}(\tau)=2^{\frac{n+1}{2}}-1$. By Lemma \ref{Rss-semi}, $ (\tau,\lambda,\lambda')\in C_0$  if and only if
 $\tau \neq 0$,  $\mathcal W_{g_{1,\beta^{\tau}}}(\lambda \beta^{\tau}+\lambda')=2^{\frac{n+1}{2}}$, or,
 $\mathcal W_{g}(\lambda +\beta^{-\tau})=2^{\frac{n+1}{2}}, \lambda'=\infty$, or,
 $\mathcal W_{g}(\lambda' +\beta^{\tau})=2^{\frac{n+1}{2}}, \lambda=\infty$.
 By Lemma \ref{lem:Walsh-semi},
 $\# C_2= \left (2^n-2 \right ) 2^n \left ( 2^{n-2} +    2^{\frac{n-3}{2}} \right )    +2\cdot \left (2^n-1 \right )\left ( 2^{n-2} +    2^{\frac{n-3}{2}} \right)$, that is, $\# C_2 =\left (2^{2n}-2 \right ) \left ( 2^{n-2} +    2^{\frac{n-3}{2}} \right) $.

Case 3: Let $C_3$ be the set of all $ (\tau,\lambda,\lambda')$  with $\lambda, \lambda' \in \mathbb F_{2^n} \cup \{\infty\}$,
 $0\le \tau <2^{n}-1$ and
 $R_{s_{\lambda}, s_{\lambda'}}(\tau)=-2^{\frac{n+1}{2}}-1$. By Lemma \ref{Rss-semi}, $ (\tau,\lambda,\lambda')\in C_0$  if and only if
 $\tau \neq 0$,  $\mathcal W_{g_{1,\beta^{\tau}}}(\lambda \beta^{\tau}+\lambda')=-2^{\frac{n+1}{2}}$, or,
 $\mathcal W_{g}(\lambda +\beta^{-\tau})=-2^{\frac{n+1}{2}}, \lambda'=\infty$, or,
 $\mathcal W_{g}(\lambda' +\beta^{\tau})=-2^{\frac{n+1}{2}}, \lambda=\infty$.
 By Lemma \ref{lem:Walsh-semi},
 $\# C_3= \left (2^n-2 \right ) 2^n \left ( 2^{n-2} -    2^{\frac{n-3}{2}} \right )    +2\cdot \left (2^n-1 \right )\left ( 2^{n-2} -    2^{\frac{n-3}{2}} \right)$, that is, $\# C_3 =\left (2^{2n}-2 \right ) \left ( 2^{n-2} -   2^{\frac{n-3}{2}} \right) $.
This completes the proof.
\end{proof}

\begin{remark}
For any cyclic semi-bent function $g$ on $\mathbb F_{2^n}$, according to Theorem \ref{thm:semi-sequence}, the family $\mathcal U_g'$ defined in
(\ref{eq:family-semi}) has period $K=2^n-1$, size $N=2^n+1$ and maximum non-trivial  correlations $R_{\mathrm{max}}=1+\sqrt{2^{n+1}}$.
Thus,  the binary family $\mathcal U_g'$ has the same period, family size and maximum non-trivial  correlations as the family of Gold sequences, which is optimal in the sense of having the lowest value of $R_{\mathrm{max}}$ possible for a family of binary sequences for the given period $K=2^n-1$ and family
size $N=2^n+1$ \cite{Gold68}.
\end{remark}

\subsection{Binary codes from cyclic semi-bent functions}\label{sec-Kerdockcodelike}

The objective of this section is to present a family of  binary codes from cyclic semi-bent function which have the best parameters known.

\begin{theorem}\label{thm-nov131}
Let $g(x)$ be a cyclic semi-bent function on $\bF_{2^{n}}$ such that
$g(0)=0$. Define a binary code by
\begin{eqnarray}\label{eqn-Kerdockcodelike}
\C(g):=\{\bc_{(a, \lambda, u)}: a, \, \lambda \in \bF_{2^{n}}, \ u \in \bF_2 \},
\end{eqnarray}
where
\begin{eqnarray}
\bc_{(a, \lambda, u)}=
(g(ax) + \tr_1^{n}(\lambda x)+ u)_{x \in \bF_{2^{n}}}.
\end{eqnarray}
Then $\C(g)$ is a binary code with parameters $(2^n, 2^{2n+1}, 2^{n-1}-2^{(n-1)/2})$ and
weight distribution
\begin{eqnarray*}
&& A_0=A_{2^n}=1, \\
&& A_{2^{n-1}}=2^{2n}+2^n-2, \\
&& A_{2^{n-1} \pm 2^{(n-1)/2}}=2^{2n-1}-2^{n-1}.
\end{eqnarray*}
\end{theorem}

\begin{proof}
By definition, $g(ax)+g(bx)$ is either the zero constant function or a semi-bent function.
It then follows that $\bc_{(a, \lambda, u)}=\bc_{(a', \lambda', u')}$ if and only if
$(a, \lambda, u)=(a', \lambda', u')$. Consequently, $|\C(g)|=2^{2n+1}$.

Note that $g(ax) + \tr_1^{n}(\lambda x)+ u$ is either a semi-bent function or an affine
function. We deduce that $\C(g)$ has only weights $0, 2^n, 2^{n-1}, 2^{n-1} \pm 2^{(n-1)/2}$.
For each fixed $a \in \bF_{2^m}^\star$, by Table \ref{tab:semi-Walsh} there are $2^n$
codewords $\bc_{(a, \lambda, u)}$ with Hamming weight $2^{n-1}$. There are $2(2^n-1)$
codewords $\bc_{(0, \lambda, u)}$  with Hamming weight $2^{n-1}$. Hence,
$$
A_{2^{n-1}} = (2^n-1)2^n+2(2^n-1)=2^{2n}+2^n-2.
$$
Obviously, $A_0=A_{2^n}=1$, as $g(ax)$ is semi-bent for all $a \ne 0$. It then follows
that
$$
A_{2^{n-1} \pm 2^{(n-1)/2}}=\frac{2^{2n+1}-A_0-A_1-A_{2^{m-1}}}{2}=2^{2n-1}-2^{n-1}.
$$

\end{proof}

\begin{remark}
When $g$ is the cyclic semi-bent function $\mathrm{Tr}^n_{1}(x^3)$, the code
$\C(g)$ becomes a linear binary code. When $g$ is one of the other cyclic semi-bent functions defined
in (\ref{eq:cyclic-semi-bent}), $\C(g)$ is a nonlinear binary code.
Hence, $\C(g)$ could be linear or nonlinear.
\end{remark}

\begin{remark}
The code $\C(g)$ in Theorem \ref{thm-nov131} has the same parameters and weight distribution
as a linear code in \cite[Theorem 17]{Dingjcd}, and has the best parameters known.
The
code $\C(g)$ can be viewed as the odd-case counterpart of the Kerdock code. We guess
that $\C(g)$ has support designs of strength $2$. It may have support designs of strength
$3$. We leave this as an open problem.
\end{remark}

\subsection{The characterization of quadratic cyclic semi-bent  functions}\label{subsec:char}

Let $U$ be a linearized polynomial over $\GF m$ and let $U^\star$ denote its adjoint
introduced in Section II. The following open  problem is closely related to Open  problem \ref{orp:1}.

\begin{orp}\label{orp:2}
Let $m$ be odd.
 Find a family of $\{L_a:=U_a+U_a^\star,a\in\GF m^\star\}$ of $2^m-1$ linearized polynomials over $\mathbb F_{2^m}$ such that
 \begin{enumerate}
 \item for every $a\in\GF m^\star$, $\dim_{\GF{}}\ker L_a=1$; and
 \item for every $a$, $b\in\GF m^\star$ with $a\not=b$, $\dim_{\GF{}}\ker (L_a+L_b)=1$.
 \end{enumerate}
\end{orp}

Below is a connection between the two problems.

\begin{theorem}
Let $\{L_a:=U_a+U_a^\star,a\in\GF m^\star\}$ be a family of size $2^m-1$ which is a solution of Open problem \ref{orp:2}.
For every $a\in\GF m^\star$, define $f_a(x)=\Tr m(xU_a(x))$. Then $\{f_a,a\in\GF m^\star\}$ is a solution to Open problem \ref{orp:1}.
\end{theorem}

\begin{proof}
Observe that $f_a(x+y)+f_a(x)+f_a(y)+f_a(0)=\Tr m(xU_a(y)+yU_a(x))=\Tr m (yL_a(x))$. Thus, $\dim_{\GF{}}\ker B_{f_a}=
 \dim_{\GF{}}L_a=1$ for every $a\in\GF m^\star$. By Theorem \ref{thm:semi-bent}, one concludes that every $f_a$ is semi-bent. Now, let $a\not=b$. One has $f_a(x+y)+f_b(x+y)+f_a(x)+f_b(x)+f_a(y)+f_b(y)+f_a(0)+f_b(0)
 =\Tr m (y(L_a(x)+L_b(x)))$. Then $f_a+f_b$ is semi-bent again
 since $\dim_{\GF{}}\ker B_{f_a+f_b}=\dim_{\GF{}}\ker(L_a+L_b)=1$.
The desired conclusion then follows.
\end{proof}

Let $m$ be an odd positive integer, $L$ be a linearized polynomial over $\GF m$, and $Q$ be a quadratic form over
$\GF m$ defined as $Q(x)=\Tr m(xL(x))$. For every $\lambda\in\GF m^\star$, define
\begin{equation*}
  f_\lambda(x) = Q(\lambda x).
\end{equation*}
Let $\mathcal S(Q)$ be the collection of all $f_\lambda$ :
\begin{equation*}
  \mathcal S(Q)=\{f_\lambda,\,\lambda\in\GF{m}^\star\}.
\end{equation*}
Let  $B_Q$ be  the symmetric bilinear form over $\GF m$ associated to $Q$, that is, the symmetric
bilinear form over $\GF m$ defined by
\begin{displaymath}
  B_Q(x,y) = Q(x+y) + Q(x) + Q(y) + Q(0).
\end{displaymath}
Notice the following relation between $B_Q$ and
$B_{f_\lambda}$:
\begin{eqnarray*}
  B_{f_\lambda}(x, y) = B_Q(\lambda x,\lambda y).
\end{eqnarray*}
For a symmetric bilinear form $B$ over $\GFm$, let $\ker B$ be the
kernel of $B$ defined as
\begin{displaymath}
  \ker B=\{x\in\GFm\mid \forall z\in\GF m,\,B(x,z)=0\}.
\end{displaymath}

We have the following characterization of semi-bent quadratic functions.

\begin{proposition}~\label{thm:characterization}
   The quadratic function $Q$ is semi-bent if and only if $\dim\ker B_Q=1$.
\end{proposition}

\begin{proof}
For $A\in \mathbb F_2$, $\chi (A)$ denotes $ (-1)^A$. The square of the modulus of the Walsh transform of $Q$ at $w\in\GF m$ is given by
  \begin{eqnarray*}
    \left\vert  \mathcal W_Q(w) \right\vert^2
    &=&\sum_{x,y\in\GFm}\chi\left(Q(x)+Q(y)+\Tr m(w(x+y))\right)\\
    &=& \sum_{x,y\in\GFm}\chi\left(Q(x+y)+B_Q(x,y)+\Tr m(w(x+y))\right)\\
    &=& \sum_{x,z\in\GFm}\chi\left(Q(z)+B_Q(x,x+z)+\Tr m(wz)\right)\\
    &=& \sum_{x,z\in\GFm}\chi\left(Q(z)+B_Q(x,z)+\Tr m(wz)\right)\\
    &=& \sum_{z\in\GFm}\chi\left(Q(z)+\Tr m(wz)\right)\sum_{x\in\GFm}\chi\left(B_Q(z,x)\right)\\
    &=& \qm\sum_{z\in\ker B_Q}\chi\left(Q(z)+\Tr m(wz)\right).
  \end{eqnarray*}
  Observe now that $Q$ is linear on its kernel $\ker B_Q$. Hence
  \begin{eqnarray*}
    \left\vert\mathcal W_Q(w) \right\vert^2 = \cc^{m+\dim\ker B_Q} \gamma_E(w),
  \end{eqnarray*}
  where $E=\{w\in\GF{m}\mid \forall z\in\ker B_Q,\, Q(z)+\Tr m(wz)=0\}$ and $\gamma_E(w)=\begin{cases}
0, &w \not \in E\\
1, &w \in E
\end{cases}
$.
  Hence, $Q$ is semi-bent if and only if $\dim\ker B_Q=1$.
 \end{proof}
In fact, the condition on the kernel of $B_Q$ stated by Theorem \ref{thm:characterization}
implies that every element $f_\lambda$ of $\mathcal S(Q)$ is also semi-bent  since
\begin{eqnarray*}
  \ker B_{f_\lambda}
  &=& \{x\in\GF m\mid \forall z\in\GFm,\,B_Q(\lambda x,\lambda z)=0\}\\
  &=& \{x\in\GF m\mid \forall z\in\GFm,\,B_Q(\lambda x,z)=0\}\\
  &=& \{x\in\GF m\mid \lambda x\in\ker B_Q\}\\
  &=& \lambda^{-1}\ker B_Q,
\end{eqnarray*}
which yields that
\begin{equation*}
  \dim \ker B_{f_\lambda} = \dim \ker B_Q.
\end{equation*}
Hence, we have the following proposition.
\begin{proposition}
  $\mathcal S(Q)$ is a collection of $\cc^m-1$ semi-bent functions if and only if $\dim\ker B_Q=1$.
 \end{proposition}

We now investigate conditions under which the sum of any of two
different elements of $\mathcal S(Q)$ is again semi-bent.

Observe that
\begin{displaymath}
   g_{\lambda,\mu}(x) = f_\lambda(x) + f_\mu(x) \s= f_{\lambda+\mu}(x) + B_q(\lambda x,\mu x).
\end{displaymath}
The symplectic form of $g_{\lambda,\mu}$ is therefore equal to
\begin{eqnarray*}
  B_{g_{\lambda,\mu}}(x,y)
  &=& B_q((\lambda+\mu)x,(\lambda+\mu)y) + B_q(\lambda x,\mu y)+B_q(\lambda y,\mu x)\\
  &=& B_q(\lambda x,\lambda y)+B_q(\mu x,\mu y).
\end{eqnarray*}
Hence,
\begin{eqnarray*}
  \ker B_{g_{\lambda,\mu}}
  &=& \{x\in\GF m\mid \forall z\in\GFm,\,B_{g_{\lambda,\mu}}(x,z)=0\}\\
  &=& \{x\in\GF m\mid \forall z\in\GFm,\, B_q(\lambda x,\lambda z)+B_q(\mu x,\mu z) = 0\}.
\end{eqnarray*}
Without loss of generality, suppose that $q(x)=\Tr m(xL(x))$,
where $L$ is a linear map from $\GF{m}$ to itself. Let $L^\star$ be
the linear map uniquely defined by
\begin{displaymath}
  \forall (x,y)\in\GF m, \Tr m(xL(y)) = \Tr m(L^\star(x)y).
\end{displaymath}
We shall call $L^\star$ the dual map (i.e., the adjoint) of $L$. Then,
\begin{displaymath}
  B_q(x,y) = \Tr m(xL(y)+L(y)x)=\Tr m((L(x)+L^\star(x))y).
\end{displaymath}
Hence,
\begin{displaymath}
  \ker B_q = \ker (L+L^\star) = \{x\in\GF m\mid L(x)+L^\star(x)=0\}.
\end{displaymath}
On the other hand,
\begin{eqnarray*}
  B_q(\lambda x,\lambda z)+B_q(\mu x,\mu z)
  &=& \Tr m((L(\lambda x)+L^\star(\lambda x))\lambda z) +
      \Tr m((L(\mu x)+L^\star(\mu x))\mu z).
\end{eqnarray*}
Hence,
\begin{eqnarray*}
  \ker B_{g_{\lambda,\mu}}
  &=& \{x\in\GF m\mid \lambda (L(\lambda x)+L^\star(\lambda x))
      + \mu (L(\mu x)+L^\star(\mu x))=0\}.
\end{eqnarray*}

\begin{proposition}
  Let $m$ be an  odd positive integer. For every $x\in \GF m$, set $q(x)=\Tr m (xL(x))$, where  $L$ is linear on $\GF m$ satisfying the following two conditions:
  \begin{enumerate}
  \item $\dim\ker (L+L^\star)=1$.
  \item For every $\tau\in\GF m\setminus\GF{}$,
    $\dim\ker \phi_{L,\tau}=1$, where
    $\phi_{L,\tau}(x)=L(x)+L^\star(x) + \tau (L(\tau x) +L^\star(\tau
    x))$.
  \end{enumerate}
  Then $\{f_\lambda(x):=q(\lambda x),\lambda\in\GF m^\star\}$  is a family
  of $2^m-1$ semi-bent functions such that $f_\lambda+f_\mu$ is
  semi-bent for every $\lambda\not=\mu$.
\end{proposition}

\begin{proof}
  Note that, $x\in\GF m^\star$ is a solution of
  $ \lambda (L(\lambda x)+L^\star(\lambda x)) + \mu (L(\mu
  x)+L^\star(\mu x))=0$
  if and only if
  $\lambda x (L(\lambda x)+L^\star(\lambda x)) + \mu x (L(\mu
  x)+L^\star(\mu x))=0$
  if and only if $y=\lambda x$ is a solution of
  $y (L(y)+L^\star(y)) + \tau y (L(\tau y)+L^\star(\tau y))=0$, where
  $\tau=\mu\lambda^{-1}$.
\end{proof}
\begin{example}
  Let $m$ be an odd positive integer and $q(x)=x^{2^i+1}$ with $\gcd(i,m)=1$. One has
  $L(x)=x^{2^i}$ and $L^\star(x)=x^{2^{-i}}$. Hence $L(x)+L^\star(x)=x^{2^i}+x^{2^{-i}}
  =\left(x^{2^{2i}}+x\right)^{2^{-i}}$. It is well known that the kernel of
  $x\in\GF m\mapsto x^{2^{2i}}+x$ is $\GF{}$ when $\gcd(2i,m)=\gcd(i,m)=1$.

  On the other hand,
  $\phi_{\tau,L}=x^{2^i}+x^{2^{-i}} + \tau\left((\tau x)^{2^i}+(\tau
    x)^{2^{-i}}\right)
  = \left((1+\tau^{2^i+1})^{2^i} x^{2^{2i}} +
    (1+\tau^{2^i+1})x\right)^{2^{-i}}$. The kernel is of dimension $1$ since $1+\tau^{2^i+1}\not=0$
  whenever $\tau\not\in\GF{}$.
\end{example}

  Let $L(x)=\sum_{i=0}^{m-1} a_ix^{2^i}$ be a linearized polynomial over $\GF m$. Its \emph{associated polynomial} is the
  polynomial $l(x)=\sum_{i=0}^{m-1} a_ix^{i}$.

A well-known result about the dimension of the kernel of linearized polynomials is given in the following proposition \cite{WL13}.

\begin{proposition}\label{Lin}
  Let $L(x)=\sum_{i=0}^{m-1} a_ix^{2^i}$ be a linearized polynomial on $\GF m$ and $l$ be its associated polynomial.
  Then $\dim\ker L=\deg \mathrm{gcrd}(l,x^m-1)$,  where $\mathrm{gcrd}$ denotes the greatest common right divisor
  of two polynomials.
\end{proposition}

\begin{proposition}\label{pop:rgcd}
Let $m$ be an odd integer. Set $q(x)=\Tr m(xL(x))$, where $L(x) = \sum_{i=0}^{m-1} a_ix^{2^i},\, a_i\in\GF m, x\in\GF m$, such that
\begin{enumerate}
  \item $\deg \mathrm{gcrd} (l+l^\star, x^n-1)=1$, where $l+l^\star(x)=\sum_{i=1}^{m-1}\left(a_i+ a_{m-i}^{2^i}\right) x^{i}$, and
  \item $\deg \mathrm{gcrd} (\phi_{l,\tau}, x^n-1)=1$, where $\phi_{l,\tau}(x)=\sum_{i=1}^{m-1}\left(a_i+ a_{m-i}^{2^i}\right) \left(1+\tau^{2^i+1}\right) x^{i}$.
\end{enumerate}
 Then $\{f_\lambda: f_\lambda(x)=q(\lambda x), \lambda\in\GF m^\star\}$ is a family
  of $2^m-1$ semi-bent functions such that $f_\lambda+f_\mu$ is
  semi-bent for every $\lambda\not=\mu$.
  \end{proposition}

  \begin{proof}
From the calculations above, we obtain
\begin{displaymath}
  \Tr m (yL(x)) = \Tr m\left(x\sum_{i=0}^{m-1} a_i^{2^{-i}} y^{2^{-i}}\right)=\Tr m\left(x\left(a_0y+\sum_{i=1}^{m-1}a_i^{2^{m-i}}x^{2^{m-i}}\right)\right).
\end{displaymath}
Hence,
\begin{displaymath}
  L^\star(x) = a_0 x + \sum_{i=1}^{m-1} a_{m-i}^{2^i} x^{2^i}.
\end{displaymath}
Then,
\begin{displaymath}
  L(x) + L^\star(x) = \sum_{i=1}^{m-1}\left(a_i+ a_{m-i}^{2^i}\right) x^{2^i}.
\end{displaymath}
On the other hand, given $\tau\in\GF m$, we have
\begin{displaymath}
  L(x) + L^\star(x) + \tau \left(L(\tau x) + L^\star(\tau x)\right)=
  \sum_{i=1}^{m-1}\left(a_i+ a_{m-i}^{2^i}\right) \left(1+\tau^{2^i+1}\right) x^{2^i}.
\end{displaymath}
The desired conclusion then follows from Proposition \ref{Lin}.
\end{proof}

With Proposition \ref{pop:rgcd},  we have the following characterization of quadratic cyclic semi-bent  functions.

\begin{theorem}\label{thm-newcharactsemibcyc}
Let $m$ be an odd integer. Set $q(x)=\Tr m(xL(x))$, where $L(x) = \sum_{i=0}^{m-1} a_ix^{2^i},\, a_i\in\GF m, x\in\GF m$. Then, $q(x)$
is a cyclic semi-bent function if and only if   the following two conditions hold:
\begin{enumerate}
  \item $\deg \mathrm{gcrd} (l+l^\star, x^n-1)=1$, where $l+l^\star(x)=\sum_{i=1}^{m-1}\left(a_i+ a_{m-i}^{2^i}\right) x^{i}$, and
  \item $\deg \mathrm{gcrd} (\phi_{l,\tau}, x^n-1)=1$, where $\phi_{l,\tau}(x)=\sum_{i=1}^{m-1}\left(a_i+ a_{m-i}^{2^i}\right) \left(1+\tau^{2^i+1}\right) x^{i}$.
\end{enumerate}

  \end{theorem}

\section{Summary and concluding remarks}\label{sec:Conc}

The main contributions of this paper are the following:
\begin{enumerate}
\item Cyclic bent functions and cyclic semi-bent functions were introduced.

\item A general construction of cyclic bent functions and cyclic semi-bent functions was
          introduced in Theorem \ref{thm-F} and Corollary  \ref{cor:semi-0-1}, respectively.

\item A construction of a complete set of MUBs in the complex vector space $\mathbb C^{2^{m-1}}$ from a cyclic bent function was given in Theorem \ref{thm:MUBs}, for even $m$.

\item Two constructions of codebooks from cyclic bent functions and cyclic semi-bent functions were presented. Cyclic bent functions were employed to construct $2^{2^{m-1}-1}$ optimal real-valued $(2^{2m-1}+2^m, 2^m)$ codebooks (Theorem \ref{thm:codebook-2}) with alphabet size $4$ meeting the Levenshtein bound  of (\ref{Lev-r}) and
an optimal  $(2^{2(m-1)}+2^{m-1}, 2^{m-1})$ complex-valued codebook (Theorem \ref{thm:codebook-4}) with alphabet size $6$ meeting the Levenshtein bound
of (\ref{Lev-c}) for even $m$. Further, cyclic semi-bent functions were used to construct a real-valued $(2^{2n}+2^n, 2^n)$ codebook (Theorem \ref{thm:semi-codebook-2}) with alphabet size  $4$ almost meeting the Levenshtein bound, where $n$ is odd.

\item A general construction of optimal or asymptotically optimal families of sequences
from cyclic bent functions and cyclic semi-bent functions was developed. With this construction, we obtained a family $\mathcal U_f$  of quaternary sequences (Theorem \ref{thm:U-f})
of period $K=2^{m-1}-1$, family size $N=2^{m-1}+1$,  maximum correlation $R_{\mathrm{max}}\le 1+\sqrt{2^{m-1}}$, and a family $\mathcal U_f^b$ of binary sequences (Theorem \ref{thm:u-f-b})
of period $K=2(2^{m-1}-1)$, family size $N=2^{m-1}$, maximum non-trivial correlation  $R_{\mathrm{max}}=2^{\frac{m}{2}}+2$  via cyclic bent functions.
Using a cyclic semi-bent function, we got a family $\mathcal U_g'$ of binary sequences (Theorem \ref{thm:semi-sequence}) of
 period $K=2^n-1$, size $N=2^n+1$ and maximum non-trivial  correlations $R_{\mathrm{max}}=1+\sqrt{2^{n+1}}$.
 All the families of sequences  constructed in this paper  almost meet the Welch bound  and are asymptotically optimal.
 We also completely determined the correlation value distributions of these  families of sequences.
\item A characterization of quadratic cyclic semi-bent functions was developed in Theorem
\ref{thm-newcharactsemibcyc}.
\item A generalization of the Kerdock code with cyclic bent functions was given in
Section \ref{sec-Kerdockcodedesign}, where many new codes with the same parameters and
weight distribution as the Kerdock code were obtained and a lot of new $3$-designs
were derived.
\item A family of  binary codes with parameters $(2^n, 2^{2n+1}, 2^{n-1}-2^{(n-1)/2})$
was constructed with cyclic semi-bent functions in Section \ref{sec-Kerdockcodelike}.
\end{enumerate}

As demonstrated in this paper, cyclic bent functions not only give
rise to  sets of bent functions meeting the requirements in  Open problem \ref{orp:0} and  sets of semi-bent functions meeting the requirements in  Open problem \ref{orp:1}, but also lead  to complete sets of MUBs,
optimal or asymptotically optimal   codebooks with respect to the Levenshtein bound, and asymptotically optimal   families of sequences almost meeting the Welch bounds.
Cyclic bent functions introduced and developped in this paper may be used in  symmetric cryptography and compressed sensing \cite{LG14}. These interesting applications are the motivations of this paper.

Though a lot of bent and semi-bent functions are available in the literature,
only a few families of cyclic bent functions and cyclic semi-bent functions are known.
It would be interesting to construct more cyclic bent functions and cyclic semi-bent functions.
Note that the algebraic degree of all
the constructed cyclic bent functions is equal to $2$. Hence we present the following two  open problems:
\begin{orp}
Completely classify quadratic cyclic bent functions and quadratic cyclic semi-bent functions.
\end{orp}
\begin{orp}
Construct classes of cyclic bent functions on  $\mathbb F_{2^{m-1}} \times \mathbb F_2$ of algebraic degree greater than $2$.
\end{orp}

%







\begin{thebibliography}{99}


\bibitem{AM74}
E. F. Assmus Jr., H. F. Mattson Jr., ``Coding and combinatorics," SIAM
Rev. 16, pp. 349--388, 1974.

\bibitem{BBRV02} S. Bandyopadhyay, P. O. Boykin, V. Roychowdhury, and F. Vatan, ``A new proof of the existence of mutually unbiased bases,"  Algorithmica, vol. 34, no. 4, pp. 512-528, 2002.

\bibitem{BHK92} S. Boztas and R. Hammons, P. Kumar,  ``4-phase sequences with near-optimum correlation properties,"
    IEEE Trans. Inf. Theory, vol. 38,
    no. 3, pp. 1101--1113, 1992.


\bibitem{Car10} C. Carlet, ``Boolean functions for cryptography and error correcting codes,"  In: Boolean models and methods in mathematics, computer science, and engineering, vol. 2, pp. 257--397, 2010.


    \bibitem{CarletMesnagerDCC2016}
C. Carlet and S. Mesnager, ``Four decades of research on bent functions,"
Designs, Codes and Cryptography, vol. 78, no. 1, pp. 5--50, 2016.


\bibitem{CCKS97} A. R. Calderbank, P. J. Cameron, W. M. Kantor, and J. J. Seidel, ``$\mathbb Z_4$-Kerdock codes, orthogonal spreads, and extremal Euclidean line-sets," Proc. London Math. Soc., vol. 75, no. 3, pp. 436--480, 1997.

\bibitem{Dingjcd} C. Ding, ``An infinite family of Steiner systems $S(2,4,2^m)$ from cyclic codes," J. Combinatorial Designs, vol. 26, no. 3, pp. 127--144,
March 2018.

\bibitem{DingYin07}
C. Ding and J. Yin, ``Signal sets from functions with
optimum nonlinearity," IEEE Trans. Communications, vol. 55, no. 5, pp. 936--940, 2007.

\bibitem{Gold68} R. Gold, ``Maximal recursive sequences with 3-valued recursive  cross-correlation functions," IEEE Trans. Inf. Theory, vol. 14, no. 1,
pp. 154--156, 1968.

\bibitem{HK98} T. Helleseth and P. Kumar,
``Sequences with low correlation," In: Pless, V., Huffman, C. (Eds.), Handbook of Coding Theory, Elsevier, Amsterdam, 1998.

\bibitem{HY17} Z. Heng and Q. Yue, ``Optimal codebooks achieving the Levenshtein bound from generalized bent functions over $\mathbb {Z}_{4}$,"
Cryptography and Communications,
vol. 9, no. 1, pp. 41--53, 2017.

\bibitem{Iva81} I. D. Ivanovi$\acute c$,  ``Geometrical description of quantal state determination," J. Phys. A, vol. 14, pp. 3241--3245, 1981.

\bibitem{Ker72} A. M. Kerdock, ``A class of low-rate nonlinear binary codes," Inform.
Contr., vol. 20, no. 2, pp. 182--187, 1972.

\bibitem{KL78} G. A. Kabatyanskii and V. I. Levenshtein, ``Bounds for packings on a sphere and in space," Problems Inf. Transm., vol. 14, pp. 1--17, 1978.

\bibitem{Lev83} V. I. Levenshtein, ``Bounds for packings of metric spaces and some of
their applications,"  Problems Cybern., vol. 40, pp. 43--110, 1983.

\bibitem{LG14} S. Li and G. Ge, ``Deterministic sensing matrices arising from
near orthogonal systems," IEEE Trans. Inf. Theory, vol. 60, no. 4,
pp. 2291--2302,  2014.

\bibitem{MS77}
F. J. MacWilliams, N. J. A. Sloane, \emph{The Theory of Error-Correcting Codes},
North-Holland, Amsterdam, 1977.


\bibitem{Mes15}  S. Mesnager, Bent Functions: Fundamentals and Results, Springer Verlag, Switzerland, 2016.



\bibitem{Rot76} O. S. Rothaus, ``On "bent" functions," J. Comb. Theory Ser.  A, vol. 20, no. 3, pp. 300--305, 1976.


\bibitem{Sch60} J. Schwinger, ``Unitary operator bases," Proc. Nat. Acad. Sci.,
    vol. 46, no. 4, 570--579, 1960.



\bibitem{Sid71} V. M. Sidelnikov, ``On mutual correlation of sequences," Soviet Math Doklady, vol. 12, no.1, pp. 197--201, 1971.



\bibitem{Sol89} P. Sol\'e, ``A quaternary cyclic code and a family of quatriphase sequences with low correlation properties,"
In: Coding Theory and Applications, Lecture Notes in Computer Science. New York:  Springer-Verlag, vol. 388, pp. 193--201, 1989.

\bibitem{SP80} D. V. Sarwate and M. B. Pursley, ``Cross-correlation properties of pseudorandom and related sequences," Proc. IEEE, vol. 68, no. 5, pp. 593--619, 1980.



\bibitem{TUF07} X. Tang, P. Udaya, and P. Fan,  ``Generalized binary Udaya-Siddiqi sequences,"  IEEE Trans. Inf. Theory, vol. 53, no. 3, pp. 1225--1230, 2007.


\bibitem{THJ08} X. Tang, T. Helleseth, and A. Johansen, ``On the correlation distribution of Kerdock sequences," In: International Conference on Sequences and Their Applications, Springer, Berlin, Heidelberg, vol. 5203, pp. 121--129, 2008.

\bibitem{Tonchev}
V. D. Tonchev, "Codes and designs," In: Pless, V. S., Huffman, W. C. (Eds.):
Handbook of Coding Theory, Vol. II., pp. 1229--1268, Elsevier, Amsterdam, 1998.

\bibitem{US98} P. Udaya and M. U. Siddiqi,  ``Optimal and suboptimal quadriphase sequences derived from maximum length sequences over $\mathbb Z_4$,"
 J. Appl. Alg. Eng. Commun., vol. 9, no. 2, pp. 161--191, 1998.

\bibitem{Wel74} L. Welch, ``Lower bounds on the maximum cross correlation of signals
(Corresp.)," IEEE Trans. Inf. Theory, vol. 20, no. 3, pp. 397--399,  1974.

\bibitem{WF89} W. K. Wootters and B. D. Fields,  ``Optimal state-determination by mutually unbiased measurements," Ann. Physics, vol. 191, no. 2, pp. 363-381, 1989.

\bibitem{WL13} B. Wu and Z. Liu, ``Linearized polynomials over finite fields revisited," Finite Fields and Their Applications, vol. 22, pp. 79--100, 2013.

\bibitem{XiangDingMesnager2015}
C. Xiang, C. Ding,  and S. Mesnager, ``Optimal codebooks from binary codes meeting the Levenstein bound," IEEE Trans. Inf. Theory, vol. 61, no. 12, pp. 6526--6535, 2015.

\bibitem{ZDL14} Z. Zhou, C. Ding, and N. Li, ``New families of codebooks achieving the Levenshtein bound,"
IEEE Trans. Inf. Theory, vol. 60, no. 11, pp. 7382--7387,  2014.

\end{thebibliography}
\end{document}